\documentclass[11pt]{article} 
\usepackage[utf8]{inputenc} 

\usepackage[margin=1in]{geometry} 
\usepackage{setspace}
\usepackage{graphicx} 
\usepackage[parfill]{parskip} 

\usepackage{parskip}
\usepackage{appendix}
\usepackage{booktabs} 
\usepackage{array} 
\usepackage{booktabs}
\usepackage[colorlinks=true, allcolors=blue]{hyperref}
\usepackage{multirow}
\usepackage{paralist}
\usepackage{verbatim}  
\usepackage{titling}
\usepackage{enumitem}
\usepackage{amsmath}
\usepackage{amsthm}
\usepackage{amsfonts}
\usepackage{amssymb}
\usepackage{mathtools}
\usepackage{hhline}
\usepackage{breqn}
\usepackage{accents}
\usepackage{subcaption}
\usepackage{xfrac}
\usepackage{bbm}
\usepackage{physics}
\usepackage{csvsimple}
\usepackage{hyperref}
\usepackage{multirow}
\usepackage{float}
\usepackage{marvosym}
\usepackage[round]{natbib}
\usepackage{pdflscape}
\usepackage[scr=boondox]{mathalfa}
\usepackage{authblk}
\usepackage{xr}
\usepackage{setspace} 

\usepackage{xcolor}

\usepackage{sectsty}
\allsectionsfont{\sffamily\mdseries\upshape} 
\newcommand{\subtitle}[1]{%
  \posttitle{%
    \par\end{center}
    \begin{center}\large#1\end{center}
    \vskip0.5em}%
}

\usepackage[nottoc,notlof,notlot]{tocbibind} 
\usepackage[titles,subfigure]{tocloft} 

\newtheorem{lemma}{Lemma}
\newtheorem{theorem}{Theorem}
\newtheorem*{theorem*}{Theorem}
\newtheorem{corollary}{Corollary}

\theoremstyle{definition}
\newtheorem{definition}{Definition}
\newtheorem{remark}{Remark}

\newtheorem*{remark*}{Remark}
\newtheorem{example}{Example}
\newtheorem*{example*}{Example}
\newtheorem{assumption}{Assumption}

\newtheorem{condition2}{Condition}[section]

\newtheorem{conditionalt}{Condition}[assumption]

\newenvironment{conditionp}[1]{
  \renewcommand\theconditionalt{#1}
  \conditionalt
}{\endconditionalt}

\makeatletter
\def\thm@space@setup{%
  \thm@preskip=\parskip \thm@postskip=0pt
}
\makeatother

\DeclareMathOperator{\Var}{Var}

\newcommand{\R}{\mathbb{R}}

\newcommand{\E}{\mathbb{E}}

\title{Generalized Bayes in Conditional Moment Restriction Models%
  \thanks{Email address: \texttt{sid.kankanala@chicagobooth.edu}. 
  I thank Xiaohong Chen, Yuichi Kitamura, and Donald Andrews for their guidance on this project. I am also grateful to Stéphane Bonhomme, Victor Chernozhukov, Tim Christensen, Ben Deaner, Chris Hansen, Han Hong, Tetsuya Kaji, Áureo de Paula, Veronika Ro\v{c}kov\'{a}, Andres Santos, and Anna Simoni for valuable comments and suggestions. This paper is based on  and supersedes an older preprint 
  \href{https://arxiv.org/abs/2311.00662}{arXiv:2311.00662} 
  \citep{kankanala2023gaussian}.}\\[3ex]}  

\author{Sid Kankanala}
\affil{University of Chicago}
\date{\today}

\begin{document}
\maketitle

\begin{center}
\large

\end{center}

\bigskip

\begin{abstract}

	\noindent
	\normalsize
	This paper develops a generalized Bayes framework for conditional moment restriction models, where the parameter of interest is a nonparametric structural function of endogenous variables. We establish contraction rates for a class of Gaussian process priors and provide conditions under which a Bernstein-von Mises theorem holds for the quasi-Bayes posterior. Consequently, we show that optimally weighted quasi-Bayes credible sets achieve exact asymptotic frequentist coverage, extending classical results for parametric GMM models. As an application, we estimate firm-level production functions using Chilean plant-level data. Simulations illustrate the favorable performance of generalized Bayes estimators relative to common alternatives.

\vspace{0.5cm}

\end{abstract}

\newpage

\section{Introduction}
Conditional moment restrictions are widely used to identify structural parameters in complex economic models. In many applications, the object of interest is an unknown nonparametric structural function $h_0(\cdot)$ that satisfies
\begin{align*} 
\mathbb{E} \!\left[ \rho\big(Y, h_0(X)\big) \,\middle|\, W \right] = \mathbf{0}\;,
\end{align*}
where $Y \in \mathbb{R}^{d_y}$ is a vector of outcomes, $X \in \mathbb{R}^d$ is a vector of endogenous regressors, $W \in \mathbb{R}^{d_w}$ is a vector of conditioning (or instrumental) variables, and the conditional distribution of $(Y,X) \mid W$ is left unrestricted. Here, $\rho(.) = [ \rho_{1}(.) , \dots , \rho_{d_{\rho}}(.)   ] $ is a $d_{\rho}$ dimensional vector of generalized residual functions, whose functional forms are assumed to be fully known. Common applications of this framework include consumer demand \citep*{blundell2007semi}, firm productivity (\citealp{doraszelski2013r}), differentiated product markets \citep{berry2024nonparametric}, production functions (\citealp*{ackerberg2015identification}), international trade \citep*{adao2017nonparametric}, treatment effects \citep{chernozhukov2005iv} and asset pricing \citep{bansal1993no,chen2009land}.

A common challenge for practitioners is that, although these restrictions are informative in the population, their finite-sample information content can be quite limited. In parametric models, this issue is typically attributed to weak instruments \citep*{stock2002survey}, whereas in nonparametric endogenous settings it reflects an “ill-posed inverse” problem \citep*{chen2012estimation}. As a result, classical nonparametric estimators often display undesirable properties such as high finite-sample variability, irregular behavior, and extreme sensitivity to small data perturbations. These difficulties are particularly evident in applications with multivariate endogenous regressors or when closed-form solutions are unavailable.

Motivated by these concerns, this paper proposes a class of nonparametric estimators and confidence sets obtained as solutions to generalized (quasi-) Bayes decision rules. In this framework, the conditional restrictions are interpreted as a quasi-likelihood which, when combined with a prior, yields a generalized Bayesian nonlinear inverse problem for the structural parameter. To fix ideas, let $\widehat{m}(\cdot)$ denote a feasible first-stage estimator of $m(W,h) = \E[\rho(Y,h)\mid W]$, $\widehat{\Sigma}(\cdot)$ a positive semi-definite weighting matrix, and $d\mu(\cdot)$ a prior on structural functions. We then study the generalized Bayes posterior distribution:   \[ \mu(\cdot \mid \mathcal{D}_n) \;=\; \frac{\exp\!\big(-\frac{n}{2}\,\E_n\!\big[\widehat{m}(W,\cdot)^{\prime}\,\widehat{\Sigma}(W)\,\widehat{m}(W,\cdot)\big]\big)\, d\mu(\cdot)}
{\int \exp\!\big(-\frac{n}{2}\,\E_n\!\big[\widehat{m}(W,h)^{\prime}\,\widehat{\Sigma}(W)\,\widehat{m}(W,h)\big]\big)\, d\mu(h)} \, . \]
 In the nonparametric endogenous models considered here, this framework provides a powerful form of data-driven regularization. Importantly, it also allows researchers to incorporate auxiliary information that strengthens the finite sample information content of the moments. Such information may range from weakly informative features, such as smoothness, to restrictions informed by application-specific microfoundations.

Over the past two decades, parametric quasi-Bayes procedures have found a variety of applications in econometrics, from models with nonsmooth objectives \citep{chernozhukov2005iv} to settings with nonstandard identification (\citealp*{chen2018monte}; \citealp{andrews2022optimal}). Most of the literature has focused on the properties of quasi-posteriors in parametric models. By contrast, relatively little is known about the behavior of quasi-Bayes in settings with a nonparametric structural parameter. This article helps bridge that gap by providing a unified treatment of quasi-Bayes for the broad class of nonparametric conditional moment restriction models commonly encountered in applied work. As we illustrate, when paired with a suitable nonparametric prior, quasi-Bayes naturally functions as a powerful form of data-driven regularization in endogenous models.

The main theoretical contributions of this paper are as follows. First, we introduce a theoretically motivated class of Gaussian process priors to model the nonparametric structural parameter. Together with the conditional restrictions, this induces a generalized (quasi-) Bayes posterior for the parameter. Second, we derive posterior contraction rates for the quasi-Bayes posterior in  classical $L^2$ metrics. Third, we establish conditions under which a nonparametric Bernstein–von Mises (BvM) theorem holds for the quasi-Bayes posterior. We use this to provide frequentist guarantees for certain optimally weighted quasi-Bayes credible
sets that are centered around the posterior mean. In particular, we show that such credible sets achieve asymptotically exact frequentist coverage. This provides the first nonparametric quasi-Bayes inferential guarantee in the literature, extending classical results (e.g. \citealp{chernozhukov2003mcmc})  for parametric GMM models.

We demonstrate the viability of our procedures across a broad class of models, including classical linear nonparametric IV, conditional quantile restrictions, and general nonlinear conditional restrictions. We complement this with extensive simulation evidence, replicating all univariate benchmark designs from the literature and extending them to settings with multivariate endogenous regressors. To highlight the flexibility of our approach, we additionally estimate models under alternative sets of restrictions whenever such alternatives are available. Overall, we expect our generalized Bayes procedures and accompanying implementation toolkit to be broadly useful for nonlinear conditional moment restrictions, particularly in ill-posed problems or when closed-form solutions are unavailable.

The paper is organized as follows. Section \ref{sec3} introduces the class of conditional moment restriction models and develops the generalized (quasi-) Bayes framework. Section \ref{motiv-sim} discusses our motivation for generalized Bayes procedures and relates it to the broader econometric literature. Section \ref{sec4} presents the assumptions and develops the main results. Sections \ref{motiv-sim} and \ref{simulations} provide simulation evidence on the performance of generalized Bayes estimators relative to common alternatives. In Section \ref{applic-prod}, we apply our methodology to nonlinear restrictions that arise in the nonparametric estimation of production functions. Section~\ref{conclus} provides additional remarks and concludes. Appendices \ref{append-sims}, \ref{implem-append}, \ref{gentheory}, and \ref{proofs} provide additional details on simulations, implementation, theory, and proofs, respectively.

\subsection{Literature} \label{lit-review}
There is a large literature on nonparametric sieve-based frequentist estimation and inference for conditional moment restriction models. As part of our general analysis, we review a subset of this literature in Sections~\ref{sec3}–\ref{sec4}. For a more comprehensive survey, particularly on early contributions, see \citet{chen2016methods}.

 In econometrics, our work is most closely related to \citet{chen2012estimation, chen2015sieve}, who developed the foundational frequentist sieve-based analysis of general conditional moment restriction models. At a high level, our procedures provide a generalized Bayes counterpart to their theory for infinite-dimensional sieves. However, instead of relying on traditional sieves and penalization, we develop procedures that are built around a class of infinite dimensional Gaussian process priors.

\citet{chernozhukov2003mcmc} developed the quasi-Bayes limit theory for parametric models strongly identified by a collection of moments. For finite-dimensional structural parameters, several alternative approaches have been proposed, including exponentially tilted empirical likelihoods \citep*{schennach2005bayesian,chib2018bayesian, chib2022bayesian} and methods that project a posterior on the data-generating distribution onto the parameter of interest \citep{chamberlain2003nonparametric, walker2024semiparametric}. By contrast, our focus is on endogenous models in which the parameters of interest are nonparametric structural functions. Importantly, in this setting, the structural parameter is infinite-dimensional, and its recovery is a challenging statistical ill-posed inverse problem.

In the statistical literature, early extensions of \citet{chernozhukov2003mcmc} to nonparametric models focused on slowly growing uninformative flat sieve priors. This line of work includes conditions for basic consistency \citep{liao2011posterior} and convergence rates in the special case of linear nonparametric IV models \citep{kato2013quasi}. These approaches parallel classical frequentist analysis (e.g. \citealp{ai2003efficient,newey2003instrumental}), where regularization is achieved by restricting estimation to a sequence of slowly expanding sieve spaces. By contrast, we study generalized Bayes procedures with infinite dimensional Gaussian process priors and develop statistical guarantees for general nonlinear conditional moment restrictions.

As we illustrate in Sections \ref{motiv-sim} and \ref{simulations}, the regularizing properties of the Gaussian process priors we study make them particularly well-suited to nonparametric endogenous models identified via general conditional moment restrictions. This motivation connects to early econometric work on the consistency of Gaussian priors in conjugate linear models with a known operator \citep*{florens2012nonparametric}.\footnote{For related work in statistics, see also \cite{knapik2011bayesian}, \cite{gugushvili2020bayesian}.} Our setting allows for general nonlinear and possibly nonsmooth restrictions with an unknown operator, leading to a non-conjugate quasi-Bayes posterior based on an estimated first-stage likelihood. Addressing this general case is necessary to cover the wide range of conditional moment restrictions commonly encountered in applied work, and our analysis develops both estimation and inferential guarantees in this setting.

Finally, in the special case of regression with exogenous covariates, our procedures relate to a growing literature in applied mathematics that examines Gaussian priors for nonlinear regression models with homoscedastic Gaussian noise \citep{dashti2015bayesian,monard2021statistical, nickl2023bayesian}. Our framework can be seen as complementary to this line of work, providing a generalized Bayes analogue that accomodates certain forms of heteroskedasticity and non-Gaussianity.
\section{Models and Procedures} \label{sec3}
Let $(Y,X,W)$ denote random vectors, where $Y \in \mathbb{R}^{d_y}$ is the outcome, $X \in \mathbb{R}^{d}$ the regressors, and $W \in \mathbb{R}^{d_w}$ the conditioning (instrumental) variables. We are interested in an unknown structural function $h_0$ that satisfies the conditional moment restriction
\begin{equation}\label{cmr}
\mathbb{E}\!\left[\rho\big(Y, h_0(X)\big)\,\middle|\, W\right] = \mathbf{0}.
\end{equation}

Here,  $\rho(.) = [ \rho_{1}(.) , \dots , \rho_{d_{\rho}}(.)   ] $ is a $d_{\rho}$ dimensional vector of generalized residual functions, whose functional forms are assumed to be fully known. Components of $X$ that are exogenous may, without loss of generality, be included in $W$. As is standard in applications, the conditional distribution of $(Y,X)$ given $W$ is not assumed to be known.

This framework is very general. By varying the choice of $\rho(\cdot)$, we can recover a large class of structural models commonly encountered in applied work. The form of the conditional restrictions, or equivalently the choice of generalized residual $\rho(\cdot)$, typically varies significantly across applications. The following examples illustrate some of these restrictions in further detail.
\begin{example}[Nonparametric Instrumental Variables] \label{ex1}  
The observed data consist of a scalar outcome variable $Y$, a vector of endogenous regressors $X$, and a vector of instrumental variables $W$. The structural function $h_0(\cdot)$ is identified by the conditional moment restriction:
$$
    \E[Y - h_0(X) \mid W] = 0.
$$
The generalized residual is $\rho(Y, h(X)) = Y - h(X)$. This model has been studied extensively in econometrics (e.g. \citealp{ai2003efficient}; \citealp{newey2003instrumental}; \citealp{hall2005nonparametric}; \citealp{darolles2011nonparametric}). As a special case, when the regressors are exogenous ($W = X$), the structural function is the conditional mean $h_0(X) = \E[Y \mid X]$. Generalizations of the classical NPIV restriction arise in a wide variety of settings, such as experimental price variation \citep{bergquist2020competition}, international trade \citep*{adao2017nonparametric}, and differentiated product markets (\citealp{compiani2022market}; \citealp{berry2024nonparametric}).
\end{example}

\begin{example}[Nonparametric Quantile IV] \label{ex2}
The observed data is as in Example~\ref{ex1}. Following \citet{chernozhukov2005iv,horowitz2007nonparametric,chen2012estimation}, fix a quantile $\tau \in (0,1)$, and consider the structural function $h_0(\cdot)$ that satisfies the restriction
\begin{align*}
    \mathbb{P} \big( Y - h_0(X) \leq 0 \mid W \big) - \tau = 0.
\end{align*}
The generalized residual function is $\rho_{\tau}(Y, h(X)) = \mathbbm{1} \{ Y - h(X) \leq 0 \} - \tau$. In this setting, we interpret $h_0(X)$ as a quantile structural effect. As discussed in \citet*{chernozhukov2007instrumental, chen2014local}, conditional quantile restrictions can also be used to estimate a large class of structural models with nonseparable disturbances.
\end{example}

\begin{example}[Production functions]  \label{ex3}  Following \cite*{levinsohn2003estimating,ackerberg2015identification}, consider the value-added output model
\[
y_{it} = F(x_{it}) + \omega_{it} + \varepsilon_{it},
\]
where $F(x_{it})$ is a production function for inputs $x_{it} \in \mathbb{R}^d$ (e.g., capital and labor), $\varepsilon_{it}$ represents shocks to production that are unobserved by the firm, and $\omega_{it}$ denotes shocks that are observed (or predictable) before the firm’s input decisions at time $t$. Assume $\omega_{it}$ is first-order Markov with conditional mean $\E[\omega_{it}\mid \omega_{i,t-1}] = g(\omega_{i,t-1})$. Let $m_{it}$ denote an intermediate input (e.g., electricity, fuel), and define $\Phi_t(x_{it},m_{it}) = \E[y_{it}\mid x_{it},m_{it}]$. If $\mathcal{I}_{t}$ denotes the firm’s information set at time $t$, \cite*{ackerberg2015identification} show that $h_0 = F(\cdot)$ satisfies the conditional restriction
\begin{equation}\label{eq:va-cmr}
\E\!\left[\, y_{it} - F(x_{it}) - g\!\big( \Phi_{t-1}(x_{i,t-1},m_{i,t-1}) - F(x_{i,t-1}) \big) \,\middle|\, \mathcal{I}_{t-1} \right] = 0.
\end{equation}
Similar nonlinear restrictions arise in a variety other settings, such as models of firm productivity \citep*{doraszelski2013r,boler2015r} and dynamic panel data \citep{blundell2000gmm}.
\end{example}

For intuition and as a guide to our general analysis, we will frequently refer to Examples \ref{ex1} and \ref{ex2}. We view these two examples as useful benchmark models for the following reason. In Example \ref{ex1}, the residual 
\( \rho(.) \) is a smooth linear function of $h$, whereas in Example \ref{ex2}, it is highly nonlinear and nonsmooth in $h$. In particular, they exemplify two distinct classes of models, distinguished by the regularity of the residual function. Although the restrictions encountered in empirical applications often appear more complex, their analysis and limiting structure can typically be characterized between these two extremes.
\subsection{Framework}
Given a function \( h(X) \), we denote the conditional mean of the generalized residual by $$
m(W,h) = \E \big[ \rho(Y,h(X)) \mid W \big].$$ The restriction \(m(W,h_0)=\mathbf{0}\) implies that \(h_0\) is the minimizer of the population criterion
\[
Q(h)=\mathbb{E}\!\left[m(W,h)^{\prime}\,\Sigma(W)\,m(W,h)\right],
\]
where \(\Sigma(W)\in\mathbb{R}^{d_{\rho}\times d_{\rho}}\) is a positive-definite weighting matrix.

As the distributional structure of the data is not assumed to be known, working with \( Q(h) \) directly is infeasible. The standard approach (e.g. \citealp{ai2003efficient,newey2003instrumental,chen2012estimation}) replaces \( m(W,h) \) and $\Sigma (\cdot)$ with suitable empirical analogs. Specifically, let \( \widehat{m}(W,h) \) and \( \widehat{\Sigma}(W) \) be  ``first-stage'' estimators of \( m(W,h) \) and \( \Sigma(W) \), respectively. Then, a feasible finite-sample objective function is  
\begin{align}
    \label{cmr-obj-feasible} 
    Q_n(h) = \E_n \big[ \widehat{m}(W,h)' \widehat{\Sigma}(W) \widehat{m}(W,h) \big].
\end{align}
The classical approach to estimating \( h_0 \) involves a ``second stage'', where \( Q_n(\cdot) \) is minimized over a suitable parameter space $\mathcal{H}_n$ to obtain an estimator \( \widehat{h} \). As noted in the literature (e.g.  \citealp*{chetverikov2017nonparametric}), these solutions often exhibit substantial finite-sample variability and are highly sensitive to small perturbations in the data and user-selected tuning parameters such as the complexity of $\mathcal{H}_n$. Intuitively, the second stage is ``ill-posed'' and the large finite-sample variability of these estimators arises from their representation as the inverse of an ill-posed objective.

To stabilize the inverse problem and more efficiently utilize the information content in the conditional moments, we examines a class of nonparametric estimators that arise as solutions to generalized Bayes decision rules. Specifically, we view the conditional moment restriction as a nonlinear inverse problem for the infinite dimensional structural parameter $h_0$. The restriction $ m(W,h_0) = \mathbf{0} $ then motivates a quasi-Bayes likelihood of the form \begin{align} \label{quasi-lik}
    L^{}(h) = \exp \bigg( - \frac{n}{2}  \E_n \big[    \widehat{m}(W,h) ' \widehat{\Sigma}(W)  \widehat{m}(W,h)  \bigg).
\end{align}
Denote the observed data by  $\mathcal{D}_n = \{ (X_1,Y_1,W_1), \dots , (X_n,Y_n,W_n) \}$. By combining the likelihood $L(.)$ with a (possibly data dependent) prior $\mu$ over structural functions, we obtain the generalized (quasi-) Bayes posterior: \begin{equation}
\label{qb-general}
\mu(\cdot \mid \mathcal{D}_n) \;=\; 
\frac{\exp\!\big(-\frac{n}{2}\,\E_n\!\big[\widehat{m}(W,\cdot)^{\prime}\,\widehat{\Sigma}(W)\,\widehat{m}(W,\cdot)\big]\big)\, d\mu(\cdot)}
{\int \exp\!\big(-\frac{n}{2}\,\E_n\!\big[\widehat{m}(W,h)^{\prime}\,\widehat{\Sigma}(W)\,\widehat{m}(W,h)\big]\big)\, d\mu(h)} \, .
\end{equation}
Related to this construction, \citet{liao2011posterior} transformed the conditional moment restrictions into a growing set of unconditional moments and proved the asymptotic consistency of a classical quasi-Bayes GMM criterion \citep{chernozhukov2003mcmc} under slowly growing flat sieve priors. In contrast, we follow the conventional frequentist approach, in which the first-stage functional $\widehat{m}(\cdot)$ is estimated directly, and we then treat the objective function $L(\cdot)$ in (\ref{quasi-lik}) as a quasi-likelihood for the model.

In this paper, we focus on a class of infinite dimensional Gaussian process priors for $d \mu (\cdot)$. When the structural function $h_0(\cdot)$ is defined over a bounded smooth domain $\mathcal{X} \subset \mathbb{R}^d$, a common choice is the family of Whittle--Matérn Gaussian process priors \citep{williams2006gaussian}. \begin{remark}[Weighting]  \label{opt-weight}
The weighting matrix $\widehat{\Sigma}(\cdot)$ may be deterministic or data dependent. For instance, analogous to two-step GMM, it may be constructed using a first step preliminary estimator of $h_0$. For estimation, a common choice is identity weighting $\widehat{\Sigma} = I_{d_{\rho}}$. We will refer to the quasi-Bayes posterior as \emph{optimally weighted} if $\widehat{\Sigma}(\cdot)$ is a consistent estimator of the efficient weighting matrix $ \Sigma_0(W) = \left\{ \E\!\left[ \rho(Y,h_0(X)) \rho(Y,h_0(X))' \,\big|\, W \right] \right\}^{-1}. $
\end{remark}
\subsection{Gaussian process priors} \label{gp-prior}
Gaussian process priors are widely employed in Bayesian nonlinear inverse problems, especially in applications arising within applied mathematics \citep{nickl2023bayesian}. To fix ideas, consider a mean-zero Gaussian process \( G \) with realizations in a  Hilbert space \( \mathcal{H} \) and covariance operator \( \Lambda \). By the spectral theorem, there exists an orthonormal basis of eigenfunctions \( (e_i)_{i=1}^{\infty} \subset \mathcal{H} \) that diagonalizes \( \Lambda \). If \( \lambda_i \) denotes the non-negative eigenvalue associated with \( e_i \), then \( G \) admits a unique Karhunen-Loève expansion of the form: \begin{equation}
    G \stackrel{d}{=} \sum_{i=1}^{\infty} \sqrt{\lambda_i}\, Z_i e_i, 
    \quad Z_i \stackrel{\text{i.i.d.}}{\sim} N(0,1).
\end{equation}
Intuitively, the rate at which \( \lambda_i \to 0 \) serves as a measure of the process's smoothness relative to the eigenbasis. If \( (e_i)_{i=1}^{\infty} \) denotes the standard Fourier basis, this corresponds to classical Sobolev smoothness. 

Similar to the analysis in \cite*{knapik2011bayesian}, we consider a family of Gaussian process priors $\{G_{\alpha} : \alpha \in \mathcal{L} \}$ that are indexed by a regularity hyperparameter $\alpha \in \mathcal{L} \subset \mathbb{R}_+$. In this setting, each process $G_{\alpha}$ admits an expansion of the form\footnote{If the mapping $\alpha \mapsto \lambda_{i,\alpha}$ influences the exponent in a different way, the results can also be stated in terms of the induced exponent $s(\alpha)$, i.e., $\lambda_{i,\alpha} \asymp i^{-s(\alpha)}$.} \begin{align}
    \label{gexpand2} 
    G_{\alpha} \stackrel{d}{=}  \sum_{i=1}^{\infty} \sqrt{\lambda_{i,\alpha}} Z_i e_i, \quad Z_i \stackrel{\text{i.i.d.}}{\sim} N(0,1).
\end{align} where $\lambda_{i,\alpha} \asymp i^{-(1 + 2\alpha/d)}$ and $(e_i)_{i=1}^{\infty}$ is an orthonormal basis of $L^2(\mathcal{X})$. 

While we do not impose any restrictions on the eigenbasis $(e_i)_{i=1}^{\infty}$ directly, we will typically require the sample paths of the Gaussian process $G_{\alpha}$ (for $\alpha \in \mathcal{L}$) to satisfy some minimum regularity (see Condition \ref{gp-structure} below). In most cases, this can be satisfied by restricting the regularity index set to $ \alpha \in  \mathcal{L}  \subseteq [\underline{\alpha}, \infty)$ for some minimum regularity $\underline{\alpha} > 0$. The following example illustrates the general idea for a widely used family of Gaussian process priors.
\begin{example*}[Mat\'ern Gaussian Priors]
If the structural function $h_0(.)$ is defined over a bounded smooth domain $\mathcal{X} \subset \R^d$, a popular choice  is the Whittle–Matérn Gaussian process $G_{\alpha}$, indexed by smoothness regularity $\alpha > 0$. This Gaussian process  has covariance kernel \begin{align} \label{covst}
    \Lambda_{\alpha}(s,t) = \int_{\R^d} e^{- \mathbf{i} \langle s-t , \zeta \rangle} (1 + \| \zeta \|_{\ell^2}^2)^{-(\alpha+d/2)} d \zeta \; \; \; \; \; \forall  \; s,t \in \mathcal{X}.
\end{align}
It is well known \citep[Proposition I.4]{ghosal2017fundamentals} that $G_{\alpha}$ has sample paths belonging almost surely to the Hölder spaces $C^{\beta}(\mathcal{X})$ for any $\beta < \alpha$, so that $G_{\alpha}$ can be viewed as an ``almost $\alpha$ smooth" process. Furthermore, the process $G_{\alpha}$ satisfies, for some $\kappa > 0$, the stochastic partial differential equation $$ \big( \kappa - \Delta  \big)^{ \frac{\alpha}{2} + \frac{d}{4}} G_{\alpha} =   \mathcal{Z} \:, $$
where $\Delta$ is the Laplacian operator and $\mathcal{Z}$ is Gaussian white noise. It follows that the covariance operator $\Lambda_{\alpha}$ of $G_{\alpha}$ diagonalizes in the same eigenbasis as the Laplacian. Since the eigenvalues $ (\kappa_i)_{i=1}^{\infty}  $ of the Laplacian scale as  $\kappa_i \asymp i ^{2/d}$, it follows that the eigenvalues $(\lambda_{i,\alpha})_{i=1}^{\infty}$ of  $\Lambda_{\alpha}$ scale at rate $ \lambda_{i,\alpha} \asymp i^{-(1 +2 \alpha/d)}$.
\end{example*}
Intuitively, larger values of $\alpha$ correspond to smoother sample paths. In certain applications, suitable smoothness levels can be motivated by prior studies or application-specific microfoundations. In settings where such guidance is unavailable, $\alpha = 3/2$ and $\alpha = 5/2$ are widely used as standard defaults \citep{williams2006gaussian}, offering a balance between regularity and flexibility to accommodate irregular variation.
\begin{remark}[Centering] \label{loc-remark} We focus on a mean-zero process for simplicity. In most settings, the data can be appropriately standardized for this location to be natural. For instance, in Example \ref{ex1} and \ref{ex2}, we have $\E[Y] = \E[h_0(X)]$, which motivates the use of a mean-zero process for the ``standardized model" that uses $ \widetilde{Y} = \big[Y - \E_n(Y)\big] \big(\widehat{Var}(Y)\big)^{-1/2}$.
\end{remark}

\begin{remark}[Scaling] \label{scale-remark}
It is also possible to define a new process by scaling and stretching an existing one. 
Specifically, if 
$
G = \{ G(x) : x \in \mathcal{X} \}
$ 
is a base process, we can define 
$$
G_{\theta}(x) = \sigma \, G(\ell^{-1} x),
$$
where the notation $\ell^{-1} x$ is interpreted coordinate-wise as
$
\ell^{-1}x = (\ell_1^{-1} x_1, \dots, \ell_d^{-1} x_d).
$
Here, $\theta = (\sigma, \ell)$, where $\sigma \in \mathbb{R}_{+}$ denotes the signal variance and $\ell \in \mathbb{R}_{+}^d$ the length-scale parameter. Intuitively, $\sigma$ controls the vertical scale of the process, while $\ell$ controls the rate at which correlations decay with distance. The theoretical properties for any fixed $\theta$ are similar to those of the base process. However, in practice, it is common to partially tune these hyperparameters using the observables. We discuss hyperparameter tuning in Section \ref{conclus} and Appendix \ref{implem-append}.
\end{remark}

\subsection{First stage estimation} \label{first-stage}
Researchers have considerable flexibility in the choice of the first-stage estimator for the conditional mean $ m(W,h) = \mathbb{E}[\rho(Y,h(X)) \mid W] $. This can accomodate a broad range of regression and machine learning methods. In practice, however, it will be convenient to focus on estimators that are computationally efficient, as this ensures that the quasi-likelihood $L ( \cdot) $ in (\ref{quasi-lik}) can be evaluated efficiently.

A common and efficient choice is to consider sieve-based first stages, defined as linear projections onto a set of basis functions. Let $b^K(W) =  [ b_1(W), \dots , b_K(W)   ]'$ denote a vector of first stage approximating functions. Then, for a given function $h(X)$, we estimate the conditional mean by the least squares projection:
\begin{align} \label{cmean}  &  \widehat{m}(w,h)  =      \E_n \big[   \rho(Y,h_{}(X)) \big(b^K(W) \big)'    \big]               [ \widehat{G}_{b,K} ]^{-1}   b^K(w)     \; , \\ &  \text{where} \; \; \;    \widehat{G}_{b,K} = \E_n \big[
 \big(b^K(W) \big) \big( b^K(W) \big)'    \big]      .        \nonumber
\end{align}
In low dimensions, approximating functions can be formed from tensor products of standard univariate bases (e.g. Fourier series, splines), eigenfunction expansions and indicator functions to accommodate discrete instruments. In higher dimensions, common alternatives are bases constructed using randomized features (e.g.  \citealp{rahimi2007random}).

To facilitate detailed analysis and clarity of exposition, we focus on a classical first stage defined by a linear projection onto approximating functions.\footnote{In Appendix \ref{gentheory}, we provide some theory for contraction with generic first-stage estimators.} Although our main results extend to other first-stage estimators, the conditions required to obtain statistical guarantees will generally depend on the specific choice of estimator. By concentrating on the sieve case, we keep the first-stage analysis self-contained and directly comparable to the classical frequentist analysis of conditional moment restriction models. 
  
In the classical frequentist literature (e.g., \citealp*{blundell2007semi, chen2012estimation}), 
the choice of first stage estimator is typically not viewed as a ``key tuning parameter.'' 
Intuitively, estimating the smooth conditional mean $
 \mathbb{E}[\rho(Y,h(X)) \mid W] $
is a well-posed regression problem and is far less sensitive to tuning than a classical ill-posed inverse problem. This is also true in our setting. Specifically, if $\Theta_n$ denotes a suitable collection of high probability regular sample paths of the Gaussian process, the first stage is best viewed as providing an efficient approximation to the conditional mean operator
$ \Theta_n \ni
h \mapsto \mathbb{E}[\rho(Y,h(X)) \mid W]
$.

\section{Motivation} \label{motiv-sim}
In this section, we discuss the econometric and practical motivation for quasi-Bayes procedures, with emphasis on their application to nonparametric endogenous models. We begin with the econometric motivation, particularly in comparison with fully Bayesian and classical frequentist approaches.

A fully Bayes approach to this problem would typically require explicit modeling of the conditional distribution $(Y,X)\mid W$. Since our primary object of interest is the structural parameter, this distribution is a complex nuisance, and modeling it may be undesirable in many settings. Analogous to the econometric motivation underlying classical GMM \citep{hansen1982generalized}, it is often preferable to target the structural parameter directly, particularly when the parameter itself is a complex nonparametric object.\footnote{For finite dimensional structural parameters, 
a similar point was made by \citet{chernozhukov2003mcmc}.}

Beyond modeling challenges, the analysis in \citet*{bornn2019moment,florens2021gaussian} also highlight that, even with parametric structural parameters, there are subtle probabilistic difficulties in specifying a joint prior on the nuisance law $F_{(Y,X)\mid W}$ and structural parameter.\footnote{Constructing a reasonable prior on the low dimensional manifold 
$\Theta=\{(h,F): \E_{F}[\rho(Y,h(X)) \mid W]=0\}$ is challenging: 
for any fixed $h$, classical priors typically assign probability zero to the fiber 
$\mathcal{F}_h=\{F:\E_{F}[\rho(Y,h(X)) \mid W]=0\}$. 
This difficulty arises even in simpler settings with unconditional moments and 
finite-dimensional structural parameters.} In our setting with an infinite dimensional structural function, this becomes considerably more challenging. Although it may be possible, in theory, to proceed without a prior on the structural function, this is ill-advised for the nonparametric endogenous models we study, as it forgoes the regularization, interpretability, and flexibility gained by placing the prior directly on the structural function.

\begin{remark}[Frequentist estimation]  \label{frequentist-app} Frequentist approaches (e.g. \citealp{ai2003efficient}; \citealp{newey2003instrumental}; \citealp{chen2012estimation})  have typically focused on the objective function in (\ref{cmr-obj-feasible}), which avoids the need to model the nuisance explicitly. Generalizing the intuition from classical GMM, these approaches exploit the fact that identification of $h_0$ depends on the nuisance only through the first stage functional $ h \mapsto \mathbb{E}[\rho(Y,h) \mid W] $, 
which can be accurately estimated using a wide range of off-the-shelf regression methods. Intuitively, for the purpose of estimating the structural function $h_0$, the first stage is an efficient ``sufficient functional statistic" for the nuisance.
\end{remark}

From the preceding discussion, it follows that quasi-Bayes can be viewed as a convenient hybrid between frequentist and fully Bayes methods. Similar to classical frequentist procedures, it utilizes the efficient first stage as a sufficient statistic for the nuisance. In the second stage, the difficult, ill-posed recovery of the structural function is formulated as a generalized Bayesian nonlinear inverse problem \citep{nickl2023bayesian}. In this setting, the prior on the structural function provides a powerful form of data driven regularization, while also allowing the researcher to incorporate domain-specific knowledge.

\subsection{Simulation Evidence} \label{sim-evid}
To illustrate some of our motivation in greater detail, we make use of all the benchmark designs previously employed in the nonparametric instrumental variable (NPIV) literature. Specifically, we consider the designs from \citet{newey2003instrumental}, \citet{santos2012inference}, \citet*{chernozhukov2015constrained}, \citet{chetverikov2017nonparametric}, and \citet*{chen2025adaptive}, which we refer to as \textbf{NP}, \textbf{S}, \textbf{CNS}, \textbf{CW} and \textbf{CCK}, respectively. In all of these designs, the regressor is univariate and the structural function is estimated under a nonparametric instrumental variable (NPIV) restriction (Example \ref{ex1}). Details on all the designs are contained in Appendix \ref{append-sims}.
 
Let $\mathcal{D}_n$ denote the observed data, and let $X'$ be an independent draw from the distribution of $X$. 
Given an estimator $\widehat{h} = \widehat{h}(\mathcal{D}_n)$, we define the \emph{expected out-of-sample root mean squared risk:} 
\[
\mathcal{R}(\widehat{h}, h_0)
= \left\{ \mathbb{E}_{\mathcal{D}_n,\,X'} \left[ \big( \widehat{h}(X') - h_0(X') \big)^2 \right] \right\}^{1/2}.
\]
Let \textbf{2SLS} denote the two-stage least squares estimator, where the first stage uses thin-plate splines and the structural function uses natural splines, both of dimension $J$.\footnote{Natural splines provide some regularization by enforcing $h''(x) = 0$ at the data boundary, implying linearity beyond. For larger $J$, results appeared more unstable with alternative bases.}
\begin{table}[!htbp]
\caption{Sample size: \( n = 1000 \). Risk \( \mathcal{R}(\widehat{h}, h_0) \) for NPIV 2SLS estimators.}
\label{table1}
\centering
\scalebox{1}{
\begin{tabular}{
  l
  @{\hspace{30pt}}c
  @{\hspace{15pt}}c
  @{\hspace{15pt}}c
  @{\hspace{15pt}}c
}
\toprule
Design
  & \multicolumn{4}{c}{\textbf{2SLS}} \\
\cmidrule(lr){2-5}
  & \multicolumn{1}{c}{$J=3$}
  & \multicolumn{1}{c}{$J=4$}
  & \multicolumn{1}{c}{$J=5$}
  & \multicolumn{1}{c}{$J=6$} \\
\midrule
\textbf{NP}   & 0.131 & 0.154 & 0.355 & 4.84  \\
\textbf{S}    & 0.292 & 7.30 & 37.52 & 132.11
   \\
\textbf{CNS}  & 0.189 & 11.77 & 34.83 & 74.35   \\
\textbf{CW}   & 1.623 & 8.20 & 34.19 & 113.37  \\
\textbf{CCK}  & 0.345 & 6.01 & 130.04 & 435.91 \\
\bottomrule
\end{tabular}}
\end{table}

As Table \ref{table1} illustrates, in endogenous models, classical estimators are highly sensitive to tuning parameters that determine the complexity of the parameter space. In some univariate settings (e.g. NPIV, \citealp*{chen2025adaptive}), this complexity can be tuned in a data driven way. However, in models with generalized nonlinear restrictions, multivariate regressors, or no closed-form solutions, effective tuning becomes substantially more challenging. Indeed, to the best of our knowledge, no regularization mechanism has yet been demonstrated to perform successfully across the broad range of models, restrictions and data generating processes encountered in theoretical and empirical work.

It is well known that Bayes procedures regularize naturally via the prior, albeit at the cost of potential finite-sample bias. In endogenous settings, the resulting variance reduction can be substantial. In nonparametric Bayes procedures, this bias typically takes the form of a preference for well-behaved or regular functions. We argue that this property is particularly valuable as a regularization mechanism in nonparametric endogenous models, where structural function regularity is typically already a prerequisite for any meaningful analysis. Indeed, this feature is evident in all the designs reported in Table \ref{table1} and all other designs considered in the broader literature.

To further illustrate the preceding point, consider all the designs in Table \ref{table1}. 
They can be estimated using either of the following generalized residuals:
\[
\begin{aligned}
(i) \quad & \rho(Y,h(X)) = Y - h(X) && \text{(NPIV)}, \\
(ii) \quad & \rho(Y,h(X)) = \mathbbm{1}\{ Y - h(X) \leq 0 \} - 0.5 && \text{(median NPQIV)}.
\end{aligned}
\]

In general, the NPQIV restriction is considered more challenging, as it involves a nonlinear and nonsmooth residual. Let \textbf{QB} denote the quasi-Bayes posterior mean, based on  a first-stage thin plate spline of dimension $K$ and a classical Whittle–Matérn Gaussian process prior. We use the same prior and implementation algorithm across all designs and both sets of restrictions. Further details are provided in Appendix \ref{implem-append}.

\begin{table}[!htbp]
\caption{Sample size $n=1000$. Risk $\mathcal{R}(\widehat{h},h_0)$ for \textbf{QB} estimators, based on 1000 replications.}
\label{table2}
\centering
\setlength{\tabcolsep}{9pt} 
\renewcommand{\arraystretch}{1.1}
\setlength{\arrayrulewidth}{0.6pt} 
\scalebox{1}{
\begin{tabular}{l @{\hspace{12pt}} *{3}{c} @{\hspace{18pt}} | @{\hspace{18pt}} *{3}{c}}
\toprule
Design
& \multicolumn{3}{c}{\textbf{QB} (NPIV)}
& \multicolumn{3}{c}{\textbf{QB} (NPQIV)} \\
\cmidrule(lr){2-4}
\cmidrule(lr){5-7}
& $K=5$ & $K=7$ & $K=10$
& $K=5$ & $K=7$ & $K=10$ \\
\midrule
\textbf{NP}  & 0.155 & 0.148 & 0.141 & 0.362 & 0.361 & 0.359 \\
\textbf{S}   & 0.232 & 0.210 & 0.197 & 0.608 & 0.608 & 0.609 \\
\textbf{CNS} & 0.138 & 0.134 & 0.134 & 0.105 & 0.100 & 0.105 \\ 
\textbf{CW}  & 0.126 & 0.122 & 0.118 & 0.176 & 0.173 & 0.173 \\
\textbf{CCK} & 0.285 & 0.276 & 0.266 & 0.330 & 0.326 & 0.329 \\
\bottomrule
\end{tabular}}
\end{table}

Table~\ref{table2} reports the quasi-Bayes risk for all designs in Table~\ref{table1}, under both NPIV and NPQIV restrictions. The estimates appear remarkably accurate and stable across both restrictions. A natural question is how far these findings extend. For example, can they generalize to more challenging settings with multivariate regressors? In Section~\ref{simulations}, we provide additional evidence by examining multivariate extensions of the designs in Table~\ref{table2}.

\begin{figure}
  \centering
  \includegraphics[width=0.9\textwidth]{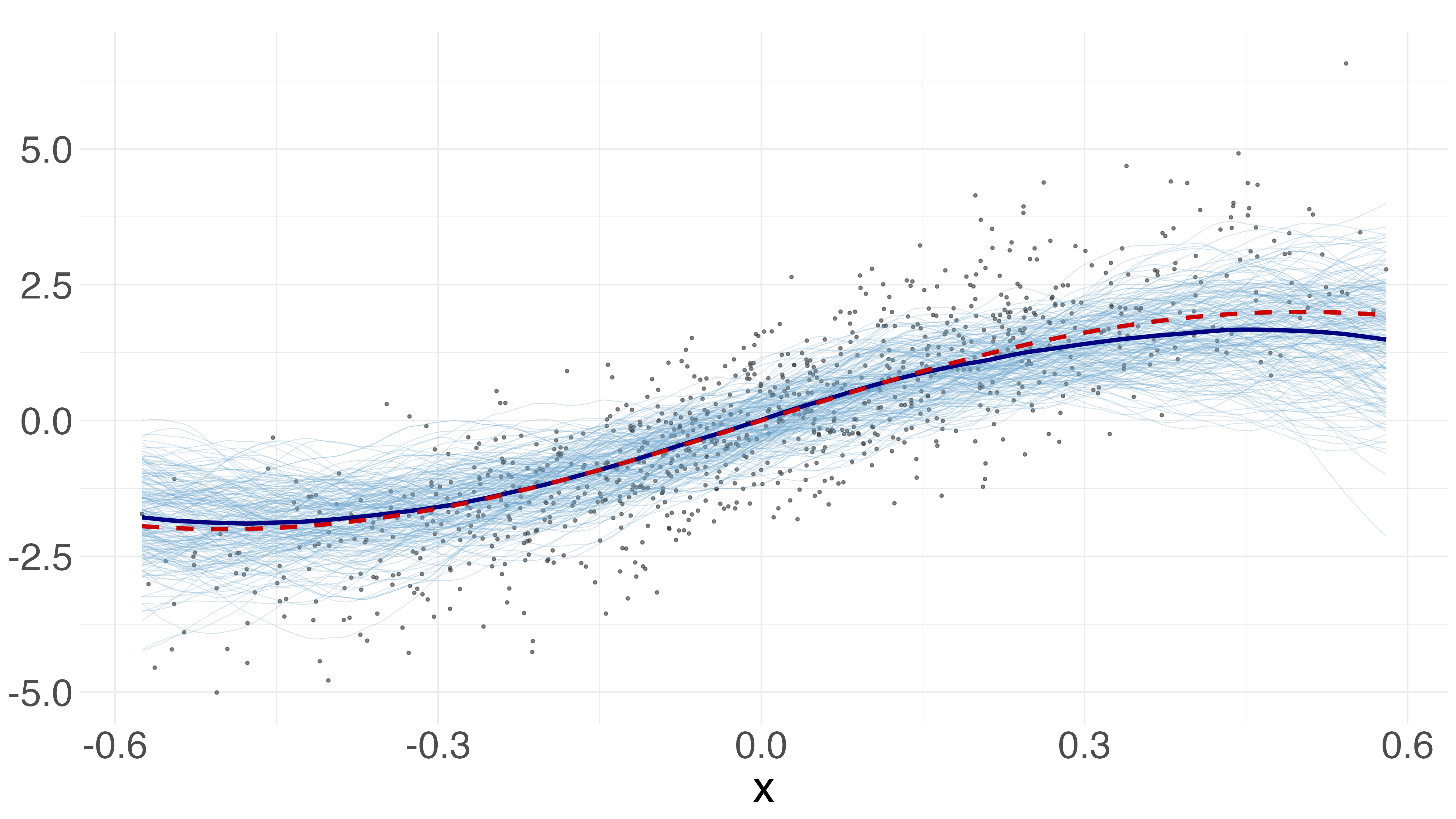}
  \caption{Sample size $n = 1000$. NPIV posterior for the design in \citet{santos2012inference}. 
  The red dashed line shows the true function, the dark blue solid line is the posterior mean, and the light blue lines are posterior draws.}
  \label{posmeanz}
\end{figure}

As a final remark, we note that these procedures differ from classical frequentist regularization in two key ways. First, as noted earlier, devising a broadly effective data-driven regularization scheme that works across all models and restrictions is highly challenging. By contrast, in our quasi-Bayes framework, the priors we employ induce a nontrivial form of regularization that has proven effective in a wide range of applications, particularly in nonlinear inverse problems.\footnote{See \cite{ghosal2017fundamentals,nickl2023bayesian} for an overview of applications.} Second, quasi-Bayes procedures are inherently data-driven through the interplay between the prior and the information in the conditional moments. This interaction is precisely what allows the information content in the moments to dominate in settings with strong identification and enables a single prior specification to yield reasonable results across all the designs and restrictions in Table \ref{table2}.

\section{Theory} \label{sec4}
In this section, we develop the limit theory for the generalized (quasi-) Bayes posterior in (\ref{qb-general}). Specifically, we examine the following questions in detail: (i) What are the minimal conditions on the model and prior that ensure quasi-Bayes consistency? (ii) How do convergence rates depend on the smoothness of the structural function \( h_0 \)? (iii) When do nonparametric quasi-Bayes credible sets achieve exact frequentist coverage?

\subsection{Assumptions on the Generalized Residual} \label{sect-genres}
To begin with, we state our main conditions on the generalized residual function $\rho(\cdot)$ that defines the conditional moment restriction in (\ref{cmr}). We assume that the endogenous regressor \(X\) is supported on a smooth bounded domain \(\mathcal X\subset\R^{d}\), and the instrument \(W\) is supported on a domain \(\mathcal W\subseteq\R^{d_w}\). This is standard in the literature and, if necessary, can always be satisfied by applying an appropriate transformation of the regressors.\footnote{In practice, apart from basic standardization, no transformations are used in our implementation.}

For any $t > 0$, let $(\mathbf{H}^t, \| \cdot \|_{\mathbf{H}^t})$ denote the usual Sobolev space of order $t$ over $\mathcal{X}$. The Sobolev ball of radius $M$ is denoted by $  \mathbf{H}^t(M) = \{ h : \| h \|_{\mathbf{H}^t} \leq M \}$.

\begin{condition2}[Local $L^2$ continuity] \label{residuals} For some  $ \kappa  \in (0,1]$, $ t > d/ 2\kappa  $ and any $M < \infty$, there exists $C_1 = C_1(M) < \infty$ such that
\begin{align*}
 & \sup_{w \in \mathcal{W}} \E_{} \bigg(  \sup_{h \in  \mathbf{H}^t(M) :  \| h' - h \|_{ \infty } \leq \xi}  \|   \rho(Y, h(X) ) - \rho(Y,h'(X))            \|_{\ell^2}^2 \big|W=w   \bigg) \leq C_1^2  \xi^{2 \kappa}  , \\ & \sup_{h \in  \mathbf{H}^t(M) :  \| h' - h \|_{ L^2(\mathbb{P}) } \leq \xi} \;  \sup_{w \in \mathcal{W}} \E \bigg(   \|   \rho(Y, h(X) ) - \rho(Y,h'(X))            \|_{\ell^2}^2 \big|W=w    \bigg) \leq C_1^2 \xi^{2 \kappa}
\end{align*}
holds for all $h' \in \mathbf{H}^{t}(M)$ and $\xi > 0$ small enough.
\end{condition2}
In Condition \ref{residuals}, the two expectations differ in the metrics they employ. The first expectation is over the the supremum with respect to the stronger $\| \cdot \|_{\infty}$ norm, whereas the outer supremum of the second expectation is taken under the weaker $ \|\cdot \|_{L^2(\mathbb{P})}$ norm. Intuitively, because the expected supremum is more difficult to control, it is taken over functions that are closer in a stronger metric. 

Condition \ref{residuals} is analogous to conditions that are frequently used in the analysis of non-smooth objectives \citep*{chen2003estimation}. In particular, it permits a pointwise discontinuous residual function (e.g. NPQIV models) provided that $\rho(\cdot)$ is suitably uniformly continuous in $L^{2}(\mathbb{P})$ expectation. The parameter $\kappa$ is typically referred to as the local continuity exponent.  It holds with $\kappa=1$ for the NPIV model (Example \ref{ex1}) and $\kappa=1/2$ for the NPQIV model (Example \ref{ex2}).
\begin{condition2}[Residual moments] \label{residuals2} 
There exists $\epsilon, \delta > 0 $ and $t > d/2\kappa$  such that for any $M > 0 $, there exists finite constants $C_2(M) , C_3(M), C_4(M) < \infty$ that satisfy   
\begin{align*} &  (i) \; \; \; \;  \sup_{w \in \mathcal{W}} \E \bigg(  \sup_{h \in \mathbf{H}^{t}(M)}  \|  \rho(Y,h(X))  \|_{\ell^2}^2 \big| W=w        \bigg) \leq C_2^2 \;  , \\ & (ii) \; \; \; \; \E \bigg(  \sup_{h \in \mathbf{H}^{t}(M)}  \|  \rho(Y,h(X))  \|_{\ell^2}^{2+ \epsilon}       \bigg) \leq C_3^2 \; , \\ & (iii)  \; \; \; \; \mathbb{P} \bigg(  \sup_{h,h' \in \mathbf{H}^{t}(M) : \|h - h' \|_{L^2(\mathbb{P})} \leq \delta  }  \|  \rho(Y,h(X)) - \rho(Y,h'(X))  \|_{\ell^2} \leq C_4 \bigg) = 1.
\end{align*}
\end{condition2}
Condition~\ref{residuals2} imposes mild moment restrictions on the residual function:  
the bounds only need to hold over any fixed Sobolev ball.  
The assumption is trivially satisfied with bounded residual functions (e.g. NPQIV).
More generally, if $t>d/2$, the Sobolev embedding theorem \citep{evans2022partial} implies that  
$\mathbf H^{t}$ embeds continuously into a Hölder space, so functions in $\mathbf H^{t}(M)$ are uniformly bounded in the $\|\cdot\|_{\infty}$ norm. In most settings, this observation makes it straightforward to verify Condition~\ref{residuals2}.  
For example, in the NPIV model, Condition~\ref{residuals2} holds if the unobserved error $u$ satisfies  
$\mathbb E\big(|u|^{2+\epsilon}\big)<\infty$ and $\mathbb E[u^{2}\big| W]\le\bar{\sigma}^{2}$ for some  $\bar{\sigma}^{2} < \infty$.

While the generalized residual may be non-smooth, we assume (as is standard) that its smoothed conditional mean
$m(W,h) = \E[\rho(Y,h(X)) \mid W]$ is sufficiently regular, in the sense that it satisfies a local Lipschitz property. This is formalized below in Condition~\ref{mean-lip}.
\begin{condition2}[Locally Lipschitz conditional mean]  \label{mean-lip}
 For some $t > d/(2\kappa)$, the map $h \mapsto m(W,h)$ from $(\mathbf{H}^t, \| \cdot \|_{L^2(\mathcal{X})})$ to $(L^2(W),\| \cdot \|_{L^2(\mathbb{P})})$ is continuous. Furthermore, for every $M > 0 $, there exists a constant $C_5(M) < \infty$ such that $\| m(W,h) - m(W,h_0) \|_{L^2(\mathbb{P})} \leq C_5  \| h - h_0 \|_{L^2(\mathcal{X})}$ for every $h \in \mathbf{H}^t(M)$.
\end{condition2}

\subsection{Consistency} \label{sect-cons}
In this section, we establish the consistency of general quasi-Bayes posteriors arising from suitably rescaled Gaussian process priors. As discussed in Section~\ref{first-stage}, we consider a classical first stage based on projecting onto a set of basis (approximating) functions $
b^K(W) = \big[b_1(W), \dots, b_K(W)\big].$ Denote by $\Pi_K(\cdot)$, the $L^2(\mathbb{P})$ projection operator onto the span of these functions.

Following the discussion in Section~\ref{gp-prior}, let $G_{\alpha}$ denote a mean-zero Gaussian process with regularity parameter $\alpha > 0$. Let $(e_i)_{i=1}^{\infty}$ be the orthonormal eigenfunction basis of its covariance operator $\Lambda_{\alpha}$. Similar to the analysis in \citet*{knapik2011bayesian}, it will be convenient to measure regularity directly with respect to this basis.\footnote{When $G_{\alpha}$ is a Whittle--Matérn Gaussian process, or when $(e_i)_{i=1}^{\infty}$ is a standard Fourier basis, this reduces to classical Sobolev regularity.} To that end, for any $p > 0$, we define the associated $p$-regularity class as
\begin{align} \label{sob} 
\mathcal{H}^p 
  = \Bigg\{ h \in L^2(\mathcal{X}) : 
      h = \sum_{i=1}^{\infty} c_i e_i \; , \; 
      \| h \|_{\mathcal{H}^p}^2 
        = \sum_{i=1}^{\infty} i^{2p/d} c_i^2 < \infty 
    \Bigg\}.
\end{align}

Given  $G_{\alpha}$  and first stage sieve dimension $K$, we consider the rescaled prior: \begin{align} \label{prior-scale} d\mu(.) \sim  \frac{G_{\alpha}}{\sqrt{K}}. \end{align}
Rescaled Gaussian process priors are frequently employed in the analysis of Bayesian nonlinear inverse problems \citep{monard2021statistical,nickl2024posterior,nickl2025bayesian}. In our conditional moment restriction framework, the scaling provides additional regularization that is crucial both for $(i)$ controlling the nonlinear ill-posedness of the inverse problem and $(ii)$ obtaining high-probability guarantees on the behavior of the first-stage estimator $\widehat{m}(W,h)$ used to approximate the conditional mean  $h \mapsto m(W,h)$.

 Intuitively, the posterior limit theory is determined by the interplay between the prior and the quasi-Bayes likelihood  $
h \mapsto \E_n\!\big[ \widehat{m}(W,h)' \, \widehat{\Sigma}(W) \, \widehat{m}(W,h) \big] $. 
To formalize this interplay, we impose low-level conditions on three components: the prior, the weighting matrix $\widehat{\Sigma}(\cdot)$, and the first-stage basis functions $\{b_1(W), \dots, b_K(W)\}$ used to construct the conditional mean estimator $\widehat{m}(W,h)$. Our main requirements on these objects are summarized in the following two conditions.

 \begin{condition2}[Regularity] \label{gp-structure} $(i)$ The density of $X$ with respect to the Lebesgue measure is bounded away from $0$ and $\infty$ on $\mathcal{X}$. $(ii)$ $G_{\alpha}$ is a Gaussian random element on a separable subspace of the Sobolev space $\mathbf{H}^t$ for some $t > d/2\kappa$.
\end{condition2}
Condition~\ref{gp-structure}$(i)$ is imposed for convenience, as it ensures the equivalence of the norms $\| \cdot \|_{L^2(\mathbb{P})}$ and $\| \cdot \|_{L^2(\mathcal{X})}$, where the latter is taken with respect to the Lebesgue measure. Condition \ref{gp-structure}$(ii)$ can be interpreted as a minimum regularity requirement in that it ensures the Gaussian process $G_{\alpha}$ has continuous and bounded sample paths.\footnote{This is a consequence of the Sobolev inequality \citep{evans2022partial}, since $\mathbf{H}^t$ (for $t > d/2$) embeds into a H\"{o}lder space $C^{\beta}$ for some $\beta > 0$.} 

\begin{condition2}[First stage approximation] \label{fsbasis} $(i)$ The matrix  $G_{b,K} = \E \big( [b^K(W)] [ b^K(W) ]'    \big)$ is positive definite for all $K$ and $\zeta_{b,K} =  \sup_{w \in \mathcal{W} }  \|  G_{b,K}^{-1/2} b^K(w) \|_{\ell^2}  \lessapprox \sqrt{K}$. $(ii)$ The eigenvalues of $\widehat{\Sigma}(W)$ are asymptotically bounded above and below: $\mathbb{P}\big(  c \leq   \lambda_{\min}(  \widehat{\Sigma}(W) ) \leq    \lambda_{\max}(  \widehat{\Sigma}(W) )   \leq C   \big)  \rightarrow 1$ for some $0 < c \leq C < \infty$. $(iii)$ For any fixed $ M > 0 $, the first stage is uniformly consistent over the Sobolev ball $\mathbf{H}^t(M)$: $ \sup_{h \in \mathbf{H}^t(M)  } \|(\Pi_K - I) m(W,h)   \|_{L^2(\mathbb{P})} \rightarrow 0 $ as  $K \rightarrow \infty$.  
\end{condition2}
Both Condition~\ref{fsbasis}$(i)$, which restricts the growth of the $\| \cdot \|_{\ell^2}$
 norm, and Condition~\ref{fsbasis}$(iii)$, which requires uniform consistency over bounded regularity classes, are mild assumptions. They are satisfied by many standard bases, including splines, CDV wavelets, and Fourier series (see, e.g.,   \citealp{chen2015optimal,belloni2015some}).

\begin{theorem}[Consistency] \label{t1}  Suppose  Conditions \ref{residuals}-\ref{fsbasis} hold and $h_0 \in L^2(\mathbb{P})$ is the unique structural function that satisfies $ \E\big( \| m(W,h_0) \|_{\ell^2}^2\big) = 0$. Let $K = K_n \rightarrow \infty$ denote any sequence that satisfies $ n^{d/2(\alpha + d)}  \lessapprox K_n  $ and $ \log(n) K_n = o(n)$. If $h_0 \in \mathcal{H}^p$ for some $p \geq \alpha + d/2$, the quasi-Bayes posterior is consistent: \begin{align}
    \label{const-eq} \mu(h : \| h - h_0  \|_{L^2(\mathbb{P})} > \epsilon \:\big|\: \mathcal{D}_n) \xrightarrow{\mathbb{P}} 0 \; \; \; \; \; \; \; \forall \: \epsilon > 0.
\end{align} 
\end{theorem}
Theorem~\ref{t1} establishes that the quasi-Bayes posterior is consistent provided that the regularity of the true function exceeds that of the Gaussian process by a factor of $d/2$. The upper bound constraint on $K_n$ is very mild: it guarantees that the first stage estimator $\widehat{m}(w,h)$ is well defined and uniformly approximates its population analog $\Pi_K m(w,h)$. By contrast, the theorem imposes a strict lower bound on the growth rate of the first-stage basis. Intuitively, larger values of $K_n$ increase sampling variability but simultaneously act as a form of regularization by shrinking the Gaussian process prior in (\ref{prior-scale}). This regularization is essential for controlling the nonlinear ill-posedness in the model. The lower bound on $(K_n)_{n=1}^{\infty}$ can be further relaxed in settings where the conditional mean function $
m(W,h) = \mathbb{E}\!\left[\rho(Y,h(X)) \,\middle|\, W\right] $
is known to smooth features of $h$ in a neighborhood of $h_0$.

Theorem~\ref{t1} can be extended in several directions. One possibility is to consider a continuously updated version of the quasi-Bayes posterior. In this case, the data-dependent weighting matrix $\widehat{\Sigma}$ may depend pointwise on both $W$ and the prior realization $h$, i.e.\ $\widehat{\Sigma} = \widehat{\Sigma}(W,h)$. The continuously updated quasi-Bayes posterior is then given by
\begin{align}
    \label{posteriorcu}  
    \mu^{CU}(\,\cdot \mid \mathcal{D}_n) 
    = \frac{\exp\big(    - \frac{n}{2}  \E_n \big[    \widehat{m}(W, \cdot) ' \widehat{\Sigma}(W, \cdot)  \widehat{m}(W, \cdot)            \big]       \big) d \mu(.)  }{\int \exp\big(    - \frac{n}{2}  \E_n \big[    \widehat{m}(W,h) ' \widehat{\Sigma}(W,h)  \widehat{m}(W,h)            \big]       \big) d \mu(h) } .
\end{align}
For example, a natural choice is a feasible estimate of the optimal continuously updated weighting matrix:  
\[
\Sigma(W,h) = \big\{ \, \mathbb{E}[ \rho(Y,h(X)) \rho(Y,h(X))' \mid W ] \, \big\}^{-1}.
\]
Another possible extension is to generalize the contraction result in Theorem~\ref{t1} to settings where the unknown function $h_0$ is not uniquely identified from the data. In this case, the identified set is given by $
\Theta_0 = \big\{ h : \| m(W,h) \|_{L^2(\mathbb{P})} = 0 \big\}.$ 
Intuitively, regardless of point identification, samples from the quasi-Bayes posterior should concentrate in regions where the quasi-Bayes objective function is minimized, i.e.\ around the identified set $\Theta_0$. Below, we state a version of Theorem~\ref{t1} that accommodates both of the preceding extensions. To this end, we impose the following condition on the weighting matrix.

\begin{conditionp}{\ref*{fsbasis}$^*$}[Weighting matrix] \label{fsbasis2}
 Over any Sobolev ball, the eigenvalues of $\widehat{\Sigma}(W,h)$ are uniformly bounded away from $0$ and $\infty$. Specifically, for every $M > 0$, there exist constants $c(M), C(M) > 0$ such that $$ \mathbb{P}\!\left( \, c \;\leq\; \inf_{h \in \mathbf{H}^t(M)} \lambda_{\min}\big(\widehat{\Sigma}(W,h)\big) \;\leq\; \sup_{h \in \mathbf{H}^t(M)} \lambda_{\max}\big(\widehat{\Sigma}(W,h)\big) \;\leq\; C \, \right) \rightarrow 1 .$$
\end{conditionp}
\begin{theorem}[Identified Set Consistency]  \label{t1-2} Let $ \Theta_0 = \{ h \in L^2(\mathbb{P})  : \| m(W,h) \|_{L^2(\mathbb{P})} = 0  \} $ denote the identified set. Suppose Conditions \ref{residuals}-\ref{fsbasis} and \ref{fsbasis2} hold. Let $K = K_n \rightarrow \infty$ denote any sequence that satisfies $ n^{d/2(\alpha + d)}  \lessapprox K_n  $ and $ \log(n) K_n = o(n)$. If there exists some $h_0 \in \Theta_0 \cap \mathcal{H}^{p}$ for $p \geq \alpha+d/2$, the continuously updated quasi-Bayes posterior $\mu^{CU}(.)$ in (\ref{posteriorcu}) is consistent for the identified set. That is, \begin{align}
    \label{const-eq-2} \mu^{CU}(h :  d(h,\Theta_0)  > \epsilon \:\big| \:\mathcal{D}_n)  \xrightarrow{\mathbb{P}} 0 \; \; \; \; \; \; \; \; \; \; \forall \: \epsilon > 0
\end{align} 
where $d(h,\Theta_0) = \inf_{h^* \in \Theta_0} \| h-h^* \|_{L^2(\mathbb{P})} $.
\end{theorem}
Theorem~\ref{t1-2} establishes the consistency of the continuously updated quasi-Bayes posterior, provided that at least one element of the identified set possesses sufficient regularity relative to the Gaussian process sample paths.
\begin{remark}[Sufficient conditions]
Consider the usual case where $\widehat{\Sigma}(w,h)$ is uniformly (over $\mathbf{H}^t(M)$ and $w$) consistent for 
$\Sigma(w,h) = \big\{ \mathbb{E}[ \rho(Y,h(X)) \rho(Y,h(X))' \mid W = w ] \big\}^{-1}$. In Example~\ref{ex1} (NPIV), we have 
$\Sigma^{-1}(W,h) = \mathbb{E}[u^2 \mid W] + \mathbb{E}[(h(X)-h_0(X))^2 \mid W]$. For any $t > d/2$, the functions in $\mathbf{H}^t(M)$ are uniformly bounded in the $\|\cdot\|_{\infty}$ norm. Thus, Condition~\ref{fsbasis2} holds if the conditional variance $\sigma^2(w) = \mathbb{E}[u^2 \mid W=w]$ is bounded above and below. In Example~\ref{ex2} (NPQIV) with a quantile $\tau \in (0,1)$, we have 
$\Sigma^{-1}(W,h)\in \{\tau^2,(1-\tau)^2\}$ for all $h$, so that Condition~\ref{fsbasis2} is trivially satisfied.
\end{remark}

For the remainder of Section~\ref{sec4}, we focus on the case with a standard weighting matrix and a uniquely identified structural function. Extensions to continuously updated weighting and partial identification can be addressed analogously to Theorem~\ref{t1-2}.

\subsection{Contraction Rates} \label{contract-rate}
In this section, we establish contraction rates for the quasi-Bayes posterior. Although Theorem~\ref{t1} established consistency, it did not quantify the rate of convergence. In the following analysis, we provide explicit posterior contraction rates.

In our setting, as we illustrate below, the posterior contraction rate is determined by the interplay among 
$(i)$ the sample path properties of the Gaussian process prior, 
$(ii)$ the local curvature of the objective function that defines the quasi-Bayes posterior, 
$(iii)$ the smoothing properties of the $h \mapsto m(W,h)$ locally around $h_0$, and 
$(iv)$ the basis functions $b^K(W) = (b_1(W), \dotsc, b_K(W))'$ used to construct a first-stage estimate of $m(W,h)$.

The behavior of the nonlinear map $h \mapsto m(W,h)$ can be locally approximated around $h_0$ by a suitable linearization. Depending on the model and the assumptions on the data  $\mathcal{D} = (Y,X,W)$, there may be multiple candidates for such a linearization. If the map $h \mapsto m(W,h)$ is sufficiently regular in a neighborhood of $h_0$, the natural choice is the Fréchet derivative at $h_0$, i.e. the unique continuous linear operator $D_{h_0}: L^2(X)\to L^2(W)$ such that
\[ \| m(W,h_0+h) - m(W,h_0) - D_{h_0}[h] \|_{L^2(\mathbb{P})} = o(\|h\|_{L^2(\mathbb{P})}) \quad \text{as } \|h\|_{L^2(\mathbb{P})}\to 0. \]
Intuitively, if $D_{h_0}[h]$ provides a good local approximation to $m(W,h)$ around $h_0$, then the smoothing properties of the nonlinear map $h \mapsto m(W,h)$ can be studied through the simpler linear operator $h \mapsto D_{h_0}[h]$. In what follows, we relate the smoothing behavior of $D_{h_0}$ to changes in regularity with respect to the orthonormal basis $(e_i)_{i=1}^{\infty}$ defining the Gaussian process in (\ref{gexpand2}). Since the smoothness of $h_0$ is also defined relative to this basis through membership in the Sobolev ball (\ref{sob}), this allows us to analyze the action of $D_{h_0}(\cdot)$ on $(G_{\alpha}, h_0)$ under a common regularity scale. To this end, it will be convenient to define a family of weak norms on $L^2(\mathcal{X})$, obtained by shrinking the Fourier coefficients of a function relative to the basis $(e_i)_{i=1}^{\infty}$. We introduce the following definition:
\begin{definition}[Weak Norms] \label{weak-n}
Let $\sigma = (\sigma_i)_{i=1}^{\infty}$ be a non-negative sequence with $\sigma_i \to 0$.  
For any $h \in L^2(\mathcal{X})$ with basis expansion $h = \sum_{i=1}^{\infty} \langle h, e_i \rangle e_i$, where $\langle \cdot, \cdot \rangle$ denotes the $L^2(\mathcal{X})$ inner product, we define the weak norm
\begin{align*}
   \| h \|_{w,\sigma}^2 = \sum_{i=1}^{\infty} \sigma_i^2 \, \big| \langle h , e_i \rangle \big|^2.    
\end{align*}
\end{definition}

For $\gamma > 0$ and $\epsilon > 0$, we denote a bounded smooth local neighborhood of $h_0$ by \begin{align}
 \label{nhd-smooth}    \Omega(M,\epsilon,\gamma) =  \{ h \in  \mathcal{H}^{ \gamma}(M)  : \| h - h_0 \|_{L^2(\mathbb{P})} \leq \epsilon    \}.
\end{align}
The following two conditions quantify the smoothing properties of the map $h \rightarrow m(W,h)$ in a local neighborhood of $h_0$ by relating it to a suitable weak norm.
\begin{condition2}[Smoothing Link] \label{gp-link}
There exists $\epsilon > 0$ sufficiently small, $\gamma > 0$ and a sequence $\sigma_i \to 0$ such that, for any $M > 0$, there are constants $C_1(M), C_2(M) < \infty$ satisfying
$
\| D_{h_0}[h - h_0] \|_{L^2(\mathbb{P})} \leq C_1(M) \| h - h_0 \|_{w,\sigma}
 $  \text{and}  $
\| h - h_0 \|_{w,\sigma} \leq C_2(M) \| D_{h_0}[h - h_0] \|_{L^2(\mathbb{P})}
$
for every $h \in \Omega(M, \epsilon, \gamma)$.
\end{condition2}
\begin{condition2}[Local Curvature] \label{lcurv} 
There exists $\epsilon > 0$ sufficiently small and $ \gamma > 0$ such that, for any $M > 0$, there exists a constant $B = B(M) < \infty$ satisfying
$
\| m(W,h) \|_{L^2(\mathbb{P})} \leq B \| D_{h_0}[h - h_0] \|_{L^2(\mathbb{P})}
 $ and $
\| D_{h_0}[h - h_0] \|_{L^2(\mathbb{P})} \leq B \| m(W,h) \|_{L^2(\mathbb{P})}
$
for every $h \in \Omega(M, \epsilon, \gamma)$.

\end{condition2}

\begin{condition2}[First Stage] \label{basis-approx}
Let $\alpha > \gamma$ denote the regularity of the Gaussian process $G_{\alpha}$. There exist sufficiently small $\epsilon, \delta > 0$, a non-increasing function $\varphi: \mathbb{R}_{+} \to \mathbb{R}_{+}$ and a constant $D > 0$ such that, for any $M > 0$,
\begin{align*}
\sup_{h \in \mathcal{H}^{\zeta}(M) \, : \, \| h - h_0 \|_{L^2(\mathbb{P})} \leq \epsilon } 
\big\| (\Pi_{K} - I) m(W,h) \big\|_{L^2(\mathbb{P})}
\;\leq\; D \, \varphi(K)\, K^{- \zeta / d} \, M
\end{align*}
for all sufficiently large $K$ and $\zeta \in (\alpha - \delta, \alpha)$.
\end{condition2}
Conditions~\ref{gp-link}--\ref{basis-approx}, albeit in varied formulations, are standard in the literature.\footnote{Our conditions are equivalent to the assumptions in \cite{chen2012estimation}; see, for example, Corollary 5.3 therein. For further discussion on alternative formulations, see also Remark \ref{variation} below.} These conditions can be further weakened to hold with a sequence $\epsilon = \epsilon_n \to 0$ sufficiently slowly. Condition~\ref{lcurv} holds trivially when $h \mapsto m(W,h)$ is linear, as in the NPIV model. If $D_{h_0}^*$ denotes the adjoint, a sufficient (but not necessary) assumption for Condition \ref{gp-link} is that the self-adjoint operator $D_{h_0}^* D_{h_0}$ diagonalizes in the eigenbasis $(e_i)_{i=1}^\infty$ of the Gaussian process in~(\ref{gexpand2}). Stronger versions of Condition~\ref{gp-link} are often imposed in the literature on linear inverse problems with a known operator (e.g. \citealp*{knapik2011bayesian}; \citealp*{gugushvili2020bayesian}).

The intuition behind Condition~\ref{basis-approx}, following \cite{chen2012estimation}, is that locally around $h_0$, the map $(h,h_0) \mapsto m(W,h) - m(W,h_0)$ exhibits smoothing properties that are comparable to those of its local linear approximation $(h,h_0) \mapsto D_{h_0}[h - h_0]$. Thus, it is expected that the decay rate of $\varphi(K)$ is of the same order as the sequence $\sigma_K$ in Condition~\ref{gp-link}, while $K^{-\zeta/d}$ represents the usual sieve approximation error for bounded smoothness classes $\mathcal{H}^{\zeta}(M)$.
\begin{remark}[On Variations of Conditions] \label{variation}
   Local curvature conditions are standard in this literature, although they appear in varying forms. We follow the formulation in \citet*{chen2012estimation,chen2014local}. Commonly used variations of Condition~\ref{lcurv} can be handled without substantive changes. For example, Remark~A.2.3 in \citet*{chernozhukov2023constrained} and Theorem~2 in \citet*{dunker2014iterative} assume (in our notation) a local curvature relation between $\|\Pi_K m(W,h)\|_{L^2(\mathbb{P})}$ and $\|\Pi_K D_{h_0}[h-h_0]\|_{L^2(\mathbb{P})}$ for all sufficiently large $K$. Under that hypothesis, our revised Condition~\ref{basis-approx}, similar to \cite*{chernozhukov2023constrained}, would instead bound the local linear bias:
\begin{equation*}
\textstyle
\Psi(K)=\sup\nolimits_{h \in \mathcal{H}^{\zeta}(M):\, \|h-h_0\|_{L^2(\mathbb{P})}\le \varepsilon}\,
\|(\Pi_K - I)\, D_{h_0}[h-h_0]\|_{L^2(\mathbb{P})}.
\end{equation*}
\end{remark}
Following standard practice in the literature, we distinguish two regimes of estimation difficulty. The model is said to be \emph{mildly ill-posed} if $\sigma_K$ and $\varphi(K)$ decay at a polynomial rate, 
and \emph{severely ill-posed} if they decay at an exponential rate. The following result establishes contraction rates for the generalized Bayes posterior.
\begin{theorem}[General Contraction Rates] \label{rate}
Suppose  Conditions \ref{residuals}-\ref{basis-approx} hold and $h_0 \in \mathcal{H}^p$ for some $p \geq \alpha + d/2$.
\begin{enumerate}
\item[$(i)$] Suppose the model is \emph{mildly ill-posed}: $\sigma_i \asymp i^{-\zeta/d},\; \varphi(K) \asymp K^{-\chi/d} $ for some $\zeta,\chi \geq 0$. If $ K_n \asymp  n^{d/[2(\alpha + \zeta) + d]}$, there exists a universal  $L > 0$ such that
\begin{align*} 
   \mu \big( h :\| h -  h_0 \|_{L^2}   >  L  n^{\frac{-\alpha}{2[\alpha +\zeta] +d}\frac{(\alpha +  \min \{ \zeta,\chi  \})}{(\alpha + \zeta)}}    \sqrt{\log n}    \: \; \big| \;  \: \mathcal{D}_n \big)   \xrightarrow{\mathbb{P}} 0.
\end{align*}
\item[$(ii)$] Suppose the model is \emph{severely ill-posed}: $\sigma_i \asymp \exp(-R i^{\zeta/d}),\;  
\varphi(K) \asymp \exp(-R' K^{\chi/d})$ 
for some $R,R',\chi,\zeta > 0$. If $ K_n \asymp  (\log n)^{1+d/\zeta} $, there exists a universal  $L > 0$ such that
\begin{align*} 
   \mu \big( h : \| h -  h_0 \|_{L^2}   >  L  (\log n)^{- \min \{ \chi(d^{-1} + \zeta^{-1}),1  \} \alpha/\zeta }    \sqrt{\log \log n}   \; \: \big| \;  \: \mathcal{D}_n \big)   \xrightarrow{\mathbb{P}} 0.
\end{align*}
\end{enumerate}
\end{theorem}
In the literature (e.g.  \citealp*{chen2012estimation,chernozhukov2023constrained}), the assumption $\varphi(K) \asymp \sigma_K$ is often imposed, as it corresponds, in a certain sense, to an optimal choice of first-stage approximating functions. Theorem~\ref{rate} allows for some degree of misspecification in this choice, with the rates simplifying under the conventional hypothesis (see Corollary~\ref{posgp} below). For clarity and simplicity of notation, we proceed under the conventional hypothesis for the remainder of the paper.

 As a point estimator for $h_0$, we consider the posterior mean \begin{align}
\label{posmean}   \E \big[ h \: | \:\mathcal{D}_n \big] = \int  h \:d \mu(h \: | \: \mathcal{D}_n). 
\end{align}
Given the posterior contraction rate in Theorem~\ref{rate}, the posterior mean, as a point estimator, is expected to converge at a comparable rate. Intuitively, this follows if the posterior probability of the set where contraction fails decays sufficiently quickly. The next result formalizes this intuition.

\begin{corollary}[Rates of Convergence]
\label{posgp} 
Suppose the hypothesis of Theorem \ref{rate} holds.
\begin{enumerate}
\item[$(i)$] If the model is mildly ill-posed, there exists a universal constant $ L > 0 $ such that \begin{align*}
    \mathbb{P} \bigg(  \|  h_0 -  \E \big[ h \:| \:\mathcal{D}_n \big]  \|_{L^2(\mathbb{P})}  >    L   n^{\frac{-\alpha}{2[\alpha +\zeta] +d}}    \sqrt{\log n}       \bigg)  \rightarrow 0.
\end{align*}
\item[$(ii)$] If the model is severely ill-posed, there exists a universal constant $ L > 0 $ such that
\begin{align*}
   \mathbb{P} \bigg(  \|  h_0 -  \E \big[ h \:| \:\mathcal{D}_n \big]   \|_{L^2(\mathbb{P})}  >    L   (\log n)^{- \alpha / \zeta} \sqrt{\log \log n}     \bigg)  \rightarrow 0.
\end{align*}
\end{enumerate}
\end{corollary}
\begin{remark}[Optimal Rates]
The preceding results require that the regularity $p$ of the structural function $h_0$ exceed that of the Gaussian process $G_{\alpha}$ by at least $d/2$, i.e.  $p \geq \alpha + d/2$. Consequently, the fastest attainable rate occurs when $\alpha = p - d/2$. This rate is slower than the ``optimal'' rate in \cite{chen2012estimation}, which corresponds to $\alpha = p$. In our setting, the additional smoothness of $h_0$ relative to the prior is crucial for controlling the nonlinear inverse problem induced by the infinite-dimensional prior. While sharper rates may be possible, establishing them within the current non-conjugate framework appears challenging.
\end{remark}

\subsection{Inference} \label{inference}
In this section, we study the limiting quasi-posterior distribution for a class of linear functionals. Let $\mathbf{L}(h_0)$ denote a linear functional of interest—for example, the average value of $h_0(\cdot)$ over an interval or its average derivative. Our analysis focuses on two main questions:
$(i)$ What is the limiting quasi-Bayes posterior distribution of $\mathbf{L}(h)$?
$(ii)$ Under what conditions do quasi-Bayes credible sets for $\mathbf{L}(h_0)$ attain valid frequentist coverage?

To begin our analysis, we view the linear functional as a map $\mathbf{L} : L^2(\mathcal{X}) \rightarrow \mathbb{R}$. Then, by the Riesz representation theorem, there exists a function $\Phi \in L^2(\mathcal{X})$ such that
\begin{align} \label{Linf} \mathbf{L}(h) =  \langle h , \Phi \rangle_{L^2(\mathbb{P})} =\mathbb{E}[h(X)\Phi(X)] \qquad \forall \: h \in L^2(\mathcal{X}). \end{align}
The advantage of this representation is that properties of $\mathbf{L}(\cdot)$ (e.g. regularity) can be analyzed through its representer function $\Phi(X)$. 

In the preceding sections, the choice of the weighting matrix $\widehat{\Sigma}(\cdot)$ in the quasi-Bayes posterior (\ref{qb-general}) did not affect the limit theory, provided that the eigenvalues of $\widehat{\Sigma}(\cdot)$ remained asymptotically bounded away from $0$ and $\infty$. Intuitively, under this condition, the rates of convergence can be characterized by analyzing a quasi-Bayes posterior based on the identity weighted objective $h \mapsto \E_n \big( \| \widehat{m}(W,h) \|_{\ell^2}^2 \big)$. To characterize finer aspects of the posterior, it will be necessary to account for the limiting behavior of  $\widehat{\Sigma}(\cdot)$ in the analysis. We impose the following low level condition on the limiting behavior of the weights.
\begin{condition2}[Limiting Weights]\label{limit-weight}
There exists a limit symmetric matrix $\Sigma_0(\cdot)$ such that
$
\sup_{w\in\mathcal{W}}
\bigl\|\widehat{\Sigma}(w)-\Sigma_0(w)\bigr\|_{op}
=O_{\mathbb{P}}(\gamma_n),
$
where $ (\gamma_n)_{n=1}^{\infty} $ satisfies 
\(\gamma_nK_n\to0\). Furthermore, the eigenvalues of \(\Sigma_0(W)\) are uniformly bounded away from zero and infinity: 
\[
\mathbb{P}\Bigl(c\le\lambda_{\min}\bigl(\Sigma_0(W)\bigr)
\le\lambda_{\max}\bigl(\Sigma_0(W)\bigr)\le C\Bigr)
=1
\]
for some universal constants $c,C > 0$.
\end{condition2}
We are primarily interested in the setting where \(\Sigma_0(\cdot)\) is an efficient weighting matrix for the conditional moment restriction, so that \(\widehat\Sigma(\cdot)\) may be viewed as a preliminary first‑step estimate of the optimal weighting matrix.  In finite‑dimensional GMM models, a celebrated result by \citet{chernozhukov2003mcmc} establishes the frequentist validity of optimally weighted quasi‑Bayes credible sets. In this section, we provide a nonparametric extension to their results by studying the frequentist coverage of quasi-Bayes credible sets for the functional \(\mathbf L(h_0)\).

As in Section \ref{contract-rate}, let \( D_{h_0}(\cdot) \) denote the Fréchet derivative of the map \( h \mapsto m(W,h) \) at \( h_0 \). We denote its adjoint by \( D_{h_0}^* \).\footnote{In defining $D_{h_0}^*$, we view $D_{h_0}$ as a map $(L^2(X) , \| . \|_{L^2(\mathbb{P})}) \mapsto (L^2(W , \| . \|_{L_{\Sigma_0}^2(\mathbb{P})}) $, where $\| . \|_{L^2_{\Sigma_0}(\mathbb{P})}$ denotes the optimal weighted norm $
\| D_{h_0}(h) \|_{L^2_{\Sigma_0}(\mathbb{P})}^2 = \mathbb{E} \left[ D_{h_0}(h)' \Sigma_0(W) D_{h_0}(h) \right]$.} Let $\mathbb{H}$ denote the reproducing kernel Hilbert space (RKHS) of the Gaussian process $G_{\alpha}$. The following condition specifies our main regularity requirements on the representer function $\Phi(\cdot)$.

\begin{condition2}[Regular Functional] \label{phi-reg}
There exists $\tilde{\Phi} \in \mathbb{H}$ such that $\Phi = D_{h_0}^* D_{h_0} \tilde{\Phi}$. The first-stage approximation biases of $D_{h_0}[\tilde{\Phi}]$ and $\Sigma_0(W) D_{h_0}[\tilde{\Phi}]$ satisfy:
\begin{align*}
\text{(i)} \quad & \sqrt{K_n} \sqrt{\log n} \, \|(\Pi_{K_n} - I) D_{h_0}[\tilde{\Phi}] \|_{L^2(\mathbb{P})} \rightarrow 0, \\
\text{(ii)} \quad & \sqrt{K_n} \sqrt{\log n} \, \|(\Pi_{K_n} - I) \Sigma_0(W) D_{h_0}[\tilde{\Phi}] \|_{L^2(\mathbb{P})} \rightarrow 0.
\end{align*}
\end{condition2}
The requirement that $\Phi$ lie in a suitable range of the adjoint is a well-known necessary condition for $\sqrt{n}$ estimation of linear functionals, appearing in a variety of settings. For exogenous nonlinear regression models, see \citet*{monard2021statistical}; for NPIV models, see \citet*{severini2012efficiency}, \citet{bennett2022inference}, \citet*{deaner12trade}; and for NPQIV models, see \citet*{chen2019penalized}. This condition implicitly imposes regularity constraints on $\Phi$. Although extending to more general settings, such as irregular functionals, would be desirable, we view our analysis as an important first step toward a comprehensive nonparametric quasi-Bayes inferential theory.

Given the posterior contraction rate established in Theorem \ref{rate}, it suffices, for deriving the distributional limit theory, to restrict our analysis to a quasi-Bayes posterior whose support is contained within local neighborhoods of $h_0$. Specifically, if $\Theta_n$ denotes a sequence of shrinking local neighborhoods around $h_0$, it suffices to focus on the \emph{localized posterior}:
\begin{align} \label{local-pos}
\mu^{\star}(A \mid \mathcal{D}_n) \;=\; 
\frac{ \int_{A \cap \Theta_n}  
\exp\!\Big( - \tfrac{n}{2}\, \mathbb{E}_n \big[ \widehat{m}(W,h)^{\prime}\, \widehat{\Sigma}(W)\, \widehat{m}(W,h) \big] \Big)\, d\mu(h)}
{\int_{\Theta_n} 
\exp\!\Big( - \tfrac{n}{2}\, \mathbb{E}_n \big[ \widehat{m}(W,h)^{\prime}\, \widehat{\Sigma}(W)\, \widehat{m}(W,h) \big] \Big)\, d\mu(h)} \, .
\end{align}
Let $\delta_n$ denote the posterior contraction rate established in Theorem \ref{rate}. In our analysis, we will also make use of the contraction rate $\xi_n$, obtained with the weaker metric $d_w(h,h_0) = \| m(W,h) - m(W,h_0) \|_{L^2(\mathbb{P})}$. As a byproduct of our earlier analysis, it is straightforward to verify that this contraction rate is given by
\begin{align*}
\xi_n =    
\begin{cases}      
n^{- \frac{ \alpha + \zeta }{2(\alpha + \zeta) + d}} \sqrt{\log n} & \text{mildly ill-posed}, \\  
 (\log n)^{1 + (d/2 \zeta)} n^{-1/2} & \text{severely ill-posed}.
\end{cases}
\end{align*}
If $\gamma > 0$ is as in Condition \ref{gp-link}-\ref{basis-approx}, we consider the localized distribution  $\mu^{\star}(\cdot \:|\:\mathcal{D}_n)$ obtained through the sequence of smooth local neighborhoods: \begin{align*}
    \Theta_n = \bigg \{ h  \in  \mathcal{H}^{\gamma} (M)  & :  \|    m(W,h_{})  - m(W,h_0)   \|_{L^2(\mathbb{P})}  \leq   D  \xi_n ,        \| h- h_0 \|_{L^2(\mathbb{P})} \leq D \delta_n       \bigg \}
\end{align*}
where $D,M > 0$ are sufficiently large universal constants. 

To connect with the usual linear distributional theory, we quantify the discrepancy between  $m(W,h)$ and its linear approximation $D_{h_0}[h-h_0]$ locally around $h_0$. To that end, given any function $h : \mathcal{X} \rightarrow \R$, we denote the remainder obtained from linearizing the map $h \rightarrow m(W,h)$  locally around $h_0$ by  \begin{align}
    \label{linear-remainder} R_{h_0}(h,W) =  m(W,h) - m(W,h_0) - D_{h_0}[h-h_0] .
\end{align}
For linear problems such as NPIV (Example \ref{ex1}), we have $R_{h_0}(h,W) = 0$ for every $h$. As such, including (\ref{linear-remainder}) in the analysis is only relevant for nonlinear models. Analogous to the finite dimensional Euclidean case, the remainder vanishes as $\| h- h_0 \|_{L^2(\mathbb{P})} \rightarrow 0$. The precise rate at which this occurs depends on (among other factors) $(i)$ the ill-posedness in the model, $(ii)$ the regularity of $h$ and $(iii)$ the convergence rate of $  \| h- h_0 \|_{L^2(\mathbb{P})}$. 

Let $ \mathcal{M}_n = \{  m(\cdot,h) : h \in \Theta_n \} $ denote the image of $\Theta_n$ under the first stage map $h \mapsto m(W,h)$. As is standard, we quantify the complexity of $\mathcal{M}_n$ through its entropy integral: \begin{align}
    \label{entropy-int}  \mathcal{J}(\epsilon) =  \int_{0}^{\epsilon}  \sqrt{\log N( \mathcal{M}_n , \| . \|_{L^2(\mathbb{P})}, \tau D \xi_n  )  } d \tau \;,
\end{align}
where $N(\mathcal{S},d,\delta)$ denotes the usual $\delta-$covering number of a set $\mathcal{S}$ with respect to the metric $d$. The following condition specifies our requirements on the localized support $\Theta_n$, its image $\mathcal{M}_n$ and nonlinear remainder $\{ R_{h_0}(h,W) : h \in \Theta_n \}$.

\begin{condition2}\label{misc1}   
Let $\kappa$ and $t$ denote the local $L^2$ continuity parameters of the generalized residual $\rho(\cdot)$, as defined in Condition~\ref{residuals}. Suppose that:
\begin{align}
  \nonumber  &(i) \quad n^{-1/2} K_n^2 \, \mathcal{J}(K_n^{-1/2}) \xrightarrow[n \rightarrow \infty]{} 0. \\
  \nonumber  &(ii) \quad \sqrt{\log K_n} \cdot \max \left\{ 
        \frac{K_n^2 \log K_n}{\sqrt{n}}, \,
        \frac{K_n \delta_n^{-d/t}}{\sqrt{n}}, \,
        K_n \sqrt{\log K_n} \delta_n^{\kappa}, \,
        \sqrt{K_n} \delta_n^{\kappa - d/(2t)} 
    \right\} \xrightarrow[n \rightarrow \infty]{} 0. \\
  \nonumber  &(iii) \quad \sqrt{n} \sqrt{K_n \log n} \cdot 
    \sup_{h \in \Theta_n} \left\| \Pi_{K_n} R_{h_0}(h,W) \right\|_{L^2(\mathbb{P})} \xrightarrow[n \rightarrow \infty]{} 0.
\end{align}
\end{condition2}
Conditions~\ref{misc1}$(i)$--$(ii)$ arise primarily from empirical process techniques  used to control the uniform empirical deviation:
\[ \chi_n = \sup_{h \in \Theta_n}
\left| \mathbb{E}_n \left[ \widehat{m}(W,h)' \Sigma(W) \widehat{m}(W) \right] - \mathbb{E} \left[ \Pi_{K} m(W,h)' \Sigma(W) \Pi_{K} m(W,h) \right] \right|.
\]
If we substitute the posterior contraction rate $\delta_n$ and the optimal first-stage sieve dimension sequence $K_n$ from Theorem~\ref{rate}, Condition~\ref{misc1} can be reduced to minimum smoothness requirements on the structural function $h_0$ and prior. The dependence on $\kappa$ and $t$ arises because the generalized residual function $\rho(\cdot)$ may be nonlinear and pointwise discontinuous in $h$. Accordingly, our analysis relies on the weaker $L^2(\mathbb{P})$ continuity condition specified in Condition~\ref{residuals}.
\begin{remark}[On the Remainder Order]
Condition~\ref{misc1}$(iii)$ imposes that the nonlinear remainder vanishes sufficiently fast on local shrinking neighborhoods around $h_0$. Under weak conditions, the remainder satisfies a quadratic bound:
\begin{align} \label{quadrem}
\| \Pi_{K_n} R_{h_0}(h,W) \|_{L^2(\mathbb{P})} \leq \| R_{h_0}(h,W) \|_{L^2(\mathbb{P})} \leq C \| h - h_0 \|_{L^2(\mathbb{P})}^2 \qquad \forall \; h \in \Theta_n.
\end{align}
For mildly ill-posed models, Condition~\ref{misc1}$(iii)$ is satisfied if $\delta_n^2 \sqrt{K_n} \sqrt{\log n} = o(n^{-1/2})$. Substituting the definition of $K_n$ from Theorem~\ref{rate}, this reduces to the smoothness requirement $\alpha > \zeta + d$, similar to Condition 5.7 in \citet{chen2009efficient}. As noted in the literature (e.g.  \citealp*{hanke1995convergence}) quadratic bounds such as (\ref{quadrem}) are usually overly conservative in ill-posed settings. In nonlinear inverse problems, a more informative bound is the tangential cone condition \citep*{chen2014local}, which in our notation requires
\begin{equation}\label{tang-cone}
\|R_{h_0}(h,W)\|_{L^2(\mathbb{P})}
\;\le\;
\phi\!\big(\|h-h_0\|_{L^2(\mathbb{P})}\big)\,
\|m(W,h)-m(W,h_0)\|_{L^2(\mathbb{P})}
\qquad \forall\, h\in\Theta_n,
\end{equation}
for some function $\phi:\mathbb{R}_+\to\mathbb{R}_+$ with $\phi(0)=0$ and continuous at zero.\footnote{This is expression (1.8) in \citet*{hanke1995convergence} with $\phi(t) = t$. For uses and proofs of tangential cone conditions in other settings, see e.g. \citet{kaltenbacher2009iterative}; \citet{de2012local}; \citet{dunker2014iterative}; \citet{breunig2020specification}.} For instance, if $\phi(t) = t$, then \eqref{tang-cone} implies that Condition~\ref{misc1}$(iii)$ holds for severely ill-posed models when $\alpha > \zeta + d$, and for mildly ill-posed models when $\alpha > d$.
\end{remark}
The following result establishes that the quasi-Bayes posterior distribution of a regular functional $\mathbf{L}(.) = \langle \cdot ,  \Phi \rangle_{L^2(\mathbb{P})}$  is well approximated by a suitable Gaussian measure.

\begin{theorem}[Bernstein--von Mises] \label{bvm}
Suppose $h_0 \in \mathcal{H}^p$ for some $p \geq \alpha + d/2$, and let Conditions \ref{residuals}--\ref{misc1} hold. Then:
\begin{align*}
    (i) \;\;\; & \sqrt{n} \, \langle  h - \E \big[h \mid \mathcal{D}_n \big] , \Phi \rangle_{L^2(\mathbb{P})} \;  \big| \; \mathcal{D}_n \: \overset{\mathbb{P}}{\rightsquigarrow} \; \mathcal{N} \big( 0 ,  \E \big[ (D_{h_0} \tilde{\Phi}    )' \Sigma_0 \, (D_{h_0} \tilde{\Phi}) \big] \big), \\
    (ii) \;\; & \sqrt{n} \, \langle  h_0 - \E \big[ h \mid \mathcal{D}_n \big] , \Phi \rangle_{L^2(\mathbb{P})}  \rightsquigarrow \; \mathcal{N} \big( 0 , \E \big[ (D_{h_0} \tilde{\Phi})' \Sigma_0 \, \rho_{\star} \rho_{\star}' \Sigma_0 \, (D_{h_0} \tilde{\Phi}) \big] \big)
\end{align*}
where $\rho_{\star} = \rho(Y,h_0(X)) $ and    $\overset{\mathbb{P}}{\rightsquigarrow}$ denotes weak convergence in probability.
\end{theorem}
The two variances in Theorem \ref{bvm} coincide if and only if the quasi-Bayes posterior is optimally weighted. That is, when the weighting matrix is \begin{align*}  \Sigma_0(W) = \{ \E[ \rho(Y,h_0(X)) \rho(Y,h_0(X))'|W  ] \}^{-1} . \end{align*}  
An important implication of Theorem \ref{bvm} is that optimally weighted quasi-Bayes credible sets, centered around the posterior mean, attain asymptotically exact frequentist coverage. Specifically, given a linear functional \( \mathbf{L}(\cdot) \) and a significance level \( \gamma \in (0,1) \), define
\[
c_{1-\gamma} =  (1-\gamma) \text{ quantile of } \left| \mathbf{L}(h) - \mathbf{L}\left( \mathbb{E}[h \mid \mathcal{D}_n] \right) \right|, \quad h \sim \mu(\cdot \mid \mathcal{D}_n).
\]
The quasi-Bayes credible set at level \( \gamma \) is defined as:
\[
C_n(\gamma) = \left\{ t \in \mathbb{R} : \left| t - \mathbf{L}\left( \mathbb{E}[h \mid \mathcal{D}_n] \right) \right| \leq c_{1-\gamma} \right\}.
\]
\begin{corollary}[Frequentist coverage]
\label{bvm-col} 
Suppose the assumptions of Theorem~\ref{bvm} hold, and the quasi-Bayes posterior is optimally weighted. Then, for any significance level $\gamma$,
\begin{align*}
 \lim_{n \rightarrow \infty} \mathbb{P} \left( \mathbf{L}(h_0) \in C_n(\gamma) \right) =  1 - \gamma.
\end{align*}
\end{corollary}
To the best of our knowledge, Theorem \ref{bvm} and Corollary \ref{bvm-col} provide the first nonparametric quasi-Bayes inferential guarantees in the literature. These results extend classical quasi-Bayes inferential results for parametric GMM \citep{chernozhukov2003mcmc} to  nonparametric conditional moment restriction models.
\begin{remark}[Semiparametric efficiency]
The equality of variances in Theorem~\ref{bvm} suggests that an optimally weighted quasi-Bayes posterior mean is asymptotically efficient. Observe that, under optimal weighting, the common limiting variance is: \begin{align*}
     V_{\Phi} =  \E \big[ (D_{h_0} \tilde{\Phi} )' \{ \E[ \rho(Y,h_0(X)) \rho(Y,h_0(X))'|W  ] \}^{-1} (D_{h_0} \tilde{\Phi} )     \big]      .
\end{align*}
In settings where the semiparametric efficiency bound can be analytically characterized, quasi-Bayes efficiency can be assessed by comparing   $V_{\Phi}$ to the efficient lower bound. For example, in the NPIV model, substituting $\tilde{\Phi} = (D_{h_0}^* D_{h_0})^{-1} \Phi$ recovers the semiparametric efficiency bound derived in \citet{severini2012efficiency}. 
\end{remark}

\section{Simulations} \label{simulations}
In this section, we present additional simulation evidence on the finite-sample performance of quasi-Bayes posteriors. Whereas Section \ref{sim-evid} focused on structural functions with a univariate regressor, here we consider settings with multivariate regressors.

Specifically, we examine multivariate generalizations of the designs in \citet{newey2003instrumental}, \citet{santos2012inference}, \citet*{chernozhukov2015constrained}, \citet{chetverikov2017nonparametric}, and \citet*{chen2025adaptive}, which we denote as \textbf{NP}, \textbf{S}, \textbf{CNS}, \textbf{CW}, and \textbf{CCK}, respectively. These generalizations are constructed to mimic the endogeneity structure and ill-posedness of the original univariate designs.\footnote{GPT-5  assisted in the construction of these generalizations.} The structural functions are: \par\noindent
\(
\textstyle
\begin{aligned}
\textbf{NP:} \quad\;& h_0(x)=\sum_{j=1}^{5}\log\!\big(1+|x_j-1|\big)\,\mathrm{sign}(x_j-1)
+\tbinom{5}{2}^{-1}\!\!\sum_{1\le j<k\le5}\sin(\pi x_j x_k),\\
\textbf{S:} \quad \;&  h_0(x)=\sin(\pi x_1)+0.5\,\sin\!\big(\pi(x_3-x_2)\big)
+0.5\,\cos\!\big(\pi(x_5-x_4)\big),\\
\textbf{CNS:} \quad\;& h_0(x)=\sum_{j=1}^{5}\Big(1-2\,\Phi(x_j-0.5)\Big),\\
\textbf{CW:} \quad\;&  h_0(x)=\sum_{j=1}^{5}\!\Big(2\max(x_j-0.5,0)^2+0.5 x_j\Big)
+ x_3x_4 + \log\!\big(1+x_1x_2x_5\big),\\
\textbf{CCK:} \quad\;& h_0(x)=\sin(4x_1)\log x_1 + 1.5\,\cos(\pi x_2) + x_3^2 - 0.5\,x_4x_5.
\end{aligned}
\)\par\bigskip
In these designs, the endogenous regressor is five-dimensional, \(X \in \mathbb{R}^5\), and the instrument is two-dimensional, \(W \in \mathbb{R}^2\). 
The structural functions extend those used in the original univariate designs, and collectively span a reasonable spectrum of functional complexity. Beyond maintaining a similar endogeneity structure, we also scaled up the variance of the disturbances to ensure that the signal-to-noise ratios remain comparable to, or smaller than, those in the original univariate designs. All details are provided in Appendix \ref{append-sims}.

In endogenous models with multivariate regressors, it is very challenging to estimate the structural function using classical methods. Indeed, with a five dimensional endogenous regressor, even a minimal tensor-product sieve with three terms per coordinate yields $J = 3^5 = 243$ basis functions. In all designs, 2SLS estimation based on this tensor product produced an extremely large and unstable risk. This mirrors the univariate behavior in Table \ref{table1}, except that in higher dimensions the minimal feasible $J$ is already prohibitively large.

Let \textbf{QB} denote the quasi-Bayes posterior mean, based on a first-stage thin-plate spline with dimension $K=15$ and a Whittle–Matérn Gaussian process prior. The same prior and implementation algorithm are used across all designs and both sets of restrictions (see Appendix \ref{implem-append} for details). For comparison, we also report nonparametric regression estimates using random forests (\textbf{RF}), implemented via the \texttt{ranger} package in R.
\subsection{Results}
\begin{table}[!htbp]
\caption{Sample size $n=2000$. MSE risk $\mathcal{R}^2(\widehat{h},h_0)$ based on 1000 replications.}
\vspace{0.5em}
\label{table3}
\centering
\setlength{\tabcolsep}{3pt}
\scalebox{1}{
\begin{tabular}{
  l
  @{\hspace{15pt}}c
  @{\hspace{15pt}}c
  @{\hspace{15pt}}c
}
\toprule
 Design 
  & \textbf{QB} (NPIV)
  & \textbf{QB} (NPQIV)
  & \textbf{RF} (OLS) \\
\midrule
\textbf{NP}   & 0.541 & 0.737 & 2.05 \\
\textbf{S}    & 0.501 & 0.486 & 2.82 \\
\textbf{CNS}  & 0.156 & 0.053 & 1.94 \\
\textbf{CW}   & 0.313 & 0.268 & 1.51 \\
\textbf{CCK}  & 0.622 & 0.915 & 3.02 \\
\bottomrule
\end{tabular}}
\end{table}
Random forests (\textbf{RF}) are a reliable supervised learning method for high-dimensional regression and are expected to capture much of the variation in the structural functions. However, because of the non-trivial endogeneity in the designs, it exhibits substantial bias. The designs in Table \ref{table3} span a wide range of structural function complexities and endogeneity patterns, with some expected to serve as relatively challenging stress tests. In practice, we expect our methods to perform considerably better in more conventional settings.

The results in Table \ref{table3} demonstrate that the quasi-Bayes estimators perform well and are viable in higher dimensions. In particular, the estimators are accurate and stable across both restrictions. This is especially noteworthy since nonparametric quantile IV (NPQIV) estimation is often regarded as a substantially more difficult problem due to its nonlinear and discontinuous generalized residual. Together with the simulation evidence in Section \ref{motiv-sim}, our findings suggest that quasi-Bayes estimators may provide a broadly useful toolkit for the large class of nonlinear restrictions frequently encountered in applied work. 

\begin{figure}[!htbp]
    \centering
    \includegraphics[width=0.75\textwidth]{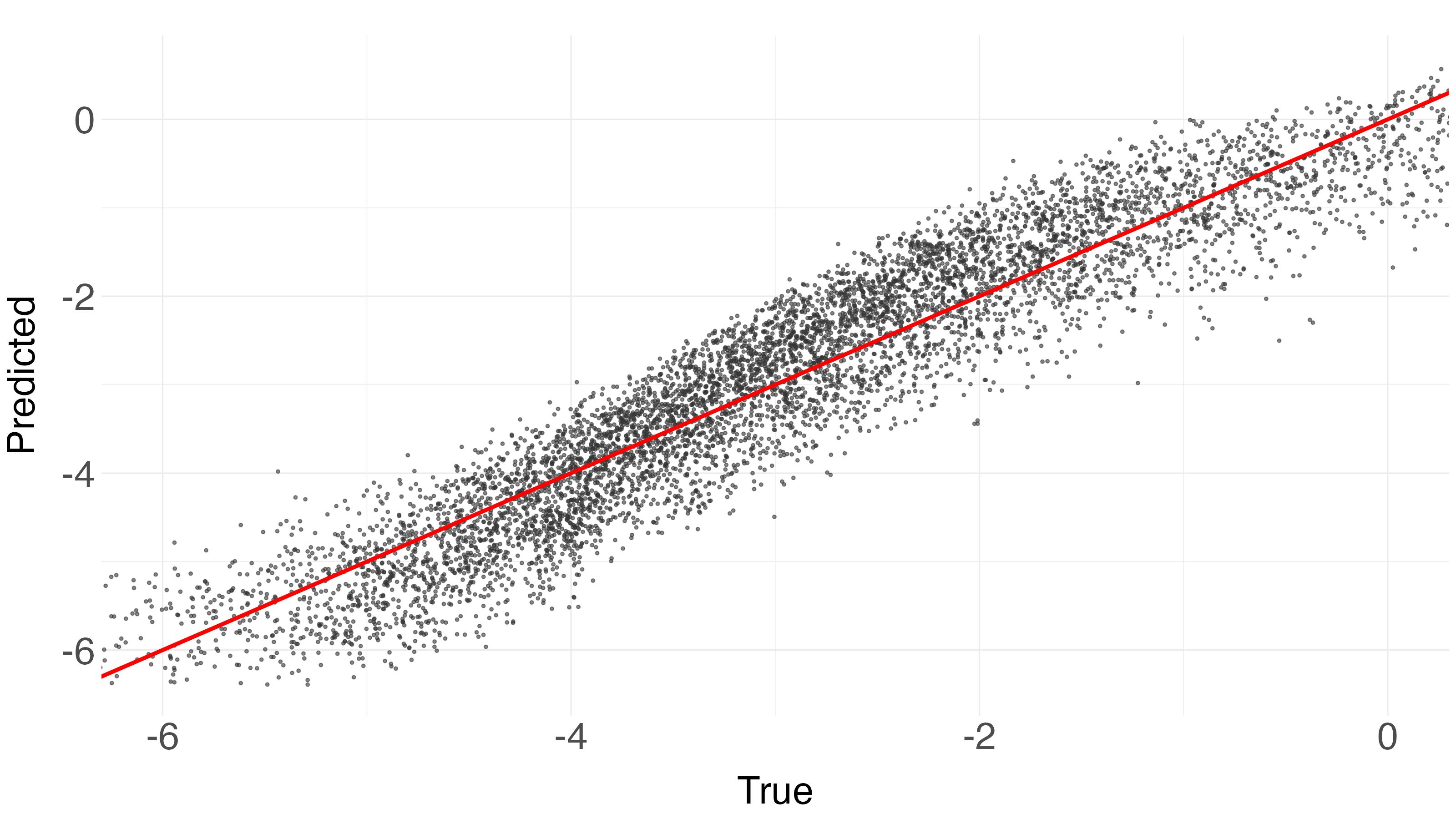}
    \caption{Scatter plot of true vs. predicted values for the multivariate \textbf{NP} design. 
    Quasi-Bayes (NPIV) predictions. The red 45$^\circ$ line denotes perfect prediction ($\text{True} = \text{Predicted}$).}
    \label{npm_qb}
\end{figure}
\begin{figure}[!htbp]
    \centering
    \includegraphics[width=0.75\textwidth]{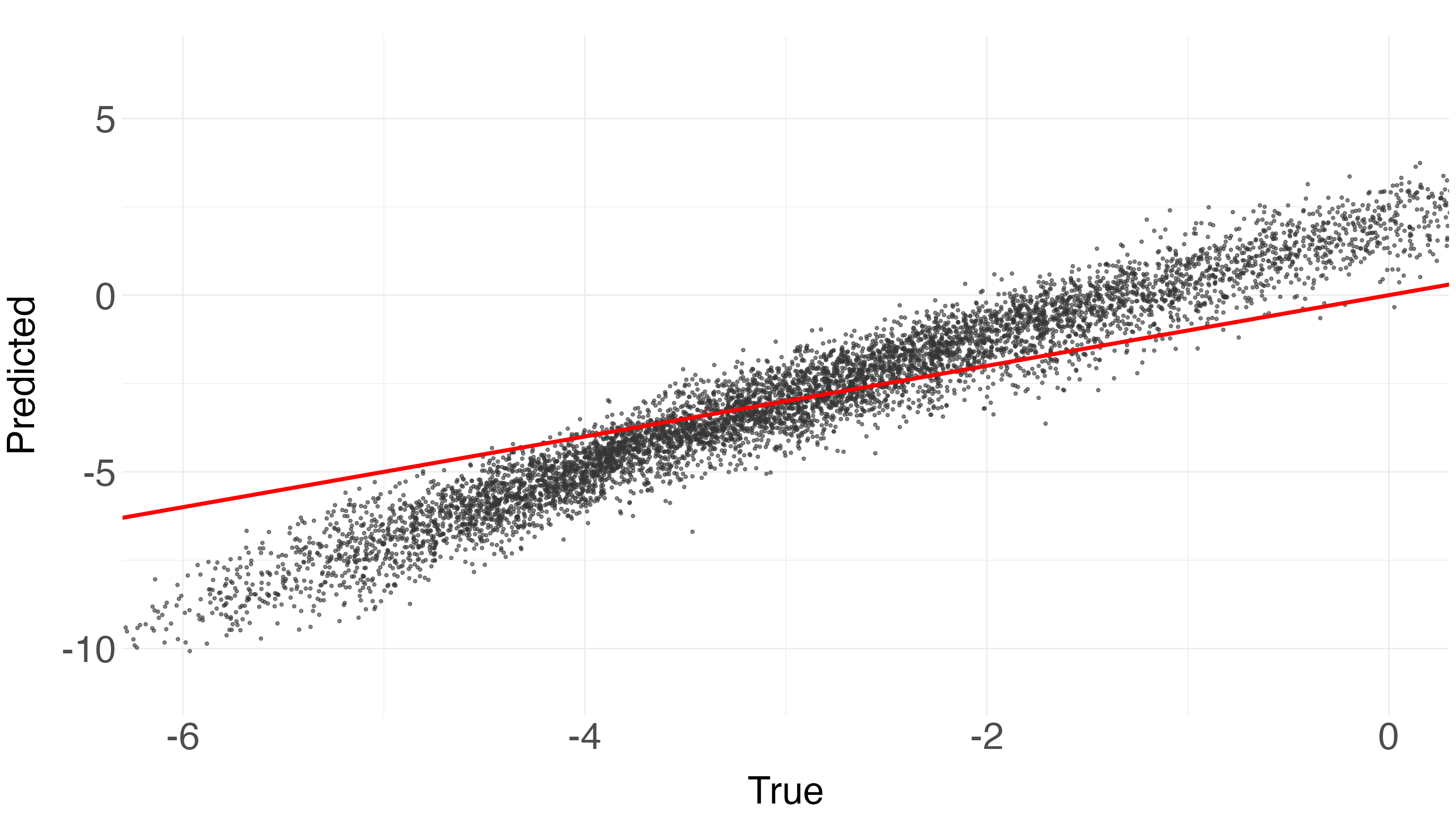}
    \caption{Scatter plot of true vs. predicted values for the multivariate \textbf{NP} design. 
    Random forest (OLS) predictions. The red 45$^\circ$ line denotes perfect prediction ($\text{True} = \text{Predicted}$).}
    \label{npm_ols}
\end{figure}
Figure \ref{npm_qb} plots a sample realization of quasi-Bayes predicted vs true values on a generated test data. The predictions closely follow the trajectory of the true values, concentrating around the 45-degree line of equality. Figure \ref{npm_ols} plots the associated fit for the biased OLS predictions.

As a final remark, it would be desirable to compare the quasi-Bayes estimators with other nonparametric alternatives. However, we are not aware of any reliable implementations for general conditional moment models with multivariate regressors. To the best of our knowledge, our simulation study also provides the first nonparametric risk estimates for quantile IV models with multivariate regressors.

\section{Application: Production Functions} \label{applic-prod}
In this section, we apply our methodology to estimate firm-level production functions in Chile, using data from the national census of manufacturing plants conducted by Chile’s \textit{Instituto Nacional de Estadística}. This dataset is frequently employed in studies of firm-level production functions (e.g. \citealp*{levinsohn2003estimating,gandhi2020identification}). Our analysis focuses on the food products industry, one of the country’s largest manufacturing sectors. We use firms with more than 10 employees and complete observations for the years 1979–1996.

Let $y_{it}, k_{it}, l_{it}$ denote the logarithms of gross output, capital, and labor, respectively, and let $m_{it}$ denote intermediate inputs (fuels, materials, and electricity). All variables are in real terms. Consider the structural value-added production model
\[
y_{it} = F(l_{it}, k_{it}) + \omega_{it} + \varepsilon_{it},
\]
where $F(\cdot)$ is the production function in inputs $(l,k)$, $\varepsilon_{it}$ are exogenous shocks unobserved by the firm, and $\omega_{it}$ are first-order Markov shocks observed (or predictable) by the firm prior to its input decisions at time $t$. We assume $\omega_{it}$ is a deterministic function of inputs, $\omega_{it} = \tilde f_t(k_{it}, l_{it}, m_{it})$, for some function $\tilde f_t$. One interpretation of this specification, following \cite*{ackerberg2015identification}, is that the gross-output production function is Leontief in the intermediate input. Define the conditional means \[
g(\omega_{it-1}) = \mathbb{E}[\omega_{it} \mid \omega_{it-1}] \quad, \quad 
\Phi_t(l_{it}, k_{it}, m_{it}) = \mathbb{E}[y_{it} \mid l_{it}, k_{it}, m_{it}].
\]
Note that, since $\varepsilon_{it}$ is exogenous noise, the function $g(\cdot)$ can be interpreted as the conditional mean regression of $\Phi_{t}(l_{it}, k_{it}, m_{it}) - F(l_{it}, k_{it})$ on $\Phi_{t-1}(l_{it-1}, k_{it-1}, m_{it-1}) - F(l_{it-1}, k_{it-1})$. If $\mathcal{I}_{t}$ denotes the firm’s information set at time $t$, it is shown in \cite*{ackerberg2015identification} that  $ F(\cdot)$ satisfies the conditional moment restriction:
\begin{align} \label{ackerberg-cmr} 
\mathbb{E}\!\left[\, y_{it} - F(l_{it}, k_{it}) - g\!\Big( \Phi_{t-1}(l_{it-1}, k_{it-1}, m_{it-1}) - F(l_{it-1}, k_{it-1}) \Big) \,\middle|\, \mathcal{I}_{t-1} \right] = 0.
\end{align}
In most industries, it is assumed that firms choose labor $l_{it}$ after period $t-1$. Under this timing assumption, the natural information set, as in \cite*{ackerberg2015identification}, is $
  \mathcal{I}_{t-1} = \{ k_{it}, \, l_{it-1}, \, \Phi_{t-1} \}
$. We use the same information set in our analysis.

The functions $g(\cdot)$ and $\Phi_{t-1}(\cdot)$ are smooth, low-dimensional regressions and can therefore be estimated accurately with standard nonparametric methods. In practice,  $\Phi_{t-1}(\cdot)$ is typically estimated using a flexible sieve regression (e.g.  splines). Similarly, for any input function $\tilde{F}$, the output of $g(\cdot)$ in the restriction is obtained from a one-dimensional conditional mean regression, typically implemented with a flexible polynomial. We adopt this approach and thus treat both functions as known for the restriction in (\ref{ackerberg-cmr}). Further implementation details are provided in Appendix \ref{implem-append}.

We aim to estimate the production function that satisfies the conditional moment restriction in (\ref{ackerberg-cmr}). This is a particularly challenging problem, as the restriction defines a complex and highly nonlinear inverse problem in $F(\cdot)$. 
\subsection{Analysis} \label{prod-res}
\begin{figure}[!htbp]
    \centering
    \includegraphics[width=0.9\textwidth]{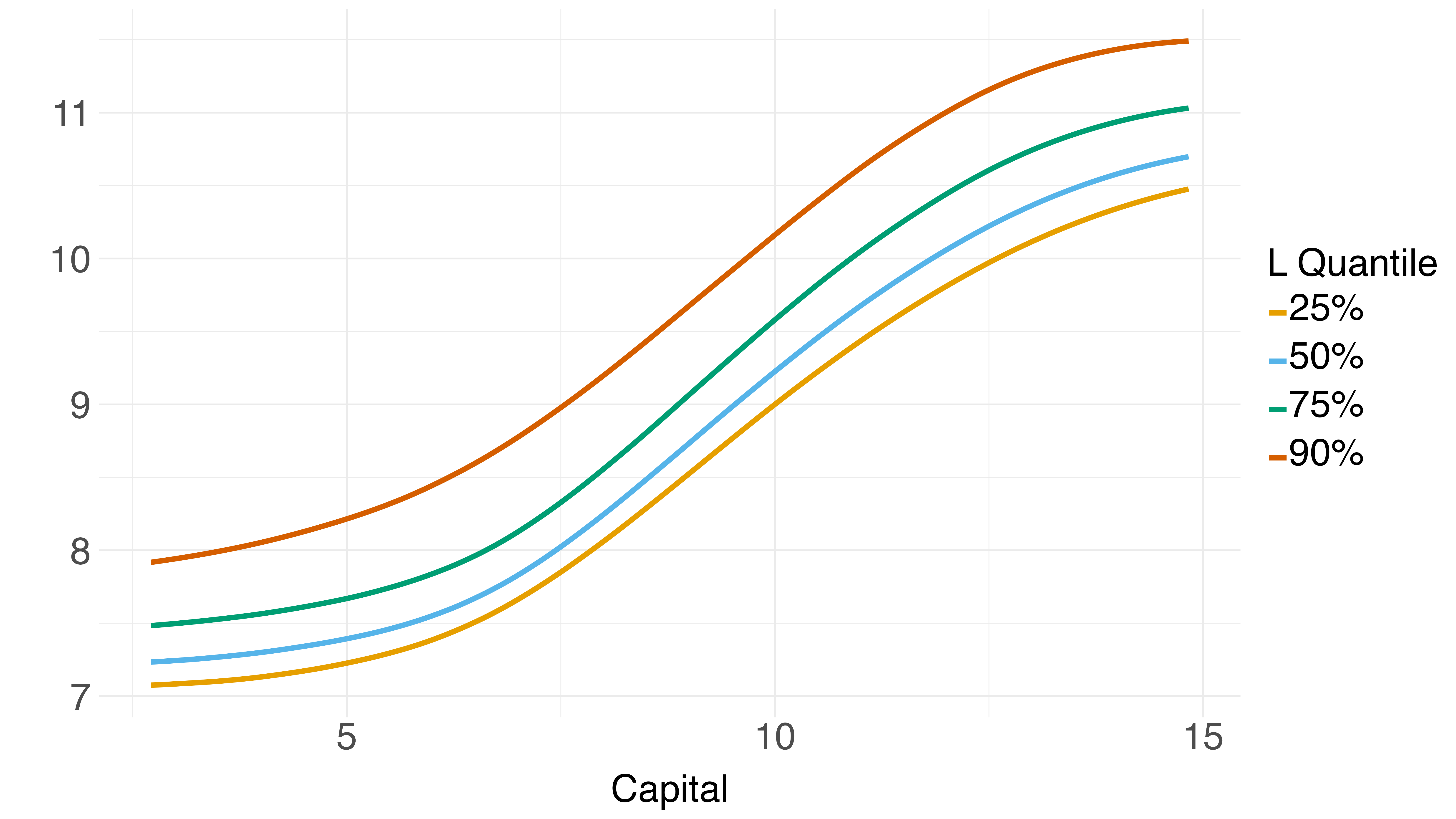}
    \caption{Estimated production function $\widehat F(k,l)$ at selected labor quantiles.}
    \label{varyk}
\end{figure}
\begin{figure}[!htbp]
    \centering
    \includegraphics[width=0.8\textwidth]{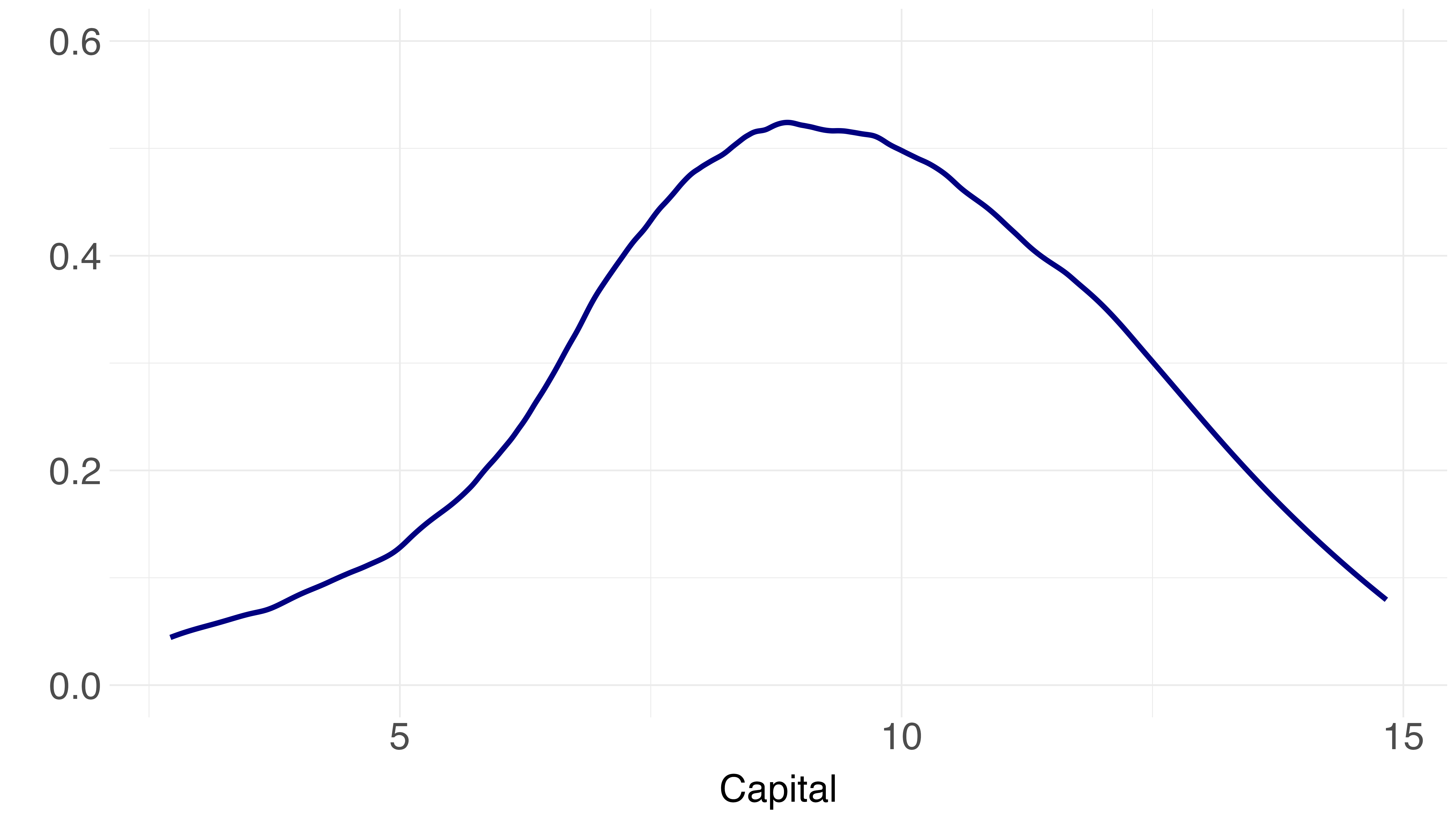}
    \caption{Estimated marginal product $\partial_k \widehat F(k,l)$ at the 0.75 labor quantile, as a function of log capitak $k$, illustrating the classical inverted-U pattern.}
    \label{derivkk}
\end{figure}
Figure~\ref{varyk} shows the posterior mean estimator
$\widehat{F}(k,l) = \E[F(k,l) \mid \mathcal{D}_n]$ as a function of log capital $k$,
with labor fixed at selected quantiles. For each labor quantile, the production function displays the familiar S-shape: convex at low $k$, where additional capital raises productivity at an increasing rate, and concave at higher $k$, where diminishing returns set in. Consequently, the marginal product in Figure~\ref{derivkk} first increases with capital but eventually declines, yielding the classical inverted-U pattern.

\begin{figure}[!htbp]
    \centering
    \includegraphics[width=0.9\textwidth]{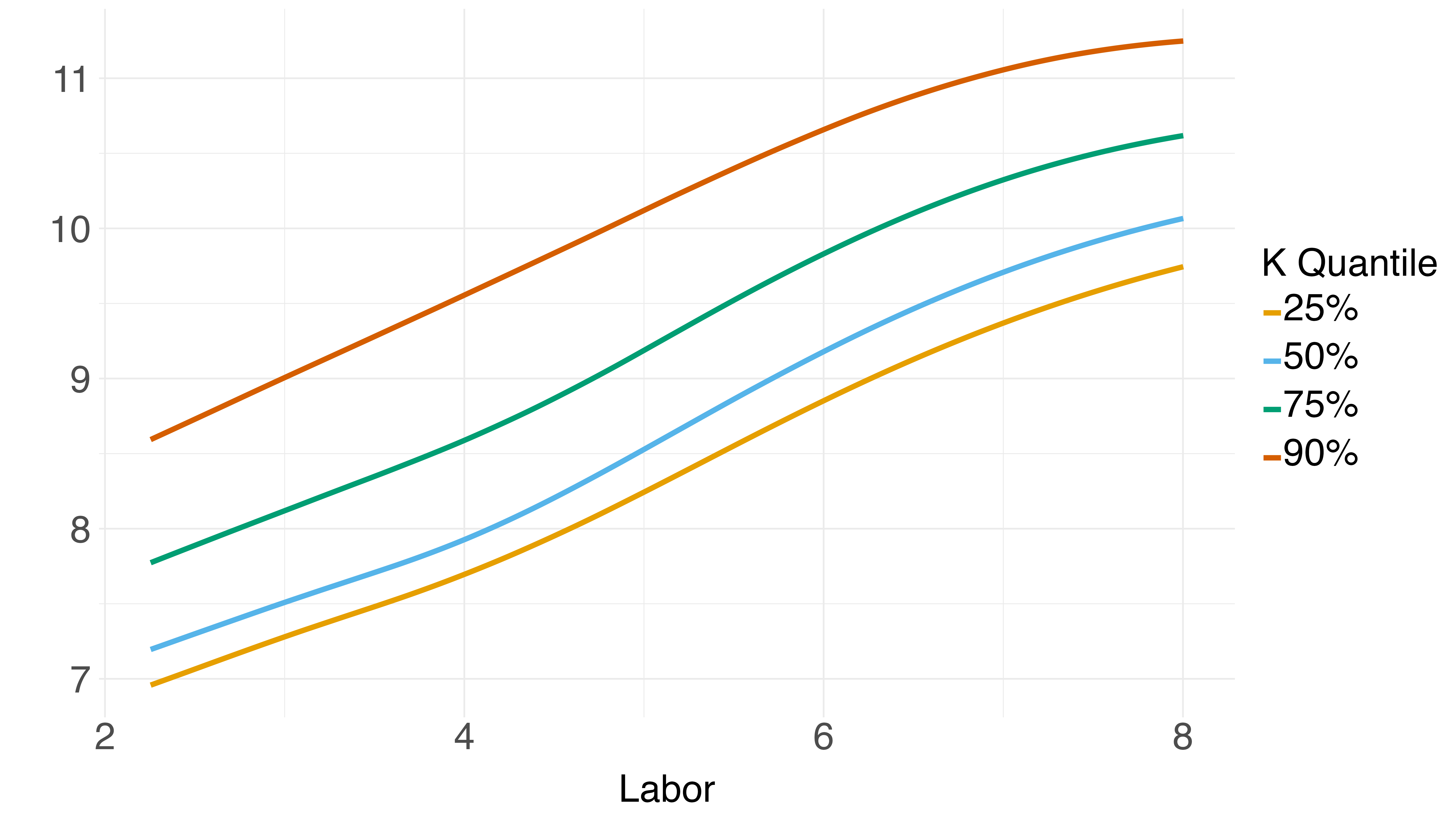}
    \caption{Estimated production function $\widehat F(k,l)$ at select capital quantiles.}
    \label{varyl}
\end{figure}

Figure~\ref{varyl} shows the estimated production function 
$\widehat{F}(k,l)$ as a function of log labor $l$, holding capital fixed at selected quantiles. 
At low to moderate capital quantiles, the function is roughly linear for small values of $l$, 
becomes convex at intermediate levels, and turns concave at higher levels. By contrast, at very high capital quantiles, the function begins at a higher level of output and maintains an almost linear trajectory with a steep slope over most of the range of $l$, turning concave only at higher values. Figure~\ref{mpl-wow} illustrates these patterns via the corresponding marginal product curves.
\begin{figure}[!htbp]
    \centering
    \begin{subfigure}{0.48\textwidth}
        \centering
        \includegraphics[width=\linewidth]{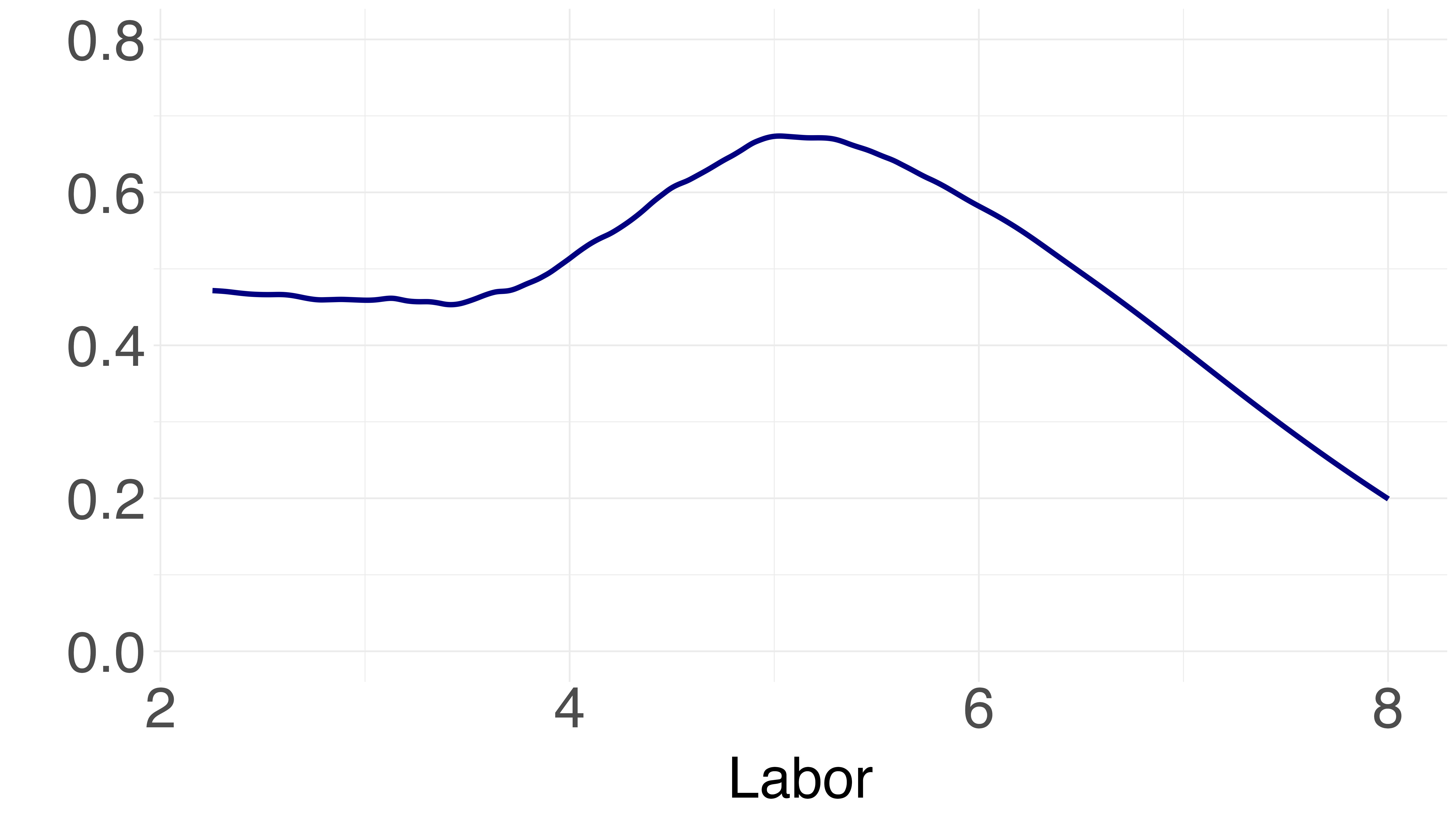}
        \caption{Capital fixed at the 0.75 quantile.}
        \label{derivll1}
    \end{subfigure}
    \hfill
    \begin{subfigure}{0.48\textwidth}
        \centering
        \includegraphics[width=\linewidth]{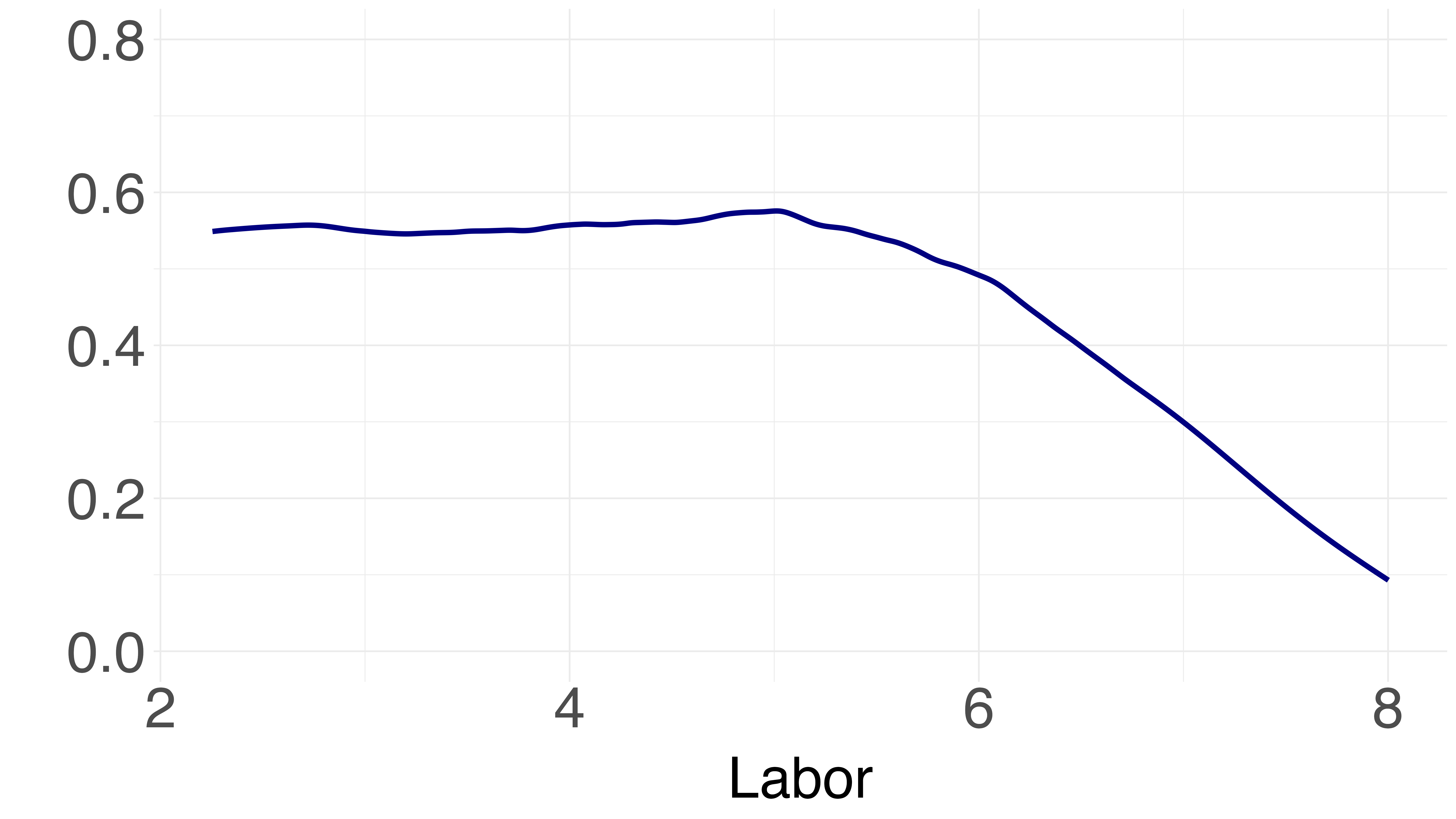}
        \caption{Capital fixed at the 0.90 quantile.}
        \label{derivll2}
    \end{subfigure}
    \caption{Estimated marginal product of labor $\partial_l \widehat{F}(k,l)$ as a function of log labor $l$, with capital fixed at different quantiles.}
    \label{mpl-wow}
\end{figure}

In the data, real capital at the 0.5, 0.75, and 0.95 quantiles equals 740.96, 3656.74, and 24,325.68, respectively, indicating a sharp increase at the upper end of the distribution. One interpretation of these patterns is that they reflect how labor interacts with available capital. With low to moderate capital, complementarities cause output to expand more rapidly as labor increases before diminishing returns set in, yielding convexity followed by concavity. With abundant capital, each worker is already highly productive, so output rises almost linearly with a steep slope in labor until very high levels, where diminishing returns set in.

As a final remark, we note that the identifying restriction for $F(\cdot)$ in (\ref{ackerberg-cmr}) is  complex and highly non-linear. It is therefore noteworthy that our procedures are still able to recover reasonable and meaningful features of $F(\cdot)$ from this restriction alone. To our knowledge, this represents the first fully nonparametric estimate of $F(\cdot)$, obtained without imposing any predetermined parametric structure. Beyond serving as a valuable nonparametric benchmark, these estimates may also provide guidance for the empirical design of approximating parametric specifications. In particular, our findings suggest a preference for specifications that can capture flexible variation in marginal products across input levels.

\section{Conclusion} \label{conclus}
This paper develops a generalized Bayes framework for a broad class of nonparametric conditional moment restriction models. Simulations demonstrate that the proposed procedures are viable and perform well. We expect these methods to be broadly useful, particularly in ill-posed settings or when closed-form solutions are unavailable. As an empirical illustration, we apply the methodology to estimate nonparametric production functions using Chilean plant-level data. We conclude with a few remarks and outline possible extensions.
\subsection{Remarks}
In Section~\ref{motiv-sim}, we motivated quasi-Bayes procedures as an attractive form of data-driven regularization for endogenous nonparametric inverse problems. An additional advantage is in their flexibility to incorporate application specific information. For instance, extending Remark~\ref{loc-remark}, one may specify informative priors centered at a fixed structural function $\widetilde{h}(\cdot)$. In many applications (e.g., \citealp*{adao2017nonparametric, bergquist2020competition}), researchers may have strong microfounded preferences for a parametrically estimated $\widetilde{h}(\cdot)$, yet still wish to accommodate potential misspecification.

As with all nonparametric methods, some degree of finite-sample tuning can often improve performance. In our setting, following Remark \ref{scale-remark}, partial tuning of the Gaussian process covariance hyperparameter $\theta = (\sigma, \ell)$ can be beneficial. When the regressors are normalized, a reasonable default is to set $\ell = 1$ and choose $\sigma$ near the scale of the observables. In nonparametric regression with Gaussian errors, it is standard practice (e.g.  \citealp{williams2006gaussian}) to empirically select $\theta$ by maximizing the Bayesian marginal likelihood. Writing the prior dependence on $\theta$ as $d\mu(h \mid \theta)$, the natural analogue in our framework is to choose $\theta$ by maximizing the quasi-Bayes marginal likelihood:\begin{align*} \mathcal{L}^{}(\theta) = \ \int \exp \bigg( -\frac{n}{2} \E_n \left[ \widehat{m}(W,h)^{\prime} \widehat{\Sigma}(W) \widehat{m}(W,h) \right] \bigg) d\mu(h \mid \theta). \end{align*}
In practice, evaluating this normalizing factor over a large grid can be computationally challenging. An intermediate strategy is to place a weakly informative prior on $\theta$, run a short exploration phase in which we sample from the full posterior over $(h, \theta)$, and then fix $\theta$ at $\hat{\theta}$—the posterior mean computed from the latter part of this exploration phase. Then, proceed with full posterior sampling from the quasi-Bayes posterior $d\mu(h \mid \mathcal{D}_n, \hat{\theta})$. This is the approach we adopt in our implementation. In high-dimensional settings, a common approach for updating $\theta$ during the exploration phase is via slice sampling steps \citep{murray2010slice}.

The first-stage regression in our procedures can use any available source of variation, including both continuous and discrete instruments. Furthermore, there is no requirement that the number of functions in the first stage exceed a fixed threshold. This is in contrast to classical IV 2SLS, which requires at least $K \geq J$ functions in the first stage to estimate a $J$-dimensional second-stage parameter. This flexibility should be particularly valuable in empirical settings where researchers have mixed sources of variation and substantially fewer instruments than endogenous regressors.

We use the same implementation algorithm across all settings considered in this paper, discussed further in Appendix \ref{implem-append}. Briefly, the approach consists of preconditioned Crank–Nicolson (pCN) steps applied to a suitable non-centered parametrization of the Gaussian process sample paths.\footnote{pCN proposals are frequently employed to target infinite-dimensional posteriors that arise in inverse problems with Gaussian process priors \citep{cotter2013mcmc,nickl2023bayesian}.} We view this as an attractive feature, as it suggests that the same algorithm, perhaps with only minor modifications, can be applied broadly.
\subsection{Extensions}
For ease of exposition, we focused on a single structural function $h_0(\cdot)$ that depends on the entire endogenous vector $X$. Adapting the framework to settings with multiple structural functions and restrictions defined on different subcomponents of the observables is straightforward, though notationally more cumbersome.

Our limit theory is developed for a class of infinite-dimensional Gaussian process (GP) priors. Extending the results to other widely used prior classes (e.g., \citealp{chipman2012bart}) or to priors that directly impose specific shape restrictions would be valuable. For GP priors in particular, there is already a substantial literature on enforcing such constraints in regression models (e.g. \citealp{lin2014bayesian}).

Section~\ref{inference} develops, to our knowledge, the first inferential results for a nonparametric quasi-Bayes framework, extending classical parametric GMM results \citep{chernozhukov2003mcmc}. The analysis focused on regular, $\sqrt{n}$-estimable functionals. A natural direction for future work is to broaden the framework to irregular functionals that are slower than $\sqrt{n}$-estimable, similar to the frequentist analysis in \citet{chen2015sieve}.

\bibliographystyle{style}

\bibliography{main.bib}  
\newpage

\appendix

\section{Appendix : Simulation designs} \label{append-sims}
In this section, we describe the simulation designs used in Section \ref{motiv-sim} and \ref{simulations}. We consider multivariate extensions of the designs in \citet{newey2003instrumental}, \citet{santos2012inference}, \citet*{chernozhukov2015constrained}, \citet{chetverikov2017nonparametric}, and \citet*{chen2025adaptive}, which we refer to as \textbf{NP}, \textbf{S}, \textbf{CNS}, \textbf{CW}, and \textbf{CCK}, respectively.

In all designs, the endogenous regressor is five-dimensional, $X \in \mathbb{R}^5$, and the instrument is two-dimensional, $W \in \mathbb{R}^2$. Each multivariate design is constructed as a natural generalization of its univariate counterpart, preserving the underlying endogeneity structure. The structural errors are scaled accordingly to maintain a comparable signal-to-noise ratio. Whenever a covariance matrix $\Sigma$ is not positive definite, it should be interpreted as its projection onto the space of positive definite correlation matrices.
\subsection{NP}
The univariate design in \cite{newey2003instrumental} is given by 
\begin{align*}
\begin{bmatrix} u \\ v \\ w \end{bmatrix}
&\sim \mathcal N \!\left(
\begin{bmatrix} 0 \\ 0 \\ 0 \end{bmatrix},
\begin{bmatrix}
1 & 0.5 & 0 \\
0.5 & 1 & 0 \\
0 & 0 & 1
\end{bmatrix}
\right), &
\begin{aligned}
x &= v + w, \\
h_0(x) &= \log\!\bigl(|x-1|+1\bigr)\,\operatorname{sgn}(x-1), \\
y &= h_0(x) + u
\end{aligned}
\end{align*}
For the multivariate design with $d=5$, we draw $(u,v_1,\dots,v_5)$ and $w=(w_1,w_2)$ as
\[
\begin{bmatrix} u \\ v \end{bmatrix}
\sim \mathcal N\!\left(
\begin{bmatrix} 0 \\ \mathbf 0_{5} \end{bmatrix},
\begin{bmatrix}
1 & \eta\,\mathbf 1_{5}^\top \\
\eta\,\mathbf 1_{5} & I_{5}
\end{bmatrix}
\right),
\qquad
w \sim \mathcal N(0, I_{2})
\]
where $\eta = 0.5$. With round-robin assignment \(\mathrm{map}(j)\in\{1,2\}\) (i.e., \(1,2,1,2,1\)), we set $
x_j = v_j + 0.5\, w_{\mathrm{map}(j)} $ and the structural function is
\[
h_0(x)=\sum_{j=1}^{5} \log(|x_j-1|+1)\,\operatorname{sgn}(x_j-1)
+ \frac{1}{\binom{5}{2}} \sum_{1\le j<k\le 5}\sin(\pi x_j x_k),
\]
and the outcome is $
y = h_0(x) + \sqrt{d}\,u.$
\subsection{CCK}
Let $\Phi( \cdot)$ denote the standard normal CDF. The univariate design in \cite*{chen2025adaptive} is given by
\[
(U,V)^\top \sim \mathcal N\!\left(\mathbf 0_{2},
\begin{bmatrix}
1 & 0.75 \\[3pt]
0.75 & 1
\end{bmatrix}
\right), \qquad
Z \sim \mathcal N(0,1), \qquad
D \sim \mathrm{Bernoulli}(0.5),
\]
\[
X = \Phi\!\big( V + D Z  \big),  \qquad W = \Phi(Z), \qquad
h_0(x) = \sin(4x)\log(x), \qquad
Y = h_0(X) + U.
\]

For the multivariate design with $d=5$, we draw
\[
\begin{bmatrix} U \\ V \end{bmatrix}
\sim \mathcal N\!\left(
\begin{bmatrix} 0 \\ \mathbf 0_{5} \end{bmatrix},
\begin{bmatrix}
1 & \rho\,\mathbf 1_{5}^\top \\[3pt]
\rho\,\mathbf 1_{5} & I_{5}
\end{bmatrix}
\right), \qquad
Z \sim \mathcal N(0, I_{2}),
\]
where $V=(v_1,\dots,v_5)$ and $\rho=0.75$. Set $w=\Phi(Z)\in(0,1)^2$.  
Each regressor is constructed via round-robin instrument assignment $\mathrm{map}(j)\in\{1,2\}$ (i.e.\ $1,2,1,2,1$) and independent switches $D_j\sim\mathrm{Bernoulli}(0.5)$: $
x_j = \Phi\!\bigl(v_j + D_j\,z_{\mathrm{map}(j)}\bigr) $. The structural function is
\[
h_0(x) = 
\sin(4x_1)\,\log(x_1)
\;+\; 1.5\cos(\pi x_2)
\;+\; x_3^2
\;-\; 0.5\,x_4x_5,
\]
and the outcome is $
Y = h_0(x) + \sqrt{d}\,U. $
\subsection{CNS}
We start with the univariate design in \cite*{chernozhukov2015constrained}. We draw latent variables $(X^\ast, Z^\ast, \varepsilon)$ jointly normal,
\[
\begin{bmatrix} X^\ast \\ Z^\ast \\ \varepsilon \end{bmatrix}
\sim \mathcal N\!\left(
\begin{bmatrix} 0 \\ 0 \\ 0 \end{bmatrix},
\begin{bmatrix}
1 & 0.5 & 0.3 \\[3pt]
0.5 & 1 & 0 \\[3pt]
0.3 & 0 & 1
\end{bmatrix}
\right).
\]
Define $x = \Phi(X^\ast)$ and $w = \Phi(Z^\ast)$. The structural function is $
h_0(x) = 1 - 2\,\Phi(x - 0.5) $,
and the outcome is $
Y = h_0(x) + \varepsilon.$
For the multivariate design with $d=5$, we draw \begin{align*}
\begin{bmatrix}
X^\ast \\ Z^\ast \\ \varepsilon
\end{bmatrix}
&\sim \mathcal N\!\left(
\mathbf 0_{\,d+3},\; \Sigma
\right), {\qquad} 
\begin{aligned}
\operatorname{Cov}(X^\ast_j,Z^\ast_1) &= \rho_1 \quad (j=1,2,3),\\
\operatorname{Cov}(X^\ast_j,Z^\ast_2) &= \rho_2 \quad (j=4,5),\\
\operatorname{Cov}(X^\ast_j,\varepsilon) &= \eta \quad (j=1,\dots,5).
\end{aligned}
\end{align*}
and all other covariances equal to $0$.
Here $\rho_1=\rho_2=0.5$ and $\eta=0.3$. We set $
x = \Phi(X^\ast) \in (0,1)^5 $ and
$w = \Phi(Z^\ast) \in (0,1)^2 $. The structural function is
$
h_0(x) = \sum_{j=1}^5 \bigl(1 - 2\,\Phi(x_j - 0.5)\bigr),
$ and the outcome is
$
Y = h_0(x) + \sqrt{d}\,\varepsilon.
$.

\subsection{CW}
We start with the univariate design in \cite{chetverikov2017nonparametric}. Fix parameters $\sigma>0$, $\rho\in(-1,1)$, and $\eta\in(-1,1)$.
Let $\zeta,\varepsilon,\nu \sim \mathcal N(0,1)$ be independent.
Define
\[
w = \Phi(\zeta), 
\qquad 
x = \Phi\!\bigl(\rho \zeta + \sqrt{1-\rho^2}\,\varepsilon\bigr),
\qquad
\epsilon = \sigma\bigl(\eta \varepsilon + \sqrt{1-\eta^2}\,\nu \bigr).
\]
The structural function is $
h_0(x) = 2\,(x-0.5)_+^2 + 0.5\,x, $ and the outcome is $
Y = h_0(x) + \epsilon$. This design uses $\sigma=0.5$, $\rho=0.3$, and $\eta=0.3$. For the multivariate version with $d=5$, fix $\sigma>0$, $\rho_1,\rho_2\in(-1,1)$, and $\eta\in(-1,1)$.
Let $\zeta=(\zeta_1,\zeta_2)^\top\sim\mathcal N(0,I_2)$, $\nu\sim\mathcal N(0,1)$, and
$\varepsilon_x=(\varepsilon_{x1},\dots,\varepsilon_{xd})^\top\sim\mathcal N(0,I_d)$.
Set the instruments and regressors
\[
w=\Phi(\zeta)\in(0,1)^2,\qquad
x_j=
\begin{cases}
\Phi\!\bigl(\rho_1\,\zeta_1+\sqrt{1-\rho_1^2}\,\varepsilon_{xj}\bigr), & j=1,2,3,\\[4pt]
\Phi\!\bigl(\rho_2\,\zeta_2+\sqrt{1-\rho_2^2}\,\varepsilon_{xj}\bigr), & j=4,5,
\end{cases}
\]
Define the composite error $
\epsilon=\sigma\Bigl(\eta\sum_{j=1}^d \varepsilon_{xj}+\sqrt{1-\eta^2}\,\nu\Bigr) $
and the structural function
\[
h_0(x)=\sum_{j=1}^d\!\Bigl(2\,(x_j-0.5)_+^2+0.5x_j\Bigr)
\;+\; x_3x_4 \;+\; \log\!\bigl(1+x_1x_2x_5\bigr),
\]
The outcome is $
Y = h_0(x) + \sqrt{d}\,\epsilon.$ The design uses $\sigma=1$, $\rho_1=\rho_2=0.3$, $\eta=0.3$.

\subsection{S}
We start with the univariate design in \cite{santos2012inference}. 
\begin{align*}
\begin{bmatrix} X^\ast \\ Z^\ast \\ \varepsilon^\ast \end{bmatrix}
&\sim \mathcal N\!\left(
\begin{bmatrix} 0 \\ 0 \\ 0 \end{bmatrix},
\begin{bmatrix}
1 & 0.5 & 0.3 \\[3pt]
0.5 & 1 & 0 \\[3pt]
0.3 & 0 & 1
\end{bmatrix}
\right), \quad
\begin{aligned}
x &= 2\bigl(\Phi(X^\ast/3)-0.5\bigr), \\
w &= 2\bigl(\Phi(Z^\ast/3)-0.5\bigr), \\
\epsilon &= \varepsilon^\ast
\end{aligned}
\end{align*}
The structural function is $
h_0(x)=2\sin(\pi x) $, and the outcome is $
Y = h_0(x) + \epsilon. $

 For the multivariate design with $d=5$, let the latent vector $
(X_1^\ast,\dots,X_d^\ast, Z_1^\ast,Z_2^\ast,\varepsilon)^\top
\sim \mathcal N(0,\Sigma),$
where $\Sigma$ is defined by $
\operatorname{Cov}(X_j^\ast, Z^\ast_{\mathrm{map}(j)}) = \rho, \; \;
\operatorname{Cov}(X_j^\ast,\varepsilon) = \eta $ with all other covariances zero, and $\mathrm{map}(j)\in\{1,2\}$ is the round-robin assignment $(1,2,1,2,1)$. We set $\rho=0.5$ and $\eta=0.5$

 Let $
x_j = 2\bigl(\Phi(X_j^\ast/3)-0.5\bigr), \; \; \;
w_k = 2\bigl(\Phi(Z_k^\ast/3)-0.5\bigr).$ The structural function is $$
h_0(x) = \sin(\pi x_1) + 0.5\sin\!\bigl(\pi(x_3-x_2)\bigr) + 0.5\cos\!\bigl(\pi(x_5-x_4)\bigr) $$
and the outcome is $
Y = h_0(x) + \sqrt{d}\,\varepsilon $.

\section{Appendix : Implementation} \label{implem-append}
Let $X_i = (X_{i1},\dots,X_{id})^\top \in \mathbb{R}^d$ denote the observed regressors. 
For each coordinate $j$, define $
\widehat{u}_{n,j} = \E_n[X_j]$ and $
\widehat{\sigma}_j = \sqrt{\Var_n(X_j)}.$ Denote the ``normalized" grid by:
\[
\mathcal{X}_n
= \left\{ 
\left( 
\frac{X_{i1}-\widehat{u}_{n,1}}{\widehat{\sigma}_1},\; \dots,\; 
\frac{X_{id}-\widehat{u}_{n,d}}{\widehat{\sigma}_d} 
\right)^\top : i=1,\dots,n 
\right\}.
\]
Let \(G\) denote the Gaussian process arising from the prior \(d\mu(\cdot)\). 
For posterior computation, it suffices to work with the finite-dimensional vector $ \mathcal{G} = \{ G(x) : x \in \mathcal{X}_n \} $, as the likelihood depends on \(G\) only 
through its evaluations at the design points. Given $\sigma > 0$ and $\ell \in \mathbb{R}_+^d$, define the scaled process $$ G_{\theta} = \sigma G(\ell^{-1} x). $$
Here, $\theta = (\sigma, \ell)$, where $\sigma \in \mathbb{R}_{+}$ denotes the signal variance and $\ell \in \mathbb{R}_{+}^d$ the length-scale parameter. Intuitively, $\sigma$ controls the vertical scale of the process, while $\ell$ controls the rate at which correlations decay with distance. In multivariate settings, the length scales can also be interpreted as measures of the regressors relative importance in modeling the structural function. The theoretical properties of $G_{\theta}$, for any fixed $\theta$, are similar to those of the base process. If the regressors are normalized, a reasonable choice is to set  $\ell=1$ and fix $\sigma$ near the scale of the response. In practice, these hyperparameters are often partially tuned using the observed data. For example, in Gaussian regression, $\theta$ is typically selected by maximizing the Bayesian marginal likelihood \citep{williams2006gaussian}.

We work with normalized regressors in all settings. As an alternative to tuning $\theta$ via the quasi-Bayes marginal likelihood (as discussed in Section \ref{conclus}), we place independent $\mathrm{LogNormal}(0,1)$ priors on $\sigma$ and each coordinate of $\ell = (\ell_1,\dots,\ell_d)$. The hierarchical posterior is then sampled during an exploration phase of $k=10{,}000$ iterations, targeting an acceptance rate of $0.25$ across all parameters. The posterior mean $\widehat{\theta}$ is computed from the second half of the draws, after which we perform full posterior sampling from the quasi-Bayes posterior $d\mu (h \mid \widehat{\theta},\mathcal{D}_n)$.

Details on the posterior sampling scheme are as follows. 
We represent \(G_{\theta}\), viewed as a process on \(\mathcal{X}_n\), 
in its non-centered parametrization:
\[
G_{\theta} = \sigma L_{\ell} z,
\]
where \(L_{\ell}\) is the \(n \times n\) Cholesky matrix
(depending on the length-scale parameter \(\ell\)), and 
\(z \sim N(0, I_n)\) is a standard Gaussian vector. 
The parameters \(\sigma\) and \(L_{\ell}\) are updated using standard Metropolis steps, 
while \(z\) is updated using preconditioned Crank--Nicolson (pCN) proposals \citep{cotter2013mcmc,nickl2023bayesian}. Once we obtain posterior samples from the quasi-Bayes posterior $d\mu (h \mid \widehat{\theta},\mathcal{D}_n)$, the value of the process at any  $x \notin \mathcal{X}_n $ is computed using the standard Gaussian kriging interpolation formula (see, e.g. \citealp{ghosal2017fundamentals}).

In settings where $|\mathcal{X}_n|$ is very large, recomputing $L_{\ell}$ at each new proposal of $\ell$ in the Markov chain can be computationally expensive during the exploration phase. There are a variety of methods to deal with this, but a simple and widely used approach is to employ a sparse GP approximation by defining the process over a smaller set of inducing points $\mathcal{Z}_n$, with $|\mathcal{Z}_n| \ll |\mathcal{X}_n|$. The value of the process at any $x \notin \mathcal{Z}_n$ can be then be efficiently computed using the kriging interpolation formula. A popular strategy is to select $\mathcal{Z}_n$ using $k$-means clustering on $\mathcal{X}_n$. Once the hyperparameters $\widehat{\theta} = (\widehat{\sigma}, \widehat{\ell})$ have been estimated, full posterior sampling can then be performed directly on $\mathcal{X}_n$, since $L_{\widehat{\ell}}$ is fixed and no longer needs to be recomputed.

\subsection{Simulations}
All simulations use a Whittle–Matérn Gaussian process with regularity $\alpha = 3/2$. The hyperparameter $\hat{\theta}$ is computed using the full grid $\mathcal{X}_n$. The first stage is computed using thin-plate regression splines \citep{wood2003thin} of dimension $K$. For univariate designs we set $K \in \{5,7,10 \}$, while for multivariate designs we use $K=15$. 

\subsection{Empirics}
The empirical application in Section 6 employs a Whittle--Matérn Gaussian process with regularity $\alpha = 3/2$. The hyperparameter $\widehat{\theta}$ is computed using $k$-means clustering to select 2000 inducing points from the set of unique $(l,k)$ pairs in the data. Both the first stage and the conditional mean function $\widehat{\Phi}_t(\cdot)$ are estimated using thin-plate regression splines with dimension $K = 15$. For any input 
function \(\tilde{F}\), define the estimated residual:
\[
\widehat{\omega}_{i,t}(\tilde{F}) 
= \widehat{\Phi}_{t}(l_{it}, k_{it}, m_{it}) - \tilde{F}(l_{it}, k_{it}).
\]
The output of the univariate conditional mean \(g(\cdot)\) is obtained by 
regressing \(\widehat{\omega}_{i,t}(\tilde{F})\) on 
\(\widehat{\omega}_{i,t-1}(\tilde{F})\). The conventional approach \citep*{ackerberg2015identification,gandhi2020identification} is to specify this regression as either an autoregression or a low-degree polynomial. We follow this strategy, employing a second-degree polynomial specification. Note that this regression is performed separately for each function proposal \(\tilde{F}\).

\section{Appendix : General Theory} \label{gentheory}
This section develops a generic contraction result that will later be applied in the derivation of our main results.

\subsection{Assumptions} \label{main-ass}
We state and  discuss the assumptions that we impose on the model and prior. Throughout this section, let $\mu_n$ denote a, possibly data dependent, prior that is supported on a class of functions $\mathcal{H}_n$. Let $(\epsilon_n)_{n=1}^{\infty} $ denote a deterministic sequence of positive constants that converge to zero at a slower than parametric rate : $\epsilon_n \downarrow 0 $ and $n \epsilon_n^2 \uparrow \infty$. 

\begin{assumption}[Sampling Uncertainty] \label{sampling-uncert}
There exists a deterministic (possibly sample size $n$ dependent) function $ \widetilde{m}(W,h) $, a set $\mathcal{S}_n \subseteq \mathcal{H}_n$ and a universal constant $D > 0 $, such that \begin{align*}
     \mathbb{P} \bigg( \sup_{h \in \mathcal{S}_n} \left|  \E_n \big( \| \widehat{m}(W,h)  \|_{\ell^2}^2 \big)   - \E \big( \| \widetilde{m} (W,h)  \|_{\ell^2}^2  \big)  \right|  > D \epsilon_n^2  \bigg) \rightarrow 0 .
\end{align*}
\end{assumption}
Assumption \ref{sampling-uncert} provides bounds on the sampling uncertainty arising from the fact that the true population distribution of $ \mathcal{D} = (Y,X,W) $ is unknown. Typically, \(\widetilde{m}\) is a suitable population analog of \(\widehat{m}\). For instance, with a first stage sieve estimator as in Section \ref{first-stage}, it is natural to set $\widetilde{m}(W,h) = \Pi_K [ m(W,h)]$ where $\Pi_K$ is a population projection operator.\footnote{Denote by $\mathcal{V}_K$, the linear space spanned by the basis functions $ \{b_1(W),\dots,b_K(W)  \} $. Then $\Pi_K (.)  $ is the $L^2(\mathbb{P})$ orthogonal projection onto $\mathcal{V}_K$. } 

The \(\mathcal{S}_n\) typically represents a ball (in an suitable metric) that is centered around a fixed function \( h_n \). The verification of Assumption \ref{sampling-uncert} then largely reduces to applying suitable empirical process techniques to control the deviation of the empirical mean from the population expectation. In some cases, the set \(\mathcal{S}_n\) also includes certain Sobolev-type norm constraints, which aid in controling the sampling uncertainty when $m(W,h)$ is highly nonlinear in $h$.
\begin{assumption}[Weak Bias] \label{weak-bias} Let $\widetilde{m}(.)$ be as in Assumption \ref{sampling-uncert}. For some function $h_n \in \mathcal{H}_n$ and a universal constant $D > 0 $, we have  \begin{align*} & (i) \; \; \;
   \E \big( \| \widetilde{m}(W,h_n) - m(W,h_n)    \|_{\ell^2}^2 \big) \leq D \epsilon_n^2\;, \\ & (ii) \; \;  \E \big( \| m(W,h_n) - m(W,h_0)    \|_{\ell^2}^2 \big) \leq D \epsilon_n^2.
\end{align*}
\end{assumption}
Assumption \ref{weak-bias} imposes bounds on the bias between \(\widetilde{m}\) and \(m\) at the fixed choice \(h_n\), as well as the bias between \(h_n\) and \(h_0\) with respect to the weak metric $$
d_{w}^2(h_0,h_n) = \mathbb{E} \big( \| m(W,h_n) - m(W,h_0) \|_{\ell^2}^2 \big). $$ In some settings, it is natural to set \( h_n = h_0 \) if the true structural function \( h_0 \) is already in the support of the prior. This will be the case when we specialize to Gaussian process priors in  Section \ref{gp-prior}.
\begin{assumption}[Local Concentration] \label{loc-conc}
Let $\widetilde{m}(W,h)$ and $\mathcal{S}_n$  be as in Assumption \ref{sampling-uncert}. For some set $\mathcal{R}_n \supseteq \mathcal{S}_n$, we have  \begin{align*}
    & (i) \; \; \; \; \; \; \mu_n(h \in \mathcal{R}_n   ) \geq c \exp(- C' n \epsilon_n^2) \\ & (ii) \; \; \; \; \; \mu_n(h \in  \mathcal{R}_n \setminus \mathcal{S}_n   ) \leq C \exp( - B n \epsilon_n^2  ) \\ & (iii) \; \; \;   \sup_{h \in \mathcal{S}_n} \E \big( \|  \widetilde{m}(W,h) - \widetilde{m}(W,h_n)    \|_{\ell^2}^2 \big) \leq D \epsilon_n^2. 
\end{align*}
where $c,C,C',B,D > 0$ are universal constants with $B > C'$.
\end{assumption}
The set \( \mathcal{R}_n \) in Assumption \ref{loc-conc} is introduced to provide some flexibility when direct verification of a local concentration bound is challenging for the $\mathcal{S}_n$ in Assumption \ref{sampling-uncert}. In such cases, $\mathcal{R}_n$ relaxes certain restrictions (e.g. Sobolev norm constraints) imposed on  \( \mathcal{S}_n \). Assumption \ref{loc-conc}\((ii)\) further requires that the subset of \( \mathcal{R}_n \) where these restrictions fail to hold is sufficiently negligible. Typically, \( \mathcal{R}_n \) is a small ball (in a suitable metric) around $h_n$. Assumption \ref{loc-conc}\((i)\) then imposes a standard small ball local concentration condition on the prior.

\subsection{Results} \label{main-res}
In this section, we verify that the quasi-Bayes posterior in (\ref{qb-general}) asymptotically concentrates on local neighborhoods of the structural function. 

Given a vector-valued function $g(W)$ and a positive semi-definite weighting matrix $\Sigma(W)$, we define the weighted empirical mean square norm by $ \|  g(W)  \|_{L^2(\mathbb{P}_n,{\Sigma})} = \sqrt{\E_n\big[ g(W)' \Sigma(W) g(W)   \big]} $. We use this norm to induce a first stage  weak metric on structural functions via \begin{align}
    \label{weak-metric} d_{w,\mathbb{P}_n}(h,h_0) = \| \widehat{m}(W,h)  - m(W,h_0)  \|_{L^2(\mathbb{P}_n, \widehat{\Sigma})}.
\end{align}

\begin{theorem}[Weak Contraction]
\label{main-contract} Suppose $\mathbb{P}(\lambda_{\max}(\widehat{\Sigma}(W)) \leq D) \rightarrow 1$ for some universal constant $D > 0$. If Assumptions \ref{sampling-uncert}-\ref{loc-conc} hold with a sequence $\epsilon_n \rightarrow 0$, then there exists a universal constant $L > 0$ such that \begin{equation}
    \label{contraction-main}  \mu_n \big ( h :  \| \widehat{m}(W,h)  - m(W,h_0)  \|_{L^2(\mathbb{P}_n,\widehat{\Sigma})}  > L \epsilon_n  \: \big| \: \mathcal{D}_n  \big ) \xrightarrow{\mathbb{P}} 0.
\end{equation}
\end{theorem}
Theorem \ref{main-contract} establishes contraction with the respect to the  weak  metric $d_{w} (h,h_0)  $. The interpretation of this convergence varies from model to model, but in general, it is meant to be interpreted as a preliminary contraction that can then be subsequently used to deduce results in a stronger metric. In particular, if (\ref{contraction-main}) holds and the bulk of the posterior mass is contained in a well-behaved subset, it is often possible to deduce results in a stronger metric like   $ d(h,h_0) = \|  h- h_0  \|_{L^2} $. To fix ideas, given a metric $d(.)$ and a class of functions $\mathcal{G}_n \subseteq \mathcal{H}_n$, we define the modulus of continuity by \begin{equation*}  \omega_n( d, \mathcal{G}_n, \epsilon) = \sup \{ d(h,h_0) :  h \in \mathcal{G}_n, \: \| \widehat{m}(W,h)  - m(W,h_0)  \|_{L^2(\mathbb{P}_n, \widehat{\Sigma})} \leq \epsilon          \}  .    \end{equation*}
The modulus of continuity is frequently used to characterize the convergence rate in inverse problems (see e.g. \citealp{chen2012estimation}; \citealp{knapik2018general}). The following result is a straightforward consequence of Theorem \ref{main-contract}.
\begin{corollary}[Contraction]
\label{contract-strong}
    Suppose the hypothesis of Theorem \ref{main-contract} holds. Let $\mathcal{G}_n$ be any subset of functions for which  $$ \mu_n(h \notin \mathcal{G}_n : \| \widehat{m}(W,h)  - m(W,h_0)  \|_{L^2(\mathbb{P}_n)} \leq L \epsilon_n )  \leq  C \exp(-D' n \epsilon_n^2) $$ 
 holds for some $C > 0$ and a sufficiently large $D' > 0$. Then  \begin{equation}
       \mu_n \big( h \in \mathcal{G}_n : d(h,h_0) \leq \omega_n(d,\mathcal{G}_n,L\epsilon_n  ) \: \big| \:  \mathcal{D}_n   \big) \xrightarrow{\mathbb{P}} 1.
    \end{equation}
\end{corollary}
Corollary \ref{contract-strong} provides contraction rates in terms of the modulus $\omega_n(d,\mathcal{G}_n,L \epsilon_n)$. The constant $D'$, which regulates the decay of mass on $\mathcal{G}_n^c$, is required to be larger than some of the preceding constants that appear in Assumption \ref{sampling-uncert} - \ref{loc-conc}. Usually, the set $\mathcal{G}_n$ is chosen as a function of $D'$ so as to ensure the desired bound holds trivially.

\section{Appendix : Proofs} \label{proofs}
We denote by $\widehat{G}_{b,K}^{o}$ the matrix \begin{align}
\label{gbko}  \widehat{G}_{b,K}^{o} =  G_{b,K}^{-1/2} \widehat{G}_{b,K} G_{b,K}^{-1/2} \; \end{align}

In this section, we provide proofs for all the main results.

\begin{lemma}
\label{aux1}
Suppose Condition \ref{fsbasis}(i) holds. Then, for every sieve dimension $K$ and $t > 0 $, we have that
 $$  \mathbb P \left( \|\widehat{G}_{b,K}^{o}- I_{K}  \|_{op} > t \right)  \leq  2 K \exp \bigg( - \frac{t^2/2}{\zeta^2_{b,K}/n + 2 \zeta^2_{b,K} t/(3n) }      \bigg).   $$
 
\end{lemma}

\begin{proof}[Proof of Lemma \ref{aux1}]
Observe that
$$ \widehat{G}_{b,K}^o - I_K  = n^{-1} \sum_{i=1}^n G_{b,K}^{- 1/2} \big \{ b^K (W_i) b^K(W_i)' - \E [b^K(W) b^K(W)' ]       \big \} G_{b,K}^{- 1/2} = \sum_{i=1}^n \Xi_{i}  \; , $$
where  $(\Xi_i)_{i=1}^n$ are i.i.d matrices of dimension $K \times K$. Furthermore, we have that  \begin{align*}
& \| \Xi_i  \|_{op} \leq 2 n^{-1} \zeta_{b,K}^2 \; , \\ & \| \E[\Xi_i \Xi_i']   \|_{op} \leq n^{-2} \| \E [ G_{b,K}^{- 1/2} b^K(W)  b^K(W)' G_{b,K}^{- 1/2}    ]    \|_{op} = n^{-2} \| I_K \|_{op} = n^{-2} \; , \\ & \| \E[\Xi_i' \Xi_i]   \|_{op} \leq n^{-2} \big | \E[ b^K(W)' G_{b,K}^{-1} b^K(W)  ]    \big| \leq n^{-2} \zeta_{b,K}^2 .
\end{align*}
The claim follows from using these bounds in an application of \citep[Theorem 1.6]{tropp}.
\end{proof}

\begin{lemma}
\label{aux2}
 Suppose Condition \ref{fsbasis}(i) holds. Let $\bar{K}_{\max}  = \bar{K}_{\max,n} $ denote a sequence  that satisfies $\bar{K}_{\max} \uparrow \infty$ and $\bar{K}_{\max}  \log( \bar{K}_{\max} ) /n \downarrow 0$. Then, there exists a universal constant $D < \infty$ such that $$  \mathbb{P} \bigg( \sup_{K \in \mathbb{N} :K \leq \bar{K}_{\max}} \| \widehat{G}_{b,K}^o- I_K \|_{op} \leq    D  \frac{\sqrt{  \bar{K}_{\max} } \sqrt{\log  \bar{K}_{\max} }}{\sqrt{n}}    \bigg) \rightarrow 1 . $$

\end{lemma}
\begin{proof}[Proof of Lemma \ref{aux2}]
Lemma \ref{aux1} and a union bound yields \begin{align*}   \mathbb{P} \bigg(\sup_{K \in \mathbb{N} :K \leq  \bar{K}_{\max} } \|\widehat  G_{b,K}^o - I_K  \|_{op} > t   \bigg) & \leq \sum_{K \in \mathbb{N} :K \leq  \bar{K}_{\max}}   \mathbb{P} \left(  \|\widehat  G_{b,K}^o - I\|_{op} > t  \right) \\ & \leq   2 \sum_{K \in \mathbb{N} :K \leq  \bar{K}_{\max}}  K \exp\left\{ - \frac{t^2 /2}{\zeta_{b,K}^2(1+ 2 t  /3)n^{-1}} \right\}  .
\end{align*}

Let $L > 0 $ be such that $\zeta^2_{b,K} \leq L K $ for all $K$ and fix any $D  > \sqrt{8 L} $. Define $t = t_n = D  \sqrt{  \bar{K}_{\max}   \log \bar{K}_{\max}  }/ \sqrt{n}$. Since $t_n \downarrow 0$, there exists $N \in \mathbb{N}$ such that  $2 t_n /3 \leq 1   $ for all $  n > N$. For $n > N$, it follows that \begin{align*}   \sum_{K \in \mathbb{N} :K \leq  \bar{K}_{\max}}  K \exp\left\{ - \frac{t_n^2 /2}{\zeta_{b,K}^2(1+ 2 t_n  /3)/n} \right\} & \leq \bar{K}_{\max}^2 \exp \bigg \{  -  \frac{D^2  \log (\bar{K}_{\max}) }{4L}    \bigg \} \\ & = \exp \bigg \{  \bigg(2- \frac{D^2}{4L} \bigg)  \log(\bar{K}_{\max})  \bigg \}  \\ & \rightarrow 0.
\end{align*}
\end{proof}

\begin{lemma}
\label{emp-b}
Suppose Conditions \ref{residuals}, \ref{residuals2}(i) and \ref{fsbasis}(i) hold. For each fixed $l \in \{ 1 ,  \dots , d_{\rho} \}$ and function $h: \mathcal{X} \rightarrow \R $, define
 \begin{align*} & R_{h,l}^K(Z) =     \big[ G_{b,K}^{-1/2}  b^K(W) \big ]   \rho_{l}(Y,h_{}(X))  .  \end{align*}
Given any $M > 0 $, there exists a universal constant $D = D(M) < \infty$ such that
 \begin{equation}  \sup_{l \in \{1 , \dots ,  d_{\rho} \} }  \E\bigg(     \sup_{h  \in \mathbf{H}^t(M)   }   \|   \E_n [ R_{h,l}^K(Z) ] - \E [ R_{h,l}^K(Z) ]        \|_{\ell^2}  \bigg) \leq D  \frac{\sqrt{K}}{\sqrt{n}} \label{bern1}   \end{equation}
holds for every $K $.
\end{lemma}
\begin{proof}[Proof of Lemma \ref{emp-b}]

It suffices to verify that (\ref{bern1}) holds for each  $l \in \{1 , \dots , d_{\rho} \}$.  Fix any such $l$.  For ease of notation, we suppress the dependence on $l$ and denote the associated vector by $   R_{h,l}^K(Z) = R_{h}^K(Z)$.  Denote the $j  \in \{ 1 , \dots , K \}$ element of  $R_{h}^K(Z)$ by $[R_{h}^K(Z)]_{j} =   \big[ G_{b,K}^{-1/2}  b^K(W) \big ]_{j}   \rho_{l}(Y,h_{}(X))   $. Observe that
 \begin{align*} &
\E_{} \bigg[ \sup_{h \in \mathbf{H}^t(M)  }  \|   \E_n [R_{h}^K(Z)] - \E_{}[R_{h}^K(Z)]       \|_{\ell^2}^2  \bigg] \\ & =  \frac{1}{n} \E \bigg[   \sup_{h \in \mathbf{H}^t(M)}  \sum_{j=1}^K \left|   \frac{1}{\sqrt{n}} \sum_{i=1}^n\big \{  [R_{h}^K(Z_i)]_{j} -\E_{} ( [R_{h}^K(Z)]_{j} ) \:  \big \}  \right|^2          \:  \bigg]   \\ & \leq  \frac{1}{n} \sum_{j=1}^K \E \bigg[ \sup_{h \in \mathbf{H}^t(M)}   \left|   \frac{1}{\sqrt{n}} \sum_{i=1}^n\big \{  [R_{h}^K(Z_i)]_{j} -\E_{} ( [R_{h}^K(Z)]_{j} ) \:  \big \}  \right|^2      \;  \bigg] \\ & \leq \frac{K}{n} \sup_{j \in \{1 , \dots , K \}}  \E \bigg[ \sup_{h \in \mathbf{H}^t(M)}   \left|   \frac{1}{\sqrt{n}} \sum_{i=1}^n\big \{  [R_{h}^K(Z_i)]_{j} -\E_{} ( [R_{h}^K(Z)]_{j} ) \:  \big \}  \right|^2      \;  \bigg] .
\end{align*}
It suffices to verify that the expectations are uniformly bounded.  Fix any such $j$.  We view  the expectation as a higher moment of an empirical process over the class of functions $$\mathcal{F}  = \{ [R_h^K(Z)]_{j}   : h \in \mathbf{H}^t(M)    \}    . $$
Let $F(Z) =  \sup_{f \in \mathcal{F}} \left| f(Z) \right| $ denote the envelope of $\mathcal{F}$.  Let $C_2(M)$ be as in Condition \ref{residuals2}(i).  By Condition \ref{residuals2}(i) and the observation that $\big[ G_{b,K}^{-1/2}  b^K(W) \big ]_{j}$ has unit $L^2(\mathbb{P})$ norm (by the definition of $G_{b,K}$), the envelope admits the bound
 
  \begin{align*}
\|  F  \|_{L^2(\mathbb{P})}^2 & = \bigg \|    \sup_{h \in \mathbf{H}^t(M)}   \big[ G_{b,K}^{-1/2}  b^K(W) \big ]_{j}   \rho_{l}(Y,h_{}(X))       \bigg \|_{L^2(\mathbb{P})}^2 \\ & \leq  \E \bigg[  \left|   \big[ G_{b,K}^{-1/2}  b^K(W) \big ]_{j}  \right| ^2  \E \bigg[    \sup_{h \in \mathbf{H}^t(M)}                \left| \rho_{l}(Y,h_{}(X)) \right|^2  \bigg| W \bigg]   \bigg] \\ & \leq C_2^2 \E \bigg[  \left|   \big[ G_{b,K}^{-1/2}  b^K(W) \big ]_{j}  \right| ^2   \bigg] \\ & = C_2^2.
 \end{align*}
By an application of \citep[Theorem 2.14.5]{weakc},  there exists a universal constant $ D > 0 $ such that \begin{align*}
 & \E \bigg[ \sup_{h \in \mathbf{H}^t(M)}   \left|   \frac{1}{\sqrt{n}} \sum_{i=1}^n\big \{  [R_{h}^K(Z_i)]_{j} -\E_{} ( [R_{h}^K(Z)]_{j} ) \:  \big \}  \right|^2      \;  \bigg] \\ & \leq D \bigg(   \E \bigg[ \sup_{h \in \mathbf{H}^t(M)}   \left|   \frac{1}{\sqrt{n}} \sum_{i=1}^n\big \{  [R_{h}^K(Z_i)]_{j} -\E_{} ( [R_{h}^K(Z)]_{j} ) \:  \big \}  \right|      \;  \bigg]    + C_2   \bigg)^2 .
\end{align*}

 By an application of \citep[Theorem 3.5.13]{gine2021mathematical},  there exists a universal constant $ D > 0 $ such that \begin{align*}
&  \E \bigg[ \sup_{h \in \mathbf{H}^t(M)}   \left|   \frac{1}{\sqrt{n}} \sum_{i=1}^n\big \{  [R_{h}^K(Z_i)]_{j} -\E_{} ( [R_{h}^K(Z)]_{j} ) \:  \big \}  \right|      \;  \bigg]    \leq \frac{D}{\sqrt{n}}  \int_{0}^{8  \|  F  \|_{L^2(\mathbb{P})}  }  \sqrt{\log N_{[]}(\mathcal{F}, \| . \|_{L^2(\mathbb{P})},  \epsilon  ) } d \epsilon.
\end{align*}
Since $t > d/2$, the set $\mathbf{H}^t(M)$ is compact under the $\| . \|_{\infty}$ norm. Let $\{  h_i \}_{i=1}^T  $  denote a $\delta > 0$ covering of  $\big( \mathbf{H}^t(M) , \| . \|_{\infty}   \big)$. Define the functions $$ e_{i}(Z) = \sup_{h \in \mathbf{H}^t(M) : \|  h -h_i \|_{\infty} < \delta  } \ \big| \:  [R_{h}^K(Z)]_{j}  -  [R_{h_i}^K(Z)]_{j}   \:   \big| \; \; \; \; \; i=1,\dots,T . $$
By definition of the $\{ e_i \}_{i=1}^T $, it follows that $ \big \{   [R_{h_i}^K(Z)]_{j}   - e_i \; ,  \;  [R_{h_i}^K(Z)]_{j}  + e_i   \big \}_{i=1}^T  $ is a bracket covering for $\mathcal{F}$.  Let $C_1(M)$ and $\kappa \in (0,1]$ be as in Condition \ref{residuals}.  By Condition \ref{residuals}, we have that  \begin{align*}
 \| e_i \|_{L^2(\mathbb{P})}^2  &   \leq \E \bigg[  \left|   \big[ G_{b,K}^{-1/2}  b^K(W) \big ]_{j}  \right| ^2  \E \bigg[    \sup_{h \in \mathbf{H}^t(M) : \|  h -h_i \|_{\infty} < \delta  }           \left| \rho_{l}(Y,h_{}(X)) - \rho_{l}(Y,h_{i}(X))  \right|^2  \bigg| W \bigg]   \bigg] \\ & \leq C_1^2 \delta^{2 \kappa} \E \bigg[  \left|   \big[ G_{b,K}^{-1/2}  b^K(W) \big ]_{j}  \right| ^2 \bigg] \\ & = C_1^2 \delta^{2 \kappa}.
\end{align*}
It follows that \begin{align*}
 & \int_{0}^{8 \|  F  \|_{L^2(\mathbb{P})}  }  \sqrt{\log N_{[]}(\mathcal{F}, \| . \|_{L^2(\mathbb{P})},  \epsilon  ) } d \epsilon \leq  \int_{0}^{8 \|  F  \|_{L^2(\mathbb{P})}  } \sqrt{\log N \bigg(  \mathbf{H}^t(M)   ,\|.\|_{\infty} ,   \bigg( \frac{\epsilon}{ 2  C_1}   \bigg)^{1/ \kappa} \:   \bigg)   } d \epsilon .
\end{align*}
 By \citep[Proposition C.7]{ghosal2017fundamentals}, we have  $  \log N_{}( \mathbf{H}^t(M),  \| . \|_{ \infty}, \epsilon) \lessapprox  \epsilon^{-d/t}$ as $\epsilon \downarrow 0$.  It follows that there exists a universal constant $ D> 0$ such that \begin{align*}
  \int_{0}^{8 \|  F  \|_{L^2(\mathbb{P})}  } \sqrt{\log N \bigg(  \mathbf{H}^t(M)   ,\|.\|_{\infty} ,   \bigg( \frac{\epsilon}{ 2  C_1}   \bigg)^{1/ \kappa} \:   \bigg)   } d \epsilon  &  \leq D  \int_{0}^{8 \|  F  \|_{L^2(\mathbb{P})}  } \epsilon^{-d/2 \kappa t} d \epsilon  \\ & \leq D \int_{0}^{8 C_2}  \epsilon^{-d/2 \kappa t} d \epsilon .
 \end{align*}
 Since $t > (2 \kappa)^{-1} d$  (by Assumption \ref{residuals}(ii)), the integral  above  is convergent. By monotonicity of the $L^p(\mathbb{P})$ norm and combining all the preceding bounds, it follows that there exists a universal constant $ D > 0 $ such that

\begin{align*}
 \E_{} \bigg[ \sup_{h \in \mathbf{H}^t(M)  }  \|   \E_n [R_{h}^K(Z)] - \E_{}[R_{h}^K(Z)]       \|_{\ell^2}  \bigg]   & \leq  \bigg \|      \sup_{h \in \mathbf{H}^t(M)  }  \|   \E_n [R_{h}^K(Z)] - \E_{}[R_{h}^K(Z)]       \|_{\ell^2}      \bigg  \|_{L^2(\mathbb{P})} \\ & \leq D \frac{\sqrt{K}}{\sqrt{n}}.
\end{align*}
\end{proof}

\begin{lemma}
\label{emp-c}
Suppose Conditions \ref{residuals}, \ref{residuals2}(i)(ii) and \ref{fsbasis}(i) hold.  For each fixed $l \in \{ 1 ,  \dots , d_{\rho} \}$ and function $h: \mathcal{X} \rightarrow \R$, define
 $$ R_{h,l}^K(Z) =     \big[ G_{b,K}^{-1/2}  b^K(W) \big] \rho_{l}(Y,h_{}(X)). $$
Let $\epsilon> 0 $ be as in Condition \ref{residuals2}(ii) and define $\gamma = 1- 1/(2+2 \epsilon) > 1/2 $. Suppose $\bar{K}_{\max} \rightarrow  \infty$ is any sequence of sieve dimensions that satisfies $ \big(\log( \bar{K}_{\max}) \big)^3 = o(n^{\gamma-1/2})  $ and $K_{\min} \asymp  \big(\log( \bar{K}_{\max}) \big)^2$. Define the grid of sieve dimensions $\mathcal{K}_n= [ K_{\min} , \bar{K}_{\max} ] \cap \mathbb{N} $. Then, given any $M > 0 $, there exists a universal constant $D = D(M) < \infty$ such that  \begin{equation}
  \mathbb{P}     \bigg(   \sup_{l \in \{ 1 , \dots ,d_{\rho} \} }     \sup_{K \in \mathcal{K}_n} \sup_{h \in \mathbf{H}^t(M)}   K^{-1/2} 
 \| \E_n \big[ R_{h,l}^K(Z) \big]   - \E_{} \big[   R_{h,l}^K(Z)  \big]       \|_{\ell^2}  \leq  \frac{D}{\sqrt{n}}   \bigg ) \rightarrow 1 .    \label{bern2}
    \end{equation}
\end{lemma}

\begin{proof}[Proof of Lemma \ref{emp-c}]
It suffices to verify that (\ref{bern2}) holds at each fixed $l \in \{  1 , \dots, d_{\rho} \} $. Fix any such $l$. For a given sequence of deterministic constants $L_n \uparrow \infty$, define   \begin{align*}  & \xi_{1,i}^K (h_{}) =  R_{h,l}^K(Z_i)  \mathbbm{1} \bigg \{ \sup_{h \in \mathbf{H}^t(M)}  | \rho_l(Y_i, h(X_i)) | \leq L_n   \bigg \} \;, \\ &  \xi_{2,i}^K (h_{}) = R_{h,l}^K(Z_i)  \mathbbm{1} \bigg \{ \sup_{h \in \mathbf{H}^t(M)}  | \rho_l(Y_i, h(X_i)) | > L_n    \bigg \}  \; .      \end{align*}
Write the deviation as \begin{align} \label{sums} (\E_n - \E_{})[R_{h,l}^K(Z) ]   = \sum_{i=1}^n  \Xi_{1,i}^K (h_{})  + \sum_{i=1}^n \Xi_{2,i}^K (h_{})    . \end{align}
where $ \Xi_{1,i}^K (h_{}) = n^{-1} [\xi_{1,i}^K (h_{}) - \E_{} \xi_{1,i}^K (h_{})] $ and $ \Xi_{2,i}^K(h_{}) = n^{-1} [ \xi_{2,i}^K(h_{}) - \E_{} \xi_{2,i}^K(h_{}) ]$.
First, we derive a bound for $ \sum_{i=1}^n \Xi_{2,i}^{K}(h) $.  Let $\epsilon > 0 $ be as in Condition \ref{residuals2}(ii).  By definition of $\zeta_{b,K}$, we have $ \zeta_{b,K}^{-1} \|  G_{b,K}^{-1/2} b^{K}(W_i)     \|_{\ell^2}  \leq   1 $ almost surely. It follows that
\begin{align*}   &  \mathbb{P}_{} \bigg( \sup_{h \in \mathbf{H}^t(M)}   \bigg \| \sum_{i=1}^n \Xi_{2,i}^K (h_{}) \bigg \|_{\ell^2}   >  \frac{\zeta_{b,K}}{\sqrt{n}}  \bigg) \\ & \leq       \frac{\sqrt{n}}{\zeta_{b,K}} \E_{} \bigg(     \sup_{h \in \mathbf{H}^t(M)}    \sum_{i=1}^n \| \Xi_{2,i}^K (h_{}) \|_{\ell^2} \bigg) \\ & \leq 2 \sqrt{n} \E_{} \bigg( \sup_{h \in  \mathbf{H}^t(M)}  | \rho_l(Y_i, h(X_i))  | \mathbbm{1} \bigg \{ \sup_{h \in \mathbf{H}^t(M)}  | \rho_l(Y_i, h(X_i)) | > L_n    \bigg \} \bigg) \\ & \leq \frac{2 \sqrt{n}}{L_n^{1+ \epsilon}} \E_{} \bigg( \sup_{h \in \mathbf{H}^t(M)}  | \rho_l(Y_i, h(X_i))  |^{2+ \epsilon} \bigg) . \end{align*}
Since $\E_{} \big( \sup_{h \in \mathbf{H}^t(M)} \left| \rho_l(Y_i, h(X_i))  \right|^{2+ \epsilon} \big) < \infty $, a union bound over $K \in \mathcal{K}_n$ yields $$  \mathbb{P}_{} \bigg( \bigcup_{K \in \mathcal{K}_n} \bigg \{ \sup_{h \in \mathbf{H}^t(M)}   \bigg  \| \sum_{i=1}^n \Xi_{2,i}^K (h) \bigg \|_{\ell^2}   >  \frac{\zeta_{b,K}}{\sqrt{n}}  \bigg \} \bigg) \lessapprox \frac{\sqrt{n} \log(\bar{K}_{\max})}{L_n^{1+ \delta}} .  $$
The term on the right is $o(1)$ when $L_n^{1+ \epsilon} \asymp   \sqrt{n} (\log \bar{K}_{\max})^{1+ \epsilon} $. The desired bound then follows from observing that $\zeta_{b,K} \lessapprox \sqrt{K}$. It remains to bound the first sum in (\ref{sums}) when $L_n^{1+ \epsilon} \asymp \sqrt{n} (\log \bar{K}_{\max})^{1+ \epsilon} $. Observe that $$ \sup_{h \in \mathbf{H}^t(M)} \bigg \|  \sum_{i=1}^n \Xi_{1,i}^K(h)        \bigg \|_{\ell^2} =  \sup_{h \in \mathbf{H}^t(M)} \; \sup_{\alpha  \in \mathbb{S}^{K-1}  }  \sum_{i=1}^n \alpha' \Xi_{1,i}^K (h)  $$
where $\mathbb{S}^{K-1} = \{ v \in \R^K : \|v \|_{\ell^2} = 1  \}$. Let $C_2 = C_2(M) < \infty$ be as in   Condition \ref{residuals2}(i). Define $\gamma = 1- 1/(2+2 \epsilon) > 1/2 $. For any fixed $\alpha \in  \mathbb{S}^{K-1}  $ and  $h \in \mathbf{H}^t(M) $, we have that   \begin{align*} &    \E_{}  \big[ \big(  \alpha' \Xi_{1,i}^K(h)  \big)^2 \big] \leq   n^{-2} \E_{} \bigg( \alpha' G_{b,K}^{- 1/2} b^K(W_i) b^K(W_i)' G_{b,K}^{- 1/2} \alpha  \sup_{h \in \mathbf{H}^t(M)} | \rho_l(Y,h (X) )|^2   \bigg) \leq  C_2^2 n^{-2}  \: , \\ &  \left|  \alpha' \Xi_{1,i}^K (h)   \right| \leq 2 n^{-1}  L_n \zeta_{b,K} \lessapprox  \frac{2  \zeta_{b,K} \log \bar{K}_{\max} }{n^{\gamma} }.    \: 
\end{align*}
By Lemma \ref{emp-b}, there exists a universal constant $D  =D(M) <  \infty  $ such that
\begin{align*}
\E_{} \bigg(    \sup_{h \in \mathbf{H}^t(M)}    \bigg     \|  \sum_{i=1}^n \Xi_{1,i}^K (h) \bigg \|_{\ell^2} \bigg) \leq  D  \frac{\sqrt{K}}{\sqrt{n}} .
\end{align*}
holds for every $K$. The preceding bounds and Talagrand's inequality \citep[Theorem 3.3.9]{gine2021mathematical} yields   \begin{align*} & \mathbb{P}_{} \bigg( \sup_{h \in \mathbf{H}^t(M)}  \bigg \|  \sum_{i=1}^n \Xi_{1,i}^K        \bigg \|_{\ell^2} \geq   \frac{ D  \sqrt{K}  }{\sqrt{n}} + \frac{\sqrt{K}}{\sqrt{n}}   \bigg) \\ &  \leq \exp \bigg(  - \frac{1}{ 2 C_2^2 K^{-1}  + (8 D  + 4/3) \big(\zeta_{b,K} \log (\bar{K}_{\max}) K^{-1/2} n^{1/2 - \gamma}   \big   )      }    \bigg) .
\end{align*}
Let $E > 0 $ be such that $\zeta_{b,K} \leq E \sqrt{K}$. From a union bound, we obtain \begin{align*}
 & \mathbb{P}_{} \bigg( \bigcup_{K \in \mathcal{K}_n} \bigg \{  \sup_{h \in \mathbf{H}^t(M)}   \bigg \|  \sum_{i=1}^n \Xi_{1,i}^K        \bigg \|_{\ell^2} \geq  \frac{ D  \sqrt{K}  }{\sqrt{n}} + \frac{\sqrt{K}}{\sqrt{n}} \bigg \}  \bigg) \\ & \lessapprox \bar{K}_{\max} \exp \bigg(   - \frac{1}{ 2 C_2^2 K_{\min}^{-1} + E(8 D + 4/3) \log(\bar{K}_{\max}) n^{1/2 - \gamma}  }         \bigg).
\end{align*}
This term is $o(1)$ since $K_{\min} \log(\bar{K}_{\max}) / n^{\gamma - 1/2} \downarrow 0$ and $ \log(\bar{K}_{\max}) K_{\min}^{-1} \downarrow 0  $.

\end{proof}

\begin{proof}[Proof of Theorem \ref{main-contract}]
Let $D > 0$  denote a generic universal constant that may change from line to line. 

\begin{enumerate}
\item[\textbf{$(i)$}]First, we derive a lower bound for the normalizing constant of the posterior. We aim to show there exists $C,C' > 0 $ such that \begin{align}
    \label{t1lb} \int \exp \bigg( - \frac{n}{2} \E_n \big[    \widehat{m}(W,h) ' \widehat{\Sigma}(W)  \widehat{m}(W,h)            \big]    \bigg) d \mu(h) \geq C \exp \big(   - C' n    \epsilon_{n}^2   \big)  
\end{align}
holds with $\mathbb{P}$ probability approaching $1$.

Let $ \mathcal{S}_n$ be as in Assumption \ref{sampling-uncert}. By Assumption \ref{sampling-uncert}, we have \begin{align*}
     & \int \exp \bigg( - \frac{n}{2} \E_n \big[    \widehat{m}(W,h) ' \widehat{\Sigma}(W)  \widehat{m}(W,h)            \big]  \bigg) d \mu(h)  \\ & \geq \int_{\mathcal{S}_n} \exp \bigg( - \frac{n}{2} \E_n \big[    \widehat{m}(W,h) ' \widehat{\Sigma}(W)  \widehat{m}(W,h)            \big]  \bigg) d \mu(h) \\ & \geq \exp(-n D \epsilon_n^2) \int_{\mathcal{S}_n} \exp \big( - nD \E \big( \| \widetilde{m}(W,h)   \|_{\ell^2}^2  \big) \big) d \mu(h)
\end{align*}
with $\mathbb{P}$ probability approaching $1$. 

Let $h_n$ be as in Assumption \ref{weak-bias}. Since $m(W,h_0) = \mathbf{0}$, we have \begin{align*}
    \| \widetilde{m}(W,h)   \|_{\ell^2} & = \|  \widetilde{m}(W,h)    -  \widetilde{m}(W,h_n) + \widetilde{m}(W,h_n) - m(W,h_n) + m(W,h_n)   \|_{\ell^2}   \\ & \leq \| \widetilde{m}(W,h)    -  \widetilde{m}(W,h_n) \|_{\ell^2} + \| \widetilde{m}(W,h_n)- m(W,h_n) \|_{\ell^2} + \|  m(W,h_n) - m(W,h_0)  \|_{\ell^2}
\end{align*}
for any $h$. By Assumption \ref{weak-bias}-\ref{loc-conc}, it follows that \begin{align*}
    & \int_{\mathcal{S}_n} \exp \big( - nD \E \big( \| \widetilde{m}(W,h)   \|_{\ell^2}^2  \big) \big) d \mu(h) \\ & \geq \exp(-n D \epsilon_n^2) \int_{\mathcal{S}_n} \exp \big( - n D \E \big( \| \widetilde{m}(W,h) - \widetilde{m}(W,h_n)   \|_{\ell^2}^2  \big)  \big) d \mu(h) \\ & \geq \exp(-n D \epsilon_n^2) \int_{\mathcal{S}_n} d \mu(h).
\end{align*}
Let $\mathcal{R}_n \supseteq \mathcal{S}_n $ be as in Assumption \ref{loc-conc}. Since $\mathcal{R}_n = \mathcal{S}_n \cup (\mathcal{R}_n \setminus \mathcal{S}_n)  $, we have $ \int_{\mathcal{S}_n} d \mu(h) = \int_{\mathcal{R}_n} d \mu(h) - \int_{\mathcal{R}_n \setminus \mathcal{S}_n} d \mu(h) $. By Assumption \ref{loc-conc}, we have \begin{align*}
    &  \int_{\mathcal{R}_n } d \mu(h)  \geq c\exp(-C'n  \epsilon_n^2) \\ & \int_{h \in \mathcal{R}_n \setminus\mathcal{S}_n  }  d \mu(h) \leq C \exp(-B n \epsilon_n^2)
\end{align*}
for some $c,C,C',B, > 0$ with $B > C'$. Since $B > C'$, it follows that $$ \int_{h \in \mathcal{S}_n  }  d \mu(h) \geq  c\exp(-C'n  \epsilon_n^2) - C \exp(-B n \epsilon_n^2) \geq\exp(-n D \epsilon_n^2)  . $$
The lower bound in (\ref{t1lb}) follows from combining all the preceding estimates.
 
\item[\textbf{$(ii)$}] For any set $\Omega$, the lower bound in part $(i)$ yields \begin{align*}
     \mu \big ( h \in \Omega \: \big| \: \mathcal{D}_n  \big ) & = \frac{ \int_{h \in \Omega}  \exp \big( - \frac{n}{2} \E_n \big[    \widehat{m}(W,h) ' \widehat{\Sigma}(W)  \widehat{m}(W,h)            \big]    \big)  }{\int_{} \exp \big( - \frac{n}{2} \E_n \big[    \widehat{m}(W,h) ' \widehat{\Sigma}(W)  \widehat{m}(W,h)            \big]    \big)  } \\ & \leq D \exp (C' n \epsilon_n^2) \int_{h \in \Omega}  \exp \bigg( - \frac{n}{2} \E_n \big[    \widehat{m}(W,h) ' \widehat{\Sigma}(W)  \widehat{m}(W,h)            \big]    \bigg) d \mu(h)
\end{align*}
with $\mathbb{P}$ probability approaching $1$, for some universal constants $D, C'  > 0$. Fix any $R >  C '$ and define the set \begin{align*}
    \Omega = \{ h : \| \widehat{m}(W,h) - m(W,h_0)  \|_{L^2(\mathbb{P}_n, \widehat{\Sigma} )}^2  > 2 R \epsilon_n^2     \}
\end{align*}
Since $m(W,h_0) = \mathbf{0}$, it follows that \begin{align*}
    \mu \big ( F \in \Omega \: \big| \: \mathcal{D}_n  \big ) & \leq D \exp(C' n \epsilon_n^2) \int_{h \in \Omega}  \exp \bigg( - \frac{n}{2} \E_n \big[    \widehat{m}(W,h) ' \widehat{\Sigma}(W)  \widehat{m}(W,h)            \big]    \bigg) d \mu(h) \\ & =  D \exp(C' n \epsilon_n^2) \int \limits_{h : \| \widehat{m}(W,h)   \|_{L^2(\mathbb{P}_n, \widehat{\Sigma} )}^2  > 2 R \epsilon_n^2  } \exp \bigg( - \frac{n}{2} \E_n \big[    \widehat{m}(W,h) ' \widehat{\Sigma}(W)  \widehat{m}(W,h)            \big]    \bigg) d \mu(h) \\ & \leq D \exp( [C' - R] n \epsilon_n^2) .
\end{align*}
Since $R > C'$ and $n \epsilon_n^2 \uparrow \infty$, the claim follows.
\end{enumerate}
\end{proof}

\begin{proof}[Proof of Corollary \ref{contract-strong}] For any set $\Omega$, the lower bound derived in the proof of part $(i)$ in Theorem \ref{main-contract} yields \begin{align*}
     \mu \big ( h \in \Omega \: \big| \: \mathcal{D}_n  \big ) & =  \frac{ \int_{h \in \Omega}  \exp \big( - \frac{n}{2} \E_n \big[    \widehat{m}(W,h) ' \widehat{\Sigma}(W)  \widehat{m}(W,h)            \big]    \big)  }{\int_{} \exp \big( - \frac{n}{2} \E_n \big[    \widehat{m}(W,h) ' \widehat{\Sigma}(W)  \widehat{m}(W,h)            \big]    \big)  } \\ & \leq D \exp (C' n \epsilon_n^2) \int_{h \in \Omega}  \exp \big( - \frac{n}{2} \E_n \big[    \widehat{m}(W,h) ' \widehat{\Sigma}(W)  \widehat{m}(W,h)            \big]    \big) d \mu(h)
\end{align*}
with $\mathbb{P}$ probability approaching $1$, for some universal constants $D, C'  > 0$. Define   \begin{align*}
    \Omega = \{ h : h \notin \mathcal{H}_n : \| \widehat{m}(W,h) - m(W,h_0)  \|_{L^2(\mathbb{P}_n, \widehat{\Sigma} )} \leq L \epsilon_n  \}.
\end{align*}
If the hypothesis of Corollary \ref{contract-strong} holds for some $D' > C'$, the preceding bound and the conclusion of Theorem \ref{main-contract} yields \begin{align*}
    \mu \bigg ( h \in \mathcal{H}_n : \| \widehat{m}(W,h) - m(W,h_0)  \|_{L^2(\mathbb{P}_n, \widehat{\Sigma} )}  \leq L \epsilon_n  \: \bigg| \: \mathcal{D}_n  \bigg ) \xrightarrow{\mathbb{P}} 1.
\end{align*}
The claim follows from the definition of the modulus $\omega_n(.)$

\end{proof}

\begin{proof}[Proof of Theorem \ref{t1}]
\begin{enumerate}
    \item[$(i)$]

First, we aim to apply Theorem \ref{main-contract} with  $\epsilon_n = \sqrt{K_n} / \sqrt{n}$. We proceed by verifying that  Assumptions \ref{sampling-uncert}-\ref{loc-conc} hold. Given any fixed function $h$, we can write \begin{align*}
     \widehat{m}(w,h)  =     \E_n \big(   \rho(Y,h_{}(X)) \big[ G_{b,K}^{-1/2} b^K(W) \big]'    \big)               [ \widehat{G}_{b,K}^{o}  ]^{-1} G_{b,K}^{-1/2}   b^K(w).
\end{align*}
It follows that
 \begin{align*}
   &  \E_n \big( \| \widehat{m}(W,h) \|_{\ell^2}^2  \big)   = \sum_{l=1}^{d_{\rho}}  [ \E_n (R_{h, l}^K)      ] ' [ \widehat G_{b,K}^o]^{-1}  [ \E_n (R_{h, l}^K)  ] \\ & \text{where} \; \; \; \;   R_{h,l}^K(Z) =     \big[ G_{b,K}^{-1/2}  b^K(W) \big] \rho_{l}(Y,h_{}(X)). 
\end{align*}
Observe that, by definition of $G_{b,K}$, the functions in the vector $G_{b,K}^{-1/2} b^K(W)$ are an orthonormal (with respect to the $L^2(\mathbb{P})$ inner product) basis of the linear space spanned by $\{ 
 b_1(W) , \dots , b_K(W) \}$. Hence, the $L^2(\mathbb{P})$ norm of $\Pi_K m(W,h)$ can be expressed as
\begin{align*}
  \|  \Pi_K m(W,h)  \|_{L^2(\mathbb{P})}^2   =  \E  \big( \| \Pi_K m(W,h) \|_{\ell^2}^2 \big) = \sum_{l=1}^{d_{\rho}} \| \E[R_{h,l}^K(Z)]   \|_{\ell^2}^2. 
\end{align*}
We denote the empirical analog of this representation by  \begin{align*} 
  \|  \widehat{\Pi}_{K} m(W,h_{})  \|_{L^2(\mathbb{P}_n)}^2 =  \sum_{l=1}^{d_{\rho}}  \| \E_n (R_{h, l}^K)      \|_{\ell^2}^2.
\end{align*}
Let $\hat{\lambda}_{K,\min}$ and $\hat{\lambda}_{K,\max}$ denote the minimum and maximum eigenvalues of $   [\widehat{G}_{b,K}^o]^{-1}$. By Lemma \ref{aux2}, we have that \begin{equation}  \label{eig} \mathbb{P}(0.9 <\hat{\lambda}_{K,\min} \leq  \hat{\lambda}_{K,\max} < 1.1 ) \rightarrow 1.
\end{equation}
Let $\widetilde{m}(W,h) = \Pi_K[m(W,h)]$. We aim to verify Assumption \ref{sampling-uncert} with $\widetilde{m}(.) $ and the set \begin{align*}
    \mathcal{S}_n = \{h : \| h \|_{\mathbf{H}^t} \leq M , \|  h - h_0 \|_{L^2(\mathcal{X})} \leq \epsilon_n   \}
\end{align*}
for some sufficiently large $ M > 0$, which we specify below.

Fix any $l \in \{1 , \dots , d_{\rho} \} $. On the set where (\ref{eig}) holds,  we have that \begin{align*}
 &  \E_n (R_{h, l}^K)      ] ' [ \widehat G_{b,K}^o]^{-1}  [ \E_n (R_{h, l}^K)  ]   \leq 1.1 \| \E_n (R_{h, l}^K)  \|_{\ell^2}^2.
\end{align*} 
 By Lemma \ref{emp-c}, there exists a $C = C(M) < \infty$ such that  \begin{align*} \sum_{l=1}^{d_{\rho}}
\| \E_n(R_{h,l}^K)       \|_{\ell^2}^2 & \leq   \sum_{l=1}^{d_{\rho}} \big(  \| \E_n(R_{h,l}^K) - \E(R_{h,l}^K)      \|_{\ell^2} + \| \E(R_{h,l}^K)  \|_{\ell^2}  \big)^2                        \nonumber \\  & \leq  C  \bigg(          \frac{K}{n} + \|  \Pi_K m(W,h)  \|_{L^2(\mathbb{P})}^2     \bigg)
\end{align*}
holds for all $h \in \mathbf{H}^t(M)$ (with $\mathbb{P}$ probability approaching $1$). Since $ \epsilon_n^2 = K /n $, Assumption \ref{sampling-uncert} follows. Assumption \ref{weak-bias} is trivially satisfied with the choice $h_n = h_0$, since $ \widetilde{m}(W,h_0) = \Pi_K m(W,h_0) = 0 $. For Assumption \ref{loc-conc}$(iii)$, Condition \ref{mean-lip} yields \begin{align*} \sup_{h \in \mathcal{S}_n}\| \Pi_K m(W,h) \|_{L^2(\mathbb{P})} \leq \sup_{h \in \mathcal{S}_n} \|  m(W,h) \|_{L^2(\mathbb{P})} & = \sup_{h \in \mathcal{S}_n} \| m(W,h) - m(W,h_0) \|_{L^2(\mathbb{P})} \\ &\leq D \sup_{h \in \mathcal{S}_n}  \| h- h_0\|_{L^2(\mathcal{X})} \\ & \leq D \epsilon_n.
\end{align*}
To verify Assumption \ref{loc-conc}$(i-ii)$, we use the set $ \mathcal{R}_n = \{h : \|  h - h_0 \|_{L^2(\mathcal{X})} \leq \epsilon_n   \} $. The RKHS associated to the Gaussian  random element  $G_{\alpha} $    can be represented as \begin{align*}
  \mathbb{H}_{\alpha}  =  \bigg \{  h \in L^2(\mathcal{X})  : \|h \|_{\mathbb{H}_{\alpha}}^2 =   \sum_{i=1}^{\infty}     i^{ 1 +   2\alpha /d}  \left| \langle h ,  e_i \rangle_{L^2(\mathcal{X})}      \right|^2 < \infty               \bigg \} .
\end{align*}
The concentration function of the scaled Gaussian measure $d\mu(.)$  at $h_0$ is given by \begin{align*} 
 \varphi_{h_0}(\epsilon) = \inf_{h \in \mathbb{H}_{\alpha}:\| h - h_0   \|_{L^2( \mathcal{X})} \leq \epsilon  } \bigg\{ \frac{   K}{2} \| h \|_{\mathbb{H}_{\alpha}}^2 - \log \mathbb{P} \bigg(  \| G_{\alpha}   \|_{L^2(\mathcal{X})} < \epsilon \sqrt{ K}           \bigg)             \bigg \}.
\end{align*}
It follows from \citep[Proposition 11.19]{ghosal2017fundamentals} 
 that there exists a $C > 0 $ such that  $ \int_{\mathcal{R}_n} d \mu(h)    \geq \exp \big(  - \varphi_{h_0}(C  \epsilon_{n})        \big)$. Since $h_0 \in \mathcal{H}^p$ for some $p \geq \alpha+d/2$, it follows that $h_0 \in \mathbb{H}_{\alpha}$. In particular, by choosing $h = h_0$ in the infimum defining $\varphi_{h_0}(.)$, we obtain \begin{align*}
    \varphi_{h_0}(\epsilon) \leq  D \bigg[ K  - \log \mathbb{P} \bigg(  \| G_{\alpha}   \|_{L^2(\mathcal{X})} < \epsilon \sqrt{  K}           \bigg)           \bigg].
\end{align*}
For the second term, by an application of  \citep[Lemma 11.47]{ghosal2017fundamentals}, we obtain  \begin{align*}
\varphi_{h_0}(  C  \epsilon_{n})  \leq D \big[    K +   (   \epsilon_{n}  \sqrt{   K})^{-d/\alpha  }        \big].
\end{align*}
Since $\epsilon_n = \sqrt{K} / \sqrt{n}$ and $K \gtrapprox n^{d/2(\alpha + d)} $, the first term on the right of the preceding inequality dominates and we obtain $ \int_{\mathcal{R}_n} d \mu(h) \geq \exp(-C' n \epsilon_n^2)  $ for some $C' > 0$. Assumption \ref{loc-conc}$(i)$ follows. Moreover, we note that the constant $C'$ is independent of $M$. 

Since $d\mu(.)$ is the distribution of  $G_{\alpha} /  \sqrt{K} $, it follows from Theorem 2.1.20 of \citet{gine2021mathematical}  that there exists a universal constant $D > 0 $ such that $$ \int_{\mathcal{R}_n \setminus \mathcal{S}_n} d \mu(h) \leq \int_{h : \| h \|_{\mathbf{H}^t} > M} d \mu(h) \leq2 \exp(-D M^2 n \epsilon_n^2).   $$
By picking $ M > 0$ large enough, we can ensure that  $DM^2  > C'$ and Assumption \ref{loc-conc}$(ii)$ follows. From the conclusion of Theorem \ref{main-contract}, we obtain $$ \mu \big ( h :  \| \widehat{m}(W,h)  - m(W,h_0)  \|_{L^2(\mathbb{P}_n,\widehat{\Sigma})}  > L \epsilon_n  \: \big| \: \mathcal{D}_n  \big ) \xrightarrow{\mathbb{P}} 0 $$
for some universal constant $  L > 0$.
\item[$(ii)$] We aim to apply Corollary \ref{contract-strong} with the metric $d(h,h_0) = \| h - h_0 \|_{L^2(\mathbb{P})}$. For a fixed $E > 0$, define the set $ \mathcal{G}_n = \{ h : \|h \|_{\mathbf{H}^t}  \leq E \} $. By Theorem 2.1.20 of \citet{gine2021mathematical}, there exists a universal constant $D > 0 $ such that $ \mu(h \notin \mathcal{G}_n) \leq 2 \exp (-D E^2 n \epsilon_n^2)$. We can pick $E > 0$ large enough so as to satisfy the hypothesis of Corollary \ref{contract-strong}.

Since the conditions of Corollary \ref{contract-strong} are satisfied, it only remains to verify that the modulus satisfies $\omega_n \xrightarrow{\mathbb{P}} 0$. Define the set $$ \mathcal{E}_n = \{ h \in \mathcal{G}_n : \| \widehat{m}(W,h)  - m(W,h_0)  \|_{L^2(\mathbb{P}_n,\widehat{\Sigma})} \leq L \epsilon_n  \} . $$

 By arguing as in part $(i)$ and using Condition \ref{fsbasis}$(ii)$ and Lemma \ref{emp-c}, we can deduce that $ \sup_{h \in \mathcal{E}_n} \| \Pi_K m(W,h) \|_{L^2(\mathbb{P})} \leq D \epsilon_n $ with $\mathbb{P} $ probability approaching $1$. It follows that $\sup_{h \in \mathcal{E}_n} \| m(W,h) \|_{L^2(\mathbb{P})} \leq D \gamma_n $ where $$ \gamma_n = \max \bigg\{ \epsilon_n , \sup_{h \in \mathcal{G}_n   } \| (\Pi_K - I) m(W,h) \|_{L^2(\mathbb{P})}   \bigg\} .$$
By Condition \ref{fsbasis}$(iii)$, we have $\gamma_n \to 0$. The set $ \mathcal{G}_n $ is compact under the $\| \cdot \|_{L^2(\mathbb{P})}$ metric, and by Condition \ref{mean-lip}, the map $h \rightarrow  m(W,h)$ is uniformly continuous on $\mathcal{G}_n$. To prove the claim, it suffices to prove that for every $\delta > 0$, there exists a $\gamma > 0$ such that $$ h \in \mathcal{G}_n \; , \; \| m(W,h)  \|_{L^2(\mathbb{P})} < \gamma \implies \| h - h_0 \|_{L^2(\mathbb{P})} < \delta.   $$
Suppose this fails. Then for some $\gamma_n \rightarrow 0$, $\delta > 0$ and a sequence $(h_n)_{n=1}^{\infty} \in \mathcal{G}_n$, we have $ \| m(W,h_n) \|_{L^2(\mathbb{P})} < \gamma_n$ and $\| h_n - h_0 \|_{L^2(\mathbb{P})} \geq \delta$. Since the set $\{h  \in \mathcal{G}_n : \| h - h_0 \|_{L^2(\mathbb{P})} \geq \delta  \}$ is a closed 
 (and hence compact) subset of $\mathcal{G}_n$, the continuous function $h \rightarrow \| m(W,h) \|_{L^2(\mathbb{P})}$ achieves its minimum on it. Since $h_0$ is the unique zero of this function, there must exist a $\gamma^* > 0$ such that $$ \inf_{h \in \mathcal{G}_n : \| h - h_0 \|_{L^2(\mathbb{P})} \geq \delta} \| m(W,h) \|_{L^2(\mathbb{P})} \geq \gamma^*. $$
This leads to a contradiction for any $\gamma_n < \gamma^*$.
\end{enumerate}

\end{proof}

\begin{proof}[Proof of Theorem \ref{t1-2}]
We argue similarly to Theorem \ref{t1}. Let $h_0 \in \mathcal{H}^p \cap \Theta_0$ be as in the statement of the theorem. As this quasi-Bayes posterior contains a continuously updated weighting matrix, Theorem \ref{main-contract} does not directly apply. However, with $\mathcal{S}_n = \{ h : \|h \|_{\mathbf{H}^t} \leq M, \| h - h_0 \|_{L^2(\mathcal{X})} \leq \epsilon_n \}$ and Condition \ref{fsbasis2}, we have \begin{align*}
     & \int \exp \bigg( - \frac{n}{2} \E_n \big[    \widehat{m}(W,h) ' \widehat{\Sigma}(W,h)  \widehat{m}(W,h)            \big]  \bigg) d \mu(h)  \\ & \geq c\int_{\mathcal{S}_n} \exp \bigg( - \frac{n}{2} \E_n \big[    \widehat{m}(W,h) '  \widehat{m}(W,h)            \big]  \bigg) d \mu(h) 
\end{align*}
for some $c > 0 $, with $\mathbb{P}$ probability approaching $1$. The remainder of the argument is identical to Theorem \ref{main-contract}. As such, we can conclude, similarly to Theorem \ref{t1}, that $$ \mu \big ( h :  \| \widehat{m}(W,h)    \|_{L^2(\mathbb{P}_n,\widehat{\Sigma})}  > L \epsilon_n  \: \big| \: \mathcal{D}_n  \big ) \xrightarrow{\mathbb{P}} 0 $$
for some universal constant $  L > 0$.

Next, we aim to apply Corollary \ref{contract-strong} with the metric $d(h,\Theta_0) = \inf_{h^* \in \Theta_0}\| h - h^* \|_{L^2(\mathbb{P})}$. For a fixed $E > 0$, define the set $ \mathcal{G}_n = \{ h : \|h \|_{\mathbf{H}^t}  \leq E \} $. By Theorem 2.1.20 of \citet{gine2021mathematical}, there exists a universal constant $D > 0 $ such that $ \mu(h \notin \mathcal{G}_n) \leq 2 \exp (-D E^2 n \epsilon_n^2)$. We can pick $E > 0$ large enough so as to satisfy the hypothesis of Corollary \ref{contract-strong}.  Define the set $$ \mathcal{E}_n = \{ h \in \mathcal{G}_n : \| \widehat{m}(W,h)  )  \|_{L^2(\mathbb{P}_n,\widehat{\Sigma})} \leq L \epsilon_n  \} . $$
By arguing as in Theorem \ref{t1}, we obtain $\sup_{h \in \mathcal{E}_n} \| m(W,h) \|_{L^2(\mathbb{P})} \leq D \gamma_n $ where $$ \gamma_n = \max \bigg\{ \epsilon_n , \sup_{h \in \mathcal{G}_n   } \| (\Pi_K - I) m(W,h) \|_{L^2(\mathbb{P})}   \bigg\} .$$
By Condition \ref{fsbasis}$(iii)$, we have $\gamma_n \to 0$. Since the distance function $h \rightarrow d(h,\Theta_0)$ is continuous, for any $\delta > 0 $, the set $\{ h \in \mathcal{G}_n : d(h,\Theta_0) \geq \delta \} $ is a closed (and hence compact) subset of $\mathcal{G}_n$. As such, by an analgous argument to Theorem \ref{t1}, there exists a sequence $\delta_n \rightarrow 0$ such that $ \sup_{h \in \mathcal{E}_n} d(h,\Theta_0) \leq \delta_n$ with $\mathbb{P}$ probability approaching $1$.
\end{proof}

\begin{proof}[Proof of Theorem \ref{rate}]
First, we aim to apply Theorem \ref{main-contract} with  $\epsilon_n = \sqrt{K_n} / \sqrt{n}$. Let $\gamma > 0$ be as in Conditions \ref{gp-link}-\ref{lcurv} and $\widetilde{m}(W,h) = \Pi_K[m(W,h)]$. Define $$
    \mathcal{S}_n^{\star} = \{h : \| h \|_{\mathbf{H}^t} \leq M , \: \| h \|_{\mathcal{H}^{\gamma}} \leq M \;,   \|  h - h_0 \|_{w,\sigma} \leq \epsilon_n   \} $$ for some sufficiently large $ M > 0$, which we specify below. Since the set $\mathcal{S}_n$ is compact, an analogous argument to Theorem \ref{t1} implies that \begin{align*}
        \mathcal{S}_n^{\star} \subseteq \mathcal{S}_n = \{h : \| h \|_{\mathbf{H}^t} \leq M , \: \| h \|_{\mathcal{H}^{\gamma}} \leq M \;,   \|  h - h_0 \|_{w,\sigma} \leq \epsilon_n \; , \| h - h_0 \|_{L^2(\mathcal{X})}  \leq \delta_n  \}
    \end{align*}
for some sequence $\delta_n \rightarrow 0$. We need to verify that Assumptions \ref{sampling-uncert}-\ref{loc-conc} hold with $\widetilde{m}(\cdot)$ and $\mathcal{S}_n$. 

Verification of Assumption \ref{sampling-uncert}-\ref{weak-bias} is analogous to Theorem \ref{t1}. We focus on Assumption \ref{loc-conc}. For Assumption \ref{loc-conc}$(iii)$, Condition \ref{gp-link} and \ref{lcurv} yields
\begin{align*} \sup_{h \in \mathcal{S}_n}\| \Pi_K m(W,h) \|_{L^2(\mathbb{P})} \leq \sup_{h \in \mathcal{S}_n} \|  m(W,h) \|_{L^2(\mathbb{P})} & = \sup_{h \in \mathcal{S}_n} \| m(W,h) - m(W,h_0) \|_{L^2(\mathbb{P})} \\ &\leq D \sup_{h \in \mathcal{S}_n}  \| D_{h_0}[h-h_0] \|_{L^2(\mathcal{X})} \\ & \leq D \sup_{h \in \mathcal{S}_n}  \| h- h_0  \|_{w,\sigma} \\ & \leq D \epsilon_n.
\end{align*}
To verify Assumption \ref{loc-conc}$(i-ii)$, we use the set $ \mathcal{R}_n = \{h : \|  h - h_0 \|_{w,\sigma} \leq \epsilon_n   \} $. The RKHS associated to the Gaussian  random element  $G_{\alpha} $    can be represented as \begin{align*}
  \mathbb{H}_{\alpha}  =  \bigg \{  h \in L^2(\mathcal{X})  : \|h \|_{\mathbb{H}_{\alpha}}^2 =   \sum_{i=1}^{\infty}     i^{ 1 +   2\alpha /d}  \left| \langle h ,  e_i \rangle_{L^2(\mathcal{X})}      \right|^2 < \infty               \bigg \} .
\end{align*}
The concentration function of the scaled Gaussian measure $d\mu(.)$  at $h_0$ is given by \begin{align*} 
 \varphi_{h_0}(\epsilon) = \inf_{h \in \mathbb{H}_{\alpha}:\| h - h_0   \|_{w,\sigma} \leq \epsilon  } \bigg\{ \frac{   K}{2} \| h \|_{\mathbb{H}_{\alpha}}^2 - \log \mathbb{P} \bigg(  \| G_{\alpha}   \|_{w,\sigma} < \epsilon \sqrt{ K}           \bigg)             \bigg \}.
\end{align*}
It follows from \citep[Proposition 11.19]{ghosal2017fundamentals} 
 that there exists a $C > 0 $ such that  $ \int_{\mathcal{R}_n} d \mu(h)    \geq \exp \big(  - \varphi_{h_0}(C  \epsilon_{n})        \big)$. By choosing $h = h_0$ in the infimum defining $\varphi_{h_0}(.)$, we obtain \begin{align*}
    \varphi_{h_0}(\epsilon) \leq  D \bigg[ K  - \log \mathbb{P} \bigg(  \| G_{\alpha}   \|_{w,\sigma} < \epsilon \sqrt{  K}           \bigg)           \bigg].
\end{align*}
To obtain the desired bound, it suffices to show that  \begin{align}
    \label{concbound} - \log \mathbb{P} \bigg(  \| G_{\alpha}   \|_{w,\sigma} < \epsilon_n \sqrt{  K}           \bigg)  \leq D K.
\end{align} 
 Consider first the case where the model is mildly ill-posed so that $\sigma_i \asymp i^{-\zeta/d}$ for some $\zeta \geq 0$. By an application of  \citep[Lemma 11.47]{ghosal2017fundamentals}, we obtain  \begin{align*}
- \log \mathbb{P} \bigg(  \| G_{\alpha}   \|_{w,\sigma} < \epsilon_n \sqrt{  K}           \bigg) \leq C     \big(  \epsilon_{n}  \sqrt{   K} \big)^{-d/(\alpha+ \zeta)  }       .
\end{align*}
Since $\epsilon_n =   \sqrt{K}/ \sqrt{n} $ and $K = K_n  \asymp n^{\frac{d}{2[\alpha + \zeta] + d}}$, the bound in (\ref{concbound}) follows from observing that
\begin{align*}
  n^{\frac{d}{2(\alpha + \zeta)}}   \lessapprox K_n^{\frac{d}{2( \alpha + \zeta)}}  n^{\frac{d}{2(\alpha + \zeta)}} \asymp K_n^{1 + \frac{d}{\alpha + \zeta} }.
\end{align*}
Now suppose the model is severely ill-posed so that $\sigma_i \asymp \exp(-R i^{\zeta/d} ) $ for some $R, \zeta \geq 0$. It follows from \citep[Lemma 5.1]{ray} that    \begin{align*}
- \log \mathbb{P} \bigg(  \| G_{\alpha}   \|_{w,\sigma} < \epsilon_n \sqrt{ K}           \bigg) \leq C     \bigg \{  \log  \bigg ( \frac{1}{  \epsilon_{n}  \sqrt{ K}   }    \bigg )     \bigg \}^{1 + \frac{d}{\zeta} }        .
\end{align*}
Since $\log( (\epsilon _n \sqrt{K})^{-1}  ) \asymp \log(n) $ and  $K = K_n \asymp (\log n)^{1 +d/\zeta}$, the bound in (\ref{concbound}) follows. Hence, we obtain $ \int_{\mathcal{R}_n} d \mu(h) \geq \exp(-C' n \epsilon_n^2)  $ for some $C' > 0$. Assumption \ref{loc-conc}$(i)$ follows. Moreover, we note that the constant $C'$ is independent of $M$. 

Since $d\mu(.)$ is the distribution of  $G_{\alpha} /  \sqrt{K} $ and $\alpha > \gamma$, it follows from Theorem 2.1.20 of \citet{gine2021mathematical}  that there exists a universal constant $D > 0 $ such that $$ \int_{\mathcal{R}_n \setminus \mathcal{S}_n} d \mu(h) \leq \int_{h : \| h \|_{\mathbf{H}^t} > M} d \mu(h) + \int_{h : \| h \|_{\mathcal{H}^\gamma} > M} d \mu(h) \leq4 \exp(-D M^2 n \epsilon_n^2).   $$
By picking $ M > 0$ large enough, we can ensure that  $DM^2  > C'$ and Assumption \ref{loc-conc}$(ii)$ follows. From the conclusion of Theorem \ref{main-contract}, we obtain $$ \mu \big ( h :  \| \widehat{m}(W,h)  - m(W,h_0)  \|_{L^2(\mathbb{P}_n,\widehat{\Sigma})}  > L \epsilon_n  \: \big| \: \mathcal{D}_n  \big ) \xrightarrow{\mathbb{P}} 0 .$$
Next, we aim to apply Corollary \ref{contract-strong} with the metric $d(h,h_0) = \| h - h_0 \|_{L^2(\mathbb{P})}$. Define $r_n = (\log K)^{-1}$ and for any fixed $E > 0$, define the set $$ \mathcal{G}_n = \{ h : \|h \|_{\mathbf{H}^t} \leq E \; , \; \| h \|_{\mathcal{H}^{\gamma}} \leq E \;, \; \| h \|_{\mathcal{H}^{\alpha - r_n}} \leq E r_n^{-1/2} \} . $$
By expressing $G_{\alpha} \stackrel{d}{=}  \sum_{i=1}^{\infty} \sqrt{\lambda_{i,\alpha}} Z_i e_i$ in its Karhunen-Loève expansion, we have    \begin{align*}
\E  \big( \| G_{\alpha} \|_{\mathcal{H}^{\alpha - r_n}}^2 \big) = \sum_{i=1}^{\infty} i ^{2(\alpha - r_n)/d} \lambda_i \; \; \; \; \text{where} \; \; \; \lambda_i \asymp i^{-1 - 2 \alpha/d}. 
\end{align*}
Therefore, from the definition of $r_n$, it follows that $    \E  \big( \| G_{\alpha} \|_{\mathcal{H}^{\alpha - r_n}}^2 \big)  \leq C    r_n^{-1}  $. Since $d\mu(.)$ is the distribution of  $G_{\alpha} /  \sqrt{K} $, it follows from Theorem 2.1.20 of \citet{gine2021mathematical}  that there exists a universal constant $D > 0 $ such that $ \mu(h \notin \mathcal{G}_n) \leq 6 \exp (-D E^2 n \epsilon_n^2)$. We can pick $E > 0$ large enough so as to satisfy the hypothesis of Corollary \ref{contract-strong}. Since the conditions of Corollary \ref{contract-strong} are satisfied, it only remains to verify the rate for the modulus $\omega_n$. Define the set $$ \mathcal{E}_n = \{ h \in \mathcal{G}_n : \| \widehat{m}(W,h)  - m(W,h_0)  \|_{L^2(\mathbb{P}_n,\widehat{\Sigma})} \leq L \epsilon_n  \} . $$
By arguing as in Theorem 1, we can deduce that $ \sup_{h \in \mathcal{E}_n} \| \Pi_K m(W,h) \|_{L^2(\mathbb{P})} \leq D \epsilon_n $ with $\mathbb{P} $ probability approaching $1$. It follows that  $$ \sup_{h \in \mathcal{E}_n} \| m(W,h) \|_{L^2(\mathbb{P})} \leq D  \max \bigg\{ \epsilon_n , \sup_{h \in \mathcal{G}_n   } \| (\Pi_K - I) m(W,h) \|_{L^2(\mathbb{P})}   \bigg\} .$$
As in Theorem 1, this implies that $\sup_{h \in \mathcal{E}_n} \|  h - h_0 \|_{L^2(\mathbb{P})} \leq \delta_n$ for some sequence $\delta_n \rightarrow 0$. In particular, for any $\epsilon > 0$, we have $\mathcal{E}_n \subseteq \{ h : \| h - h_0 \|_{L^2(\mathbb{P})} \leq \epsilon  \}$ asymptotically. Since $\epsilon_n = \sqrt{K} / \sqrt{n}$, Condition \ref{gp-link}-\ref{basis-approx} imply \begin{align*}
    \sup_{h \in \mathcal{E}_n} \| h- h_0 \|_{w,\sigma} \leq D \bigg(  \sqrt{K}n^{-1/2}  + \varphi(K) K^{-\alpha/d}  r_n^{-1/2} K^{r_n/d}    \bigg).
\end{align*}
By substituting the definition of $r_n$, we have $K^{r_n/d} = O(1)$, and \begin{align*}
    \sup_{h \in \mathcal{E}_n} \| h- h_0 \|_{w,\sigma} \leq D \bigg(  \sqrt{K}n^{-1/2}  + \varphi(K) K^{-\alpha/d} r_n^{-1/2}    \bigg).
\end{align*}
Since $h_0 \in \mathcal{H}^p$ and $\| h \|_{\mathcal{H}^{\alpha-r_n}} \leq D r_n^{-1/2}$, we have for all $h \in \mathcal{E}_n$, \begin{align*}
\|   h- h_0 \|_{L^2(\mathcal{X})}^2 =   \sum_{i=1}^{\infty}  \left|  \langle e_i , h - h_0 \rangle_{}     \right|^2  & =  \sum_{i=1}^{J}  \left|  \langle e_i , h - h_0 \rangle_{}     \right|^2 +  \sum_{i> J}^{}  \left|  \langle e_i , h - h_0 \rangle_{}   \right|^2 \\  &  \leq \big (\max_{i \leq J} \sigma_i^{-2}   \big) \sum_{i=1}^{\infty} \sigma_i^2  \left|  \langle e_i , h - h_0 \rangle_{}     \right|^2 + D  r_n^{-1}   J^{- 2 \alpha/d} J ^{2r_n/d} \\ & \leq D \bigg(   \max_{i \leq J} \sigma_i^{-2}    \|  h- h_0  \|_{w,\sigma}^2 +  r_n^{-1}   J^{- 2 \alpha/d} J ^{2r_n/d}    \bigg) 
\end{align*}
for all $J \geq 1$. From the preceding derived bounds, the last term on the right can be bounded as
\begin{align*}
& \bigg(   \max_{i \leq J} \sigma_i^{-2}    \|  h- h_0  \|_{w,\sigma}^2 +  r_n^{-1}   J^{- 2 \alpha/d} J ^{2r_n/d}    \bigg)  \\ & \leq D \bigg(  \max_{i \leq J} \sigma_i^{-2}   \bigg[ K n^{-1}  + \varphi^2(K) K^{-2\alpha/d}  r_n^{-1}         \bigg] +     J^{- 2 \alpha/d} J ^{2r_n/d}  r_n^{-1}     \bigg).
\end{align*}
It follows that \begin{align*}
  \sup_{h \in \mathcal{E}_n}   \|   h- h_0 \|_{L^2(\mathcal{X})}^2 \leq D  \inf_{J \geq 1}  \bigg(  \max_{i \leq J} \sigma_i^{-2}   \bigg[ K n^{-1}  + \varphi^2(K) K^{-2\alpha/d}  r_n^{-1}         \bigg] +     J^{- 2 \alpha/d} J ^{2r_n/d}  r_n^{-1}     \bigg).
\end{align*}
 In the mildly ill-posed case, we have $  \sigma_i \asymp i^{- \zeta/d}$ and $\varphi(K) \asymp K^{-\chi/d}$ for some $\chi,\zeta \geq 0$. Since $ K_n \asymp  n^{d/[2(\alpha + \zeta) + d]}$ satisfies $K_n n^{-1} \asymp K_n^{-2(\alpha + \zeta)/d}$, the preceding term reduces to \begin{align*}
     \sup_{h \in \mathcal{E}_n}   \|   h- h_0 \|_{L^2(\mathcal{X})}^2 \leq D  \inf_{J \geq 1}  \bigg[  J^{2 \zeta/d}   K_n n^{-1} \big( 1 + r_n^{-1}  K_n^{2 (\zeta - \chi)/d}    \big)   + J^{- 2 \alpha/d} J ^{2r_n/d}  r_n^{-1}        \bigg].
 \end{align*}
We pick $J = J_n$ to satisfy $J_n^{-2 (\alpha + \zeta)/d} \asymp  n^{-1}K_n^{1 + 2( \max \{\zeta - \chi ,0\})/d} $. This choice also ensures that $J_n^{2 r_n/d} = O(1)$. Since $ K_n \asymp  n^{d/[2(\alpha + \zeta) + d]}$, the implied rate is
 \begin{align*}
    \sup_{h \in \mathcal{E}_n} \|   h- h_0 \|_{L^2(\mathcal{X})}^2 \leq D n^{- \frac{2 \alpha}{2[\alpha + \zeta] + d} \frac{(\alpha +  \min \{ \zeta,\chi  \})}{(\alpha + \zeta)}} \log n .
\end{align*}
In the severely ill-posed case, we have $  \sigma_i \asymp  \exp( -R i^{\zeta/d})$ and $\varphi(K) \asymp \exp(-R' K^{\chi/d})$ for some $R,R',\zeta,\chi > 0$. Define $c' = \chi(d^{-1} + \zeta^{-1})  > 0$.  Since $ K_n \asymp (\log n)^{1+d/\zeta} $, we have $\varphi^2(K_n) \asymp \exp( - c (\log n)^{c'}   )$ for some $c > 0$. In this case, the choice $J = \lfloor (c_0 \log n)^{ \min \{ c' ,1 \} d/\zeta} \rfloor  $ for  a sufficiently small $c_0$ implies $J^{2 r_n/d} = O(1)$ and \begin{align*}
    \sup_{h \in \mathcal{E}_n} \|   h- h_0 \|_{L^2(\mathcal{X})}^2 \leq D \big( \log n  \big)^{-2 \min \{ c',1  \} \alpha/\zeta} \log \log n.
\end{align*}
\end{proof}

\begin{lemma}
\label{emp-bloc}
Suppose Conditions \ref{residuals}, \ref{residuals2} and \ref{fsbasis}$(i)$ hold. Given functions $h(X), h'(X): \mathcal{X} \rightarrow \R $, define the differenced residual:
 \begin{align*}  R_{h-h'}^K (Z) =\big[ G_{b,K}^{-1/2}  b^K(W) \big] \big[  \big \{  \rho_{}(Y,h_{}(X)) - \rho_{}(Y,h_{}'(X)) \big \}   \big]_{l}  \; \; \; \; \; , \; \; \; \; \; \;   l \in \{1 , \dots , d_{\rho} \}\;,\end{align*}
where $[v]_{l}$ denotes the $l^{th}$ element of a vector $v$. Then, given any $ M > 0 $ and a  sequence $\delta_n \downarrow 0$, there exists a universal constant $D  = D(M) < \infty$ such that 
    \begin{align*}
&  \sqrt{n} \sup_{l \in \{1 , \dots ,  d_{\rho} \} }  \E\bigg(     \sup_{h,h'  \in \mathbf{H}^t(M)   : \|  h- h'   \|_{L^2(\mathbb{P})} \leq \delta_n  }   \|   \E_n [ R_{h-h',l}^{K}(Z) ] - \E [ R_{h-h',l}^{K}(Z) ]        \|_{\ell^2}  \bigg)  \\ & \leq D \bigg[   \frac{K^{3/2} \log(K)}{\sqrt{n}}  + \frac{\sqrt{K} \delta_n^{-d/t} }{\sqrt{n}}  + \sqrt{K} \sqrt{\log(K)} \delta_n^{\kappa} + \delta_n^{\kappa - d/(2t)}     \bigg].
\end{align*}

\end{lemma}

\begin{proof}[Proof of Lemma \ref{emp-bloc}]
It suffices to verify the bound for each  $l \in \{1 , \dots , d_{\rho} \}$ individually. Fix any such $l$. For ease of notation, we suppress the dependence on $l$ and denote the  vector by $   R_{h-h',l}^K(Z) = R_{h-h'}^K(Z)$. Observe that
 \begin{align*} &
\E_{} \bigg[ \sup_{h,h' \in \mathbf{H}^t(M)  : \|  h- h'   \|_{L^2(\mathbb{P})} \leq \delta_n  }  \|   \E_n [R_{h-h'}^K(Z)] - \E_{}[R_{h-h'}^K(Z)]       \|_{\ell^2}  \bigg] \\ & =  \frac{1}{\sqrt{n}} \E \bigg[  \sup_{h,h'  \in \mathbf{H}^t(M)  : \|  h- h'   \|_{L^2(\mathbb{P})} \leq \delta_n  }  \sup_{\gamma \in  \mathbb{S}^{K-1} }  \frac{1}{\sqrt{n}} \sum_{i=1}^n  \gamma' \big(  R_{h-h'}^K(Z_i) -\E_{} [ R_{h-h'}^K(Z) ] \:  \big)           \:  \bigg]  
\end{align*} 
where $\mathbb{S}^{K-1} = \{ v \in \R^K : \| v \|_{\ell^2}  = 1 \}$. Define the class of functions $$\mathcal{F}_K  = \{ \gamma' R_{h-h'}^K(Z)   : h,h'  \in \mathbf{H}^t(M)  \: , \: \| h-h'  \|_{L^2(\mathbb{P})} \leq \delta_n  \: , \:  \gamma \in   \mathbb{S}^{K-1}   \}    . $$
Denote the associated envelope function by $F_{K}(Z_i) =   \sup_{f \in \mathcal{F}_K} \left|   f(Z_i)  \right|   $.  Let $C_4(M) < \infty $ be as in Condition \ref{residuals2}(iii).  By Cauchy-Schwarz, it follows that \begin{align*}
F_K(Z_i) & \leq  \sup_{\gamma \in \mathbb{S}^{K-1}} \left| \gamma' G_{b,K}^{-1/2} b^K(W) \right| \sup_{h,h' \in \mathbf{H}^t(M), \|h - h'  \|_{L^2(\mathbb{P}) \leq \delta_n}  } \left|  \rho_{l}(Y,h_{}(X)) - \rho_{l}(Y,h_{}'(X)) \right|  \leq C_4 \zeta_{b,K}.
\end{align*}
where $\zeta_{b,K} =  \sup_{w \in \mathcal{W} }  \|  G_{b,K}^{-1/2} b^K(w) \|_{\ell^2} $. 

Let $C_1(M) < \infty$ and $\kappa \in (0,1]$ be as in Condition \ref{residuals}.  For any fixed $\gamma \in \mathbb{S}^{K-1}$, we have that  \begin{align*}
  &  \sup_{h,h' \in \mathbf{H}^t(M), \|h - h'  \|_{L^2(\mathbb{P}) \leq \delta_n}  }  \E \big[  \left|  \gamma' R_{h-h'}^K(Z)   \right|^2     \big]      \\& =   \sup_{h,h' \in \mathbf{H}^t(M), \|h - h'  \|_{L^2(\mathbb{P}) \leq \delta_n}  }   \E \big[   \gamma' G_{b,K}^{-1/2} b^K(W) b^K(W)'  G_{b,K}^{-1/2} \gamma   \left|  \rho_{l}(Y, h(X))  - \rho_{l}(Y, h'(X))   \right|^2     \big]  \\ & \leq C_1^2 \delta_n^{2 \kappa}  \gamma' G_{b,K}^{-1/2} \E \big[ b^K(W) b^K(W)' \big] G_{b,K}^{-1/2} \gamma     \\ &  = C_1^2 \delta_n^{2 \kappa}.
\end{align*}
For ease of exposition in the remainder of the proof, define $ \sigma_n = \delta_n^{\kappa} $. From the preceding bound, it follows that  $
    \sup_{f \in \mathcal{F}_K}  \|  f \|_{L^2(\mathbb{P})} \leq C_1  \sigma_n$. By an application of \citep[Proposition 3.5.15]{gine2021mathematical},  there exists a universal constant $ L > 0 $ such that \begin{align*}
& \E_{} \bigg[\sup_{h,h' \in \mathbf{H}^t(M) : \|  h- h'   \|_{L^2(\mathbb{P})} \leq \delta_n  }  \|   \E_n [R_{h-h'}^K(Z)] - \E_{}[R_{h-h'}^K(Z)]       \|_{\ell^2}  \bigg] \\ &   \leq \frac{L}{\sqrt{n}}  \int_{0}^{ 2 \sigma_n  }  \sqrt{\log N_{[]}(\mathcal{F}_K, \| . \|_{L^2(\mathbb{P})},  \epsilon  ) } d \epsilon \bigg( 1 +         \frac{\zeta_{b,K}}{\sigma_n^2 \sqrt{n}} \int_{0}^{ 2 \sigma_n  }  \sqrt{\log N_{[]}(\mathcal{F}_K, \| . \|_{L^2(\mathbb{P})},  \epsilon  ) } d \epsilon  \bigg).
\end{align*}
Fix any $\delta > 0$. Let $\{  h_i\}_{i=1}^{T_1}  $  denote a  $\delta$ covering of  $\big( \mathbf{H}^t(M)    , \| . \|_{\infty}  \big)$ and $\{ \gamma_{m} \}_{ m=1}^{T_2} $ denote a $\delta$ covering of $(\mathbb{S}^{K-1} , \| . \|_{\ell^2})$. For $i,j \in \{1 , \dots , T_1 \}$ and  $m \in \{ 1 , \dots , T_2  \}$, define the functions   \begin{align*} e_{i,j,m}(Z) & =   \sup_{ \gamma \in \mathbb{S}^{K-1} : \| \gamma - \gamma_{m}   \|_{\ell^2} < \delta \; , \; h \in \mathbf{H}^t(M) \; , \; h' \in \mathbf{H}^t(M)  } \left| (\gamma - \gamma_m)' [  R_{h}^K(Z) - R_{h'}^K(Z)    ]   \right| \\ & + \sup_{\gamma \in \mathbb{S}^{K-1} \; , \; h \in \mathbf{H}^t(M) :  \| h- h_i   \|_{\infty} < \delta }  \left|  \gamma' [    R_{h}^K(Z) - R_{h_i}^K(Z)       ]        \right| \\ & +  \sup_{\gamma \in \mathbb{S}^{K-1} \; , \; h \in \mathbf{H}^t(M) :  \| h- h_j   \|_{\infty} < \delta }  \left|  \gamma' [    R_{h}^K(Z) - R_{h_j}^K(Z)       ]  \right| .
\end{align*}
Observe that $$ \bigg \{   \gamma_m ' [ R_{h_i}^K(Z) - R_{h_j}^K(Z)  ]    - e_{i,j,m} \; ,  \;   \gamma_m ' [ R_{h_i}^K(Z) - R_{h_j}^K(Z)  ]   + e_{i,j,m}   \bigg \}_{(i,j) \in \{1 , \dots , T_1 \}  \; , m \in \{ 1 , \dots , T_2 \}  }  $$ is a bracket covering for $\mathcal{F}_K$. Let $C_2(M) < \infty$ be as in Condition \ref{residuals2}(i).  By Cauchy-Schwarz: \begin{align*}
 \| e_{i,j,m} \|_{L^2(\mathbb{P})}  &   \leq   \bigg \|  \sup_{ \gamma \in \mathbb{S}^{K-1} : \| \gamma - \gamma_{m}   \|_{\ell^2} < \delta \; , \; h \in \mathbf{H}^t(M) \; , \; h' \in \mathbf{H}^t(M)  } \left| (\gamma - \gamma_m)' [  R_{h}^K(Z) - R_{h'}^K(Z)    ]   \right|   \bigg \|_{L^2(\mathbb{P})} \\ & + \bigg \|   \sup_{\gamma \in \mathbb{S}^{K-1} \; , \; h \in \mathbf{H}^t(M) :  \| h- h_i   \|_{\infty} < \delta }  \left|  \gamma' [    R_{h}^K(Z) - R_{h_i}^K(Z)       ]        \right|          \bigg \|_{L^2(\mathbb{P})} \\ & + \bigg \|     \sup_{\gamma \in \mathbb{S}^{K-1} \; , \; h \in \mathbf{H}^t(M) :  \| h- h_j   \|_{\infty} < \delta }  \left|  \gamma' [    R_{h}^K(Z) - R_{h_j}^K(Z)       ]  \right|       \bigg \|_{L^2(\mathbb{P})} \\ & \leq 2  \delta  \zeta_{b,K} C_2 +  \delta^{\kappa} \zeta_{b,K} C_1 + \delta^{\kappa} \zeta_{b,K} C_1. 
\end{align*}
In particular, for all $\delta \in (0,1]$, we have that $ \| e_{i,j,m} \|_{L^2(\mathbb{P})}  \leq C \delta^{\kappa} \zeta_{b,K}  $ for $C = 2 C_2 + 2 C_1$. By \citep[Proposition C.7]{ghosal2017fundamentals}, we have  $  \log N_{}( \mathbf{H}^t(M),  \| . \|_{ \infty}, \epsilon) \lessapprox  \epsilon^{-d/t}$ as $\epsilon \downarrow 0$. By Condition \ref{fsbasis}$(i)$, we have $\zeta_{b,K}  \lessapprox \sqrt{K}$. Since $ \log N(\mathbb{S}^{K-1},\| .\|_{\ell^2}, \epsilon ) \leq K \log (3 \epsilon^{-1} )  $, it follows that there exists a universal constant $ L > 0 $ such that  \begin{align*}
& \int_{0}^{ 2 \sigma_n  }  \sqrt{\log N_{[]}(\mathcal{F}_K, \| . \|_{L^2(\mathbb{P})},  \epsilon  ) } d \epsilon \\ & \leq L \bigg( \sqrt{K} \sqrt{\log \zeta_{b,K} } \sigma_n  + \sqrt{K} \int_{0}^{2 \sigma_n}  \sqrt{\log(\epsilon^{-1})} d \epsilon  + \int_{0}^{2 \sigma_n}  \epsilon^{-d/2 \kappa t}         d \epsilon  \bigg) \\ & \leq L \bigg(   \sqrt{K} \sqrt{\log \zeta_{b,K} } \sigma_n  + \sqrt{K} \sigma_n \sqrt{\log (\sigma_n^{-1})}   + \sigma_n^{1- d/ (2 \kappa t)}   \bigg)  \\ & \leq L \bigg( \sqrt{K} \sqrt{\log K} \sigma_n +   \sigma_n^{1- d/ (2 \kappa t)}      \bigg).
\end{align*}
From the preceding bounds, it follows that   \begin{align*} 
 & \sqrt{n} \E_{} \bigg[\sup_{h,h' \in \mathbf{H}^t(M) : \|  h- h'   \|_{L^2(\mathbb{P})} \leq \delta_n  }  \|   \E_n [R_{h-h'}^K(Z)] - \E_{}[R_{h-h'}^K(Z)]       \|_{\ell^2}  \bigg] \\ &   \leq L  \int_{0}^{ 2 \sigma_n  }  \sqrt{\log N_{[]}(\mathcal{F}_K, \| . \|_{L^2(\mathbb{P})},  \epsilon  ) } d \epsilon \bigg( 1 +         \frac{\zeta_{b,K}}{\sigma_n^2 \sqrt{n}} \int_{0}^{ 2 \sigma_n  }  \sqrt{\log N_{[]}(\mathcal{F}_K, \| . \|_{L^2(\mathbb{P})},  \epsilon  ) } d \epsilon  \bigg)  \\ &   \lessapprox \bigg( \sqrt{K} \sqrt{\log K} \sigma_n +   \sigma_n^{1- d/ (2 \kappa t)}      \bigg) + \bigg( \sqrt{K} \sqrt{\log K} \sigma_n +   \sigma_n^{1- d/ (2 \kappa t)}      \bigg)^2 \frac{\sqrt{K}}{ \sigma_n^2 \sqrt{n}}. 
 \end{align*}
By substituting back $\sigma_n = \delta_n^{\kappa}$, the preceding term reduces to  
  \begin{align*}
 \bigg( \sqrt{K} \sqrt{\log K} \delta_n^{\kappa} +  \delta_n^{\kappa- d/(2t) }            \bigg)  + \bigg( \sqrt{K} \sqrt{\log K} \delta_n^{\kappa}  +  \delta_n^{\kappa- d/(2t) }        \bigg)^2  \frac{\sqrt{K}}{ \delta_n^{2 \kappa}  \sqrt{n}}.
\end{align*}
\end{proof}
\begin{lemma}
\label{empsq}
Suppose Condition \ref{fsbasis}$(i)$ holds. For each realization of $W$, let $\Sigma(W)$ denote a positive definite matrix such that $\mathbb{P}(  \| \Sigma(W) \|_{op} \leq C  ) = 1$ for some $C > 0$. Given any fixed $M > 0 $ and sequences $\delta_n, \gamma_n \downarrow 0$, define the set  \begin{align*}
\Theta_n = \{  h \in \mathbf{H}^t(M) :   \E (  \|  \Pi_K m(W,h)  \|_{\ell^2}^2 ) \leq M  \gamma_n ^2 , \; \| h - h_0  \|_{L^2(\mathbb{P})} \leq M \delta_n  \}.
\end{align*}
Then, there exists a universal constants $D,R < \infty$ such that \begin{align*}
     & \E \bigg( \sup_{h \in  \Theta_n}  \left|
\sum_{i=1}^n \bigg \{ [\Pi_K m(W_i,h)]' \Sigma(W_i) [\Pi_K m(W_i,h)]  - \E  \big( [\Pi_K m(W,h)]' \Sigma(W) [\Pi_K m(W,h)] \big) \bigg \} \right| \bigg) \\ & \leq R \bigg[  \sqrt{n} \gamma_n ^2 K \mathcal{J}(K^{-1/2}) + \gamma_n ^2 K^3 \mathcal{J}^2(K^{-1/2})  \bigg]
\end{align*}
where $\mathcal{J}(.)$ is defined by \begin{align*}
   &  \mathcal{J}(c) =  \int_{0}^{c}  \sqrt{\log N(\mathcal{M}_n , \| . \|_{L^2(\mathbb{P})}, \tau  D \gamma_n   )  } d \tau \; \; \; \; \; \; \; \; \forall  \; c > 0 \\ & \mathcal{M}_n = \{ m(w,h) : h \in \Theta_n   \}.
\end{align*}

\end{lemma}
\begin{proof}[Proof of Lemma \ref{empsq}]
Define the class of functions \begin{align*}
\mathcal{F} = \{  g : g(.) = [\Pi_K m(.,h)]' \Sigma(.) [\Pi_K m(.,h)]  : h \in \Theta_n   \}.
\end{align*}
For every fixed $h \in \Theta_n$, we have that
\begin{align*}
  \Pi_{K}[ m(W,h)] &  = \sum_{i=1}^K c_{h,i} [  G_{b,K}^{-1/2} b^K(W) ]_{i} \; \;  , \; \; c_{h,i} =  \E \big[  \rho \big( Y , h(X)   \big) [  G_{b,K}^{-1/2} b^K(W) ]_{i}  \big]  \; ,  \end{align*}
where $[  G_{b,K}^{-1/2} b^K(W) ]_{i}$ denotes the $i^{th}$ element of the vector $  G_{b,K}^{-1/2} b^K(W) $.  For every $ l \in \{  1 , \dots , d_{\rho} \} $, denote by $c_{h}^l$ the coefficient vector \begin{align*}
c_{h}^l = \big \{    \E \big[  \rho_{l} \big( Y , h(X)   \big) [  G_{b,K}^{-1/2} b^K(W) ]_{i}  \big]                    \big   \}_{i=1}^K.
\end{align*}
Observe that  $ \sum_{l=1}^{d_{\rho}}  \| c_h^l \|_{\ell^2}^2 =  \E (  \|  \Pi_K m(W,h)  \|_{\ell^2}^2 ) $. Let $C > 0$ be such that $ \mathbb{P} \big( \| \Sigma(W) \|_{op} \leq C \big) = 1  $. By Cauchy-Schwarz and the definition of $\Theta_n$,  it follows that \begin{align*}
\sup_{g \in \mathcal{F}} \left|    g(W)     \right| \leq C  \sup_{h \in \Theta_n} \| \Pi_K m(W,h)    \|_{\ell^2}^2 &  \leq C \zeta_{b,K}^2  \sum_{l=1}^{d_{\rho}} \|   c_h^l \|_{\ell^2}^2  \\ & = C \zeta_{b,K}^2  \E (  \|  \Pi_K m(W,h)  \|_{\ell^2}^2 ) \\ & \leq C M \zeta_{b,K}^2 \gamma_n ^2.
\end{align*}
From the estimate $\zeta_{b,K} \lessapprox \sqrt{K}$, it follows that $\sup_{g \in \mathcal{F}} \left|  g(W)  \right|  \leq C \gamma_n ^2 K  $ for some constant $C < \infty$. It follows that we can take  $F = C \gamma_n ^2 K$ to be an envelope of $\mathcal{F}$.
From this bound and the definition of $\Theta_n$,  we also obtain \begin{align*}
\sup_{g \in \mathcal{F}} \E[ g^2(W) ] \leq F \sup_{g \in \mathcal{F}} \E[\left| g(W) \right|] 
 \lessapprox F  \sup_{h \in \Theta_n}  \E (  \|  \Pi_K m(W,h)  \|_{\ell^2}^2 ) \lessapprox \gamma_n^4 K.
\end{align*} 
From similar arguments to those employed above, we have for every fixed $h,h' \in \Theta_n$,  the bound \begin{align*}     
&   \sup_{w}  \left| [\Pi_K m(w,h)]' \Sigma(w) [\Pi_K m(w,h)] - [\Pi_K m(w,h')]' \Sigma(w) [\Pi_K m(w,h')] \right| \\ & \lessapprox    \sup_{w} \sup_{g \in \Theta_n}  \left| [\Pi_K m(w,h) - \Pi_K m(w,h')]' \Sigma(w) [\Pi_K m(w,g)]  \right|       \\ & \lessapprox  \sqrt{F}  \sup_{w} \|  \Pi_K m(w,h) - \Pi_K m(w,h')    \|_{\ell^2}         \\ & \lessapprox \gamma_n   K  \sqrt{   \E \big(  \|  \Pi_K m(W,h) - \Pi_K m(W,h')      \|_{\ell^2}^2 \big) }.\\ & \lessapprox \gamma_n   K  \sqrt{   \E \big(  \|   m(W,h) -  m(W,h')      \|_{\ell^2}^2 \big) }.
\end{align*}
In particular,  there exists a universal constant $c >0 $ such that
 \begin{align*} \sup_{Q}
 \log N(  \mathcal{F},   \| . \|_{L^2(Q)}  , \tau F  )  \leq  \log N(\mathcal{M}, \| . \|_{L^2(\mathbb{P})} , c \tau \gamma_n   ) \; \; \; \;\; \; \;\forall \; \; \tau \in (0,1) \;\; ,
\end{align*}
where the supremum is over all discrete probability measures $Q$ on $\mathcal{W}$. From an application of \citep[Theorem 3.5.4]{gine2021mathematical}, it follows that \begin{align*}  \E \bigg( \sup_{g \in \mathcal{F}}  \left|
\sum_{i=1}^n g(W_i) - \E g(W) \right| \bigg)  \lessapprox \sqrt{n} \gamma_n^2 K \mathcal{J}(K^{-1/2}) + \gamma_n^2 K^3 \mathcal{J}^2(K^{-1/2}) .
\end{align*}
\end{proof}

\begin{proof}[Proof of Theorem \ref{bvm}]
 Given a  positive semi-definite matrix $\Sigma \in \R^{\rho \times \rho}$, we denote the inner product and norm induced by $\Sigma$ as $ \langle v, w \rangle_{\Sigma} = v' \Sigma w $ and $\| v \|_{\Sigma}^2 = v' \Sigma v$, respectively. With this notation, the quasi-Bayes posterior can be expressed as  \begin{align*}  \mu(.| 
\:\mathcal{D}_n) = \frac{   \exp\big(    - \frac{n}{2}  \E_n \big( \|  \widehat{m}(W,.)  \|_{\widehat{\Sigma}(W)}^2    \big)      \big) d \mu(.)  }{\int_{} \exp\big(    - \frac{n}{2} \E_n \big( \|  \widehat{m}(W,h)  \|_{\widehat{\Sigma}(W)}^2    \big)      \big) d \mu(h) } . 
  \end{align*}
For notational convenience, given two functions $h,g : \mathcal{X} \rightarrow \R$, we define the pairwise difference in the empirical estimate and its projection as:
\begin{align*}
 \widehat{m}(W, h, g) = \widehat{m}(W, h) - \widehat{m}(W, g) \; \; \; , \; \; \; \Pi_K m(W, h, g) := \Pi_K m(W, h) - \Pi_K m(W, g).
\end{align*}
Given a function $h : \mathcal{X} \rightarrow \R$ and $t \in \R$, we denote by $h_t$ the function: \begin{align*}
    h_t = h - \frac{ t}{\sqrt{n}} \tilde{\Phi}.
\end{align*}
Given a vector $v \in \R^n$, we denote the least squares projection of $v$ onto the subspace spanned by $ \{  b_1(W_i) , \dots , b_K(W_i)  \}_{i=1}^n $ by $\widehat\Pi_{K}[v]$. In particular, for every $h: \mathcal{X} \rightarrow \R$ and $l \in \{1 , \dots , d_{\rho}  \}$, we have
\begin{align}
 \label{emp-proj-op}    \widehat{\Pi}_K [  \{ \rho_{l}(Y_i,h(X_i)  \}_{i=1}^n ] =   \{ \widehat{m}(W_i,h) \}_{i=1}^n.
\end{align}
The RKHS $(\mathbb{H}_n, \| \cdot \|_{\mathbb{H}_n})$ associated to the Gaussian  random element  $G_{\alpha} / \sqrt{K} $ is \begin{align*}
  \mathbb{H}_{n}  =  \bigg \{  h \in L^2(\mathcal{X})  : \|h \|_{\mathbb{H}_{n}}^2 =  K \sum_{i=1}^{\infty}     i^{ 1 +   2\alpha /d}  \left| \langle h ,  e_i \rangle_{L^2(\mathcal{X})}      \right|^2 < \infty               \bigg \} .
\end{align*}
Let \( \epsilon_n = \sqrt{K_n} / \sqrt{n} \). Define the sequences
\[
r_n =
\begin{cases}
(\log n)^{-1} & \text{if mildly ill-posed}, \\
(\log \log n)^{-1} & \text{if severely ill-posed},
\end{cases}
\quad \text{and} \quad
\delta_n =
\begin{cases}
n^{- \frac{\alpha}{2[\alpha + \zeta] + d}} \sqrt{\log n} & \text{if mildly ill-posed}, \\
(\log n)^{- \alpha / \zeta} \sqrt{\log \log n} & \text{if severely ill-posed}.
\end{cases}
\]

Given any $D,M > 0$, we define the set $\Theta_n = \Theta_n(D,M)$ by \begin{align*}
 \Theta_n  =  \bigg \{ h \in \mathbf{H}^t(M)   :&  \|   m(W,h_{})   \|_{L^2(\mathbb{P})}  \leq   D r_n^{-1/2}  \epsilon_{n},  \;  \E_n \big( \| \widehat{m}(W,h) \|_{\ell^2}^2  \big)  \leq   D^2    \epsilon_{n}^2, \\ \nonumber &   \| \Pi_K m(W,h) \|_{L^2(\mathbb{P})}^2 \leq D^2 \epsilon_n^2 \;, \:    \| h- h_0 \|_{L^2(\mathbb{P})} \leq D \delta_n \;, \big|  \langle h ,  \tilde{\Phi}   \rangle_{\mathbb{H}_n}    \big| \leq M \sqrt{n}  \epsilon_n \|  \tilde{\Phi}  \|_{\mathbb{H}_n}, \\ \nonumber   &  \|h \|_{\mathcal{H}^{\alpha-r_n}} \leq M r_n^{-1/2} , \| D_{h_0}[h-h_0]  \|_{L^2(\mathbb{P})} \leq D r_n^{-1/2} \epsilon_n    \bigg \}.
\end{align*}
The proof proceeds through several steps which we outline below.
\begin{enumerate}
    \item[$(i)$] 
    From the proof of Theorem \ref{rate} and an application of  \citet[Theorem 2.1.20]{gine2021mathematical} to the Gaussian random variable $ Z_n = \langle h ,  \tilde{\Phi}   \rangle_{\mathbb{H}_n}$, we can choose $D,M > 0 $ large enough such that \begin{align*}   \mu(\Theta_n^c|\mathcal{D}_n)  \leq R'e^{-R n \epsilon_n^2 }    \end{align*}
holds with $\mathbb{P}$ probability approaching $1$, where $R,R' > 0$ are universal constants (that depends on $D,M$). In particular, since $n \epsilon_n^2 \uparrow \infty$, we have $\mu(\Theta_n^c|\mathcal{D}_n) \xrightarrow{\mathbb{P}} 0$. Denote the localized posterior obtained by restricting $\mu(.|\mathcal{D}_n)$ to $\Theta_n$ by \begin{align*}  \mu^{\star}(A\:|\:\mathcal{D}_n)  = \frac{  \int_{A \cap \Theta_n} \exp\big(    - \frac{n}{2}  \E_n \big[ \|  \widehat{m}(W,h)  \|_{\widehat{\Sigma}(W)}^2    \big]      \big) d \mu(h)  }{\int_{\Theta_n} \exp\big(    - \frac{n}{2} \E_n \big[ \|  \widehat{m}(W,h)  \|_{\widehat{\Sigma}(W)}^2    \big]      \big) d \mu(h) } \;,
\end{align*}
for every Borel set $A$. 

If $\| . \|_{TV}$ denotes the classical total variation metric on probability measures, it is straightforward to verify that
  \begin{align*}
\|  \mu(.|\mathcal{D}_n)  - \mu^{\star} (.|\mathcal{D}_n)    \|_{TV} \leq 2  \mu(\Theta_n^c|\mathcal{D}_n) \xrightarrow{\mathbb{P}} 0. 
\end{align*}
In particular, to deduce the desired weak convergence claims of the theorem, it suffices to work with the localized posterior measure $\mu^{\star}(.|\mathcal{D}_n)$. 

 \item[$(ii)$]
Let $\Sigma_0(.)$ denote the limiting weighting matrix in Condition \ref{limit-weight}.  We aim to verify that \begin{align*}
\sup_{h \in \Theta_n} \left| \E_n  \big(   \|   \widehat{m}(W,h,h_0)   \|_{\widehat{\Sigma}(W)}^2    \big) - \E \big \{  \Pi_K m(W,h) ' \Sigma_0(W) \Pi_K m(W,h)       \big \}                  \right| =o_{\mathbb{P}}(n^{-1}).
\end{align*}
To do this, we proceed through several steps. From the definition of $\Theta_n$, we have that \begin{align*}
  \sup_{h \in \Theta_n}   \left| \E_n  \big(  \|   \widehat{m}(W,h, h_0)   \|_{\widehat{\Sigma}(W)}^2    \big)  - \E_n  \big(  \|   \widehat{m}(W,h, h_0)   \|_{\Sigma(W)}^2    \big)     \right|   & \leq    \E_n  \big(  \|   \widehat{m}(W,h, h_0)   \|_{\ell^2}^2  \|  \widehat{\Sigma}(W)  - \Sigma_0(W)    \|_{op}    \big)   \\  & \leq   \sup_{w \in \mathcal{W}}  \|  \widehat{\Sigma}(w)  - \Sigma_0(w)    \|_{op}     \E_n  \big(  \|   \widehat{m}(W,h, h_0)   \|_{\ell^2}^2 \big) \\ & =   \epsilon_{n}^2   O_{\mathbb{P}} \bigg( \sup_{w \in \mathcal{W}}  \|  \widehat{\Sigma}(w)  - \Sigma_0(w)    \|_{op}   \bigg)    \\ & =  n^{-1}  O_{\mathbb{P}} \big(  \gamma_n  K_n   \big) \\ & = n^{-1} o_{\mathbb{P}}(1) .
\end{align*}
For any fixed $h: \mathcal{X} \rightarrow \R $, note that the estimator $\widehat{m}(w,h)$ can be expressed as \begin{align*}
     \widehat{m}(w,h)  =     \E_n \big(   \rho(Y,h_{}(X)) \big[ G_{b,K}^{-1/2} b^K(W) \big]'    \big)               [ \widehat{G}_{b,K}^{o}  ]^{-1} G_{b,K}^{-1/2}   b^K(w).
\end{align*}
In particular, this leads to the identity: \begin{align*}
   &  \E_n \big( \| \widehat{m}(W,h) \|_{\ell^2}^2  \big)   = \sum_{l=1}^{d_{\rho}}  [ \E_n (R_{h, l}^K)      ] ' [ \widehat G_{b,K}^o]^{-1}  [ \E_n (R_{h, l}^K)  ] \\ &  R_{h,l}^K(Z) =     \big[ G_{b,K}^{-1/2}  b^K(W) \big] \rho_{l}(Y,h_{}(X)). 
\end{align*}
By replacing $\widehat{G}_{b,K}^{o}$ with its asymptotic population analog $I_K$, we define \begin{align*}
     \widetilde{m}(w,h)  =     \E_n \big(   \rho(Y,h_{}(X)) \big[ G_{b,K}^{-1/2} b^K(W) \big]'    \big)             G_{b,K}^{-1/2}   b^K(w).
\end{align*}
For any $h$, observe that 
\begin{align*}
& \E_n \big( \| \widehat{m}(W,h) - \widetilde{m}(W,h)   \|_{\ell^2}^2  \big)  \leq  \bigg( \sum_{l=1}^{d_{\rho}}  [ \E_n (R_{h, l}^K)      ] '   [ \E_n (R_{h, l}^K)  ] \bigg) \| \big( [ \widehat{G}_{b,K}^{o}  ]^{-1} - I   \big)   \|_{op}^2 \| \widehat{G}_{b,K}^{o}   \|_{op}.
\end{align*}
With $\mathbb{P}$ probability approaching one, an application of Lemma \ref{aux2} implies that (i) the first term on the right hand side is bounded above by $\E_n \big( \| \widehat{m}(W,h) \|_{\ell^2}^2 \big)$, (ii) the second term is bounded by $K \log(K) n^{-1}$, and (iii) the third term is bounded by a constant, with all bounds holding up to a universal constant. 

By Condition \ref{limit-weight}, the eigenvalues of $\Sigma_0(W)$ are bounded above with probability $1$. Hence, by Cauchy-Schwarz and the definition of $\Theta_n$, it follows that \begin{align*}
& \sup_{h \in \Theta_n} \left| \E_n  \big[  \widehat{m}(W,h,h_0) \Sigma_0(W) \widehat{m}(W,h,h_0) \big] - \E_n \big[  \widetilde{m}(W,h,h_0) \Sigma_0(W) \widetilde{m}(W,h,h_0) \big] \right| \\ & = O_{\mathbb{P}} \bigg(  \sup_{h \in \Theta_n}  \sqrt{\E_n \|   \widehat{m}(W,h) - \widetilde{m}(W,h)   \|_{\ell^2}^2  } \sqrt{\E_n \|  \widehat{m}(W,h)  \|_{\ell^2}^2} \bigg)    \\ & = \epsilon_n^2 n^{-1/2} O_{\mathbb{P}}\big( \sqrt{K} \sqrt{\log K}   \big ) .
\end{align*}
Since $ \epsilon_n^2 = K / n $ and  $  K \sqrt{K \log K} / \sqrt{n} = o(1) $, the preceding term is $o_{\mathbb{P}}(n^{-1}).$ Next, observe that $\Pi_K m(w,h)$ can be expressed as   \begin{align*}
 \Pi_K m(w,h) =     \E \big(   \rho(Y,h_{}(X)) \big[ G_{b,K}^{-1/2}  b^K(W) \big]'    \big)               G_{b,K}^{-1/2}  b^K(w).
\end{align*}
By Lemma \ref{aux2}, \ref{emp-bloc} and Condition $\ref{misc1}(ii)$, there exists a sequence $\xi_n$ satisfying $\xi_n  \sqrt{K_n}  \downarrow 0$ such that \begin{align*}
   \sup_{h \in \Theta_n} \E_n \| \widetilde{m}(W,h,h_0) - \Pi_K m(W,h,h_0)      \|_{\ell^2}^2  &  \leq  \sup_{h \in \Theta_n} \bigg( \sum_{l=1}^{d_{\rho}}  \|  \E_n  (R_{h, l}^K) - \E  (R_{h, l}^K)  \|_{\ell^2}^2       \bigg)  \| \widehat{G}_{b,K}^{o}   \|_{op} \\ & = O_{\mathbb{P}} \big( n^{-1} \xi_n^2     \big).
\end{align*}
By Cauchy-Schwarz, it follows that \begin{align*}
& \sup_{h \in \Theta_n} \left| \E_n \big[  \widetilde{m}(W,h,h_0) \Sigma(W) \widetilde{m}(W,h,h_0) \big] - \E_n \big[  \Pi_K m(W,h,h_0) \Sigma(W) \Pi_K m(W,h,h_0) \big] \right| \\ & = O_{\mathbb{P}} \bigg(  \sup_{h \in \Theta_n}  \sqrt{\E_n \|   \widetilde{m}(W,h) - \Pi_K m(W,h)   \|_{\ell^2}^2  } \sqrt{\E_n \|  \widetilde{m}(W,h)  \|_{\ell^2}^2} \bigg)  \\ & =  O_{\mathbb{P}} \big( n^{-1/2} \xi_n  \epsilon_n    \big) \\ & = n^{-1} O_{\mathbb{P}} \big(  \xi_n  \sqrt{K}   \big) \\ & = n^{-1} o_{\mathbb{P}}(1).
\end{align*}
Next, by Lemma \ref{empsq}, we obtain \begin{align*}
 & \sup_{h \in \Theta_n}  \left| \E_n \big \{   \Pi_K m(W,h)' \Sigma_0(W) \Pi_K m(W,h) \big \}   - \E \big \{  \Pi_K m(W,h)' \Sigma_0(W) \Pi_K m(W,h)        \big \} \right|  \\ & = n^{-1} O_{\mathbb{P}} \big(  \sqrt{n} \epsilon_n^2 K \mathcal{J}(K^{-1/2} ) +\epsilon_n^2 K^3 \mathcal{J}^2(K^{-1/2})  \big)
\end{align*}
where $\mathcal{J}(.)$ is the entropy integral in (\ref{entropy-int}). Since $\epsilon_n = K/ \sqrt{n}$, this expression is $o_{\mathbb{P}}(n^{-1})$ by Condition \ref{misc1}$(i)$.

\item[$(iii)$]
We aim to verify that   \begin{align*}
 &  \sup_{h \in \Theta_n}  \left|     \E \big \{  \Pi_K m(W,h)' \Sigma_0(W) \Pi_K m(W,h)        \big \}  - \E \big (  \Pi_K D_{h_0}[h-h_0] ' \Sigma_0(W) \Pi_K  D_{h_0}[h-h_0]        \big )                         \right|  = o(n^{-1}).
\end{align*}
Denote the remainder obtained from linearizing the map at $h$ by \begin{align*} 
R_{h_0}(h,W) = m(W,h) - m(W,h_0) -  D_{h_0}[h- h_0] .
\end{align*}
We expand the deviation as: \begin{align*}
   & \E \big \{  \Pi_K m(W,h)' \Sigma_0(W) \Pi_K m(W,h) \big \}   - \E \big (  \Pi_K D_{h_0}[h-h_0] ' \Sigma_0(W) \Pi_K  D_{h_0}[h-h_0]        \big )   \\ & = \E \big [  \Pi_K  R_{h_0}(h,W) ' \Sigma_0(W) \Pi_K  R_{h_0}(h,W) \big ] + 2 \E  \big [  \Pi_K R_{h_0}(h,W)  ' \Sigma_0(W)  \Pi_K D_{h_0}[h- h_0]    \big ] .
 \end{align*}
Since the eigenvalues of $\Sigma_0(.)$ are uniformly  bounded above, Cauchy-Schwarz yields \begin{align*}
   & n \sup_{h \in \Theta_n}  \left| \E \big \{  \Pi_K m(W,h)' \Sigma_0(W) \Pi_K m(W,h) \big \}   - \E \big (  \Pi_K D_{h_0}[h-h_0] ' \Sigma_0(W) \Pi_K  D_{h_0}[h-h_0]        \big )    \right| \\  & \lessapprox  n  \sup_{h \in \Theta_n} \bigg[ \| \Pi_K R_{h_0}(h,W)  \|_{L^2(\mathbb{P})}^2  + \| \Pi_K R_{h_0}(h,W)  \|_{L^2(\mathbb{P})} \| \Pi_K D_{h_0}[h-h_0]   \|_{L^2(\mathbb{P})}  \bigg] \\ & \lessapprox n  \sup_{h \in \Theta_n}\bigg[   \| \Pi_K R_{h_0}(h,W)  \|_{L^2(\mathbb{P})}^2 + \| \Pi_K R_{h_0}(h,W)  \|_{L^2(\mathbb{P})} \sqrt{\log n} \epsilon_n    \bigg] \\ & = n \sup_{h \in \Theta_n} \bigg[   \| \Pi_K R_{h_0}(h,W)  \|_{L^2(\mathbb{P})}^2 + \| \Pi_K R_{h_0}(h,W)  \|_{L^2(\mathbb{P})} \sqrt{\log n} \sqrt{K}n^{-1/2}   \bigg] .    
\end{align*}
The preceding quantity is $o(1)$ by Condition \ref{misc1}$(iii)$.
\item[\textbf{$(iv)$}]
By repeating the argument from parts $(i-iii)$, we  similarly obtain (for every fixed $t \in \R$) the bound:
  \begin{align*}
& \sup_{h \in \Theta_n} \left| \E_n  \big(   \|   \widehat{m}(W,h_t,h_0)   \|_{\widehat{\Sigma}(W)}^2    \big) - \E \big (  \Pi_K D_{h_0}[h_t - h_0]  ' \Sigma_0(W) \Pi_K D_{h_0}[h_t - h_0]  \big )                  \right|  =o_{\mathbb{P}}(n^{-1}).
\end{align*}
\item[\textbf{$(v)$}]
Define  \begin{align} \label{smean}  S_n =    \E_n \big[  \langle  \rho(Y,h_0(X)), D_{h_0}[\tilde{\Phi} ](W)        \rangle_{\Sigma_0(W)}            \big].
\end{align}
For any fixed $t \in \R$, we aim to verify that \begin{align}
 \label{sn-verify} \sup_{h \in \Theta_n} \left| \E_n \big[   \langle  \widehat{m}(W,h_0)   ,  \widehat{m}(W,h,h_t)     \rangle_{\widehat{\Sigma}(W)}     \big] - \frac{t}{\sqrt{n}} S_n \right|   = o_{\mathbb{P}}(n^{-1}).
\end{align}
By a similar argument to parts $(i-iii)$, it is straightforward to verify that \begin{align*}
 & \sup_{h \in \Theta_n} \left| \E_n \big[   \langle  \widehat{m}(W,h_0)   ,  \widehat{m}(W,h,h_t)     \rangle_{\widehat{\Sigma}(W)}     \big] - \E_n \big[   \langle  \widehat{m}(W,h_0)   ,  \widehat{m}(W,h,h_t)     \rangle_{\Sigma_0(W)}     \big] \right|   = o_{\mathbb{P}}(n^{-1}) \\ & \sup_{h \in \Theta_n} \left| \E_n \big[   \langle  \widehat{m}(W,h_0)   ,  \widehat{m}(W,h,h_t)     \rangle_{\Sigma_0(W)}     \big] - \E_n \big[   \langle  \widehat{m}(W,h_0)   ,  \Pi_K m(W,h,h_t)     \rangle_{\Sigma_0(W)}     \big]            \right|  = o_{\mathbb{P}}(n^{-1}).
\end{align*}
By orthogonality of the least squares projection, we can write \begin{align*}
 \E_n \big[   \langle  \widehat{m}(W,h_0)   ,  \Pi_K m(W,h,h_t)     \rangle_{\Sigma_0(W)}     \big]  &  =  \E_n \big[   \langle  \widehat{m}(W,h_0)   ,  \Sigma_0(W) \Pi_K m(W,h,h_t)     \rangle_{}     \big] \\ & =   \E_n \big[   \langle  \rho(Y,h_0(X))  ,  \widehat{\Pi}_K \big[ \Sigma_0(W) \Pi_K m(W,h,h_t)  \big]   \rangle_{}     \big] \; ,
\end{align*}
where $\widehat{\Pi}_K$ is the empirical projection operator  in (\ref{emp-proj-op}). By interchanging $\E_n$ and the inner product, the preceding term can be written as an inner product of two vectors in $\R^{d_{\rho}}$. To be specific, from the preceding expansion, we can write:
\begin{align*}
& \E_n \big[   \langle  \widehat{m}(W,h_0)   ,  \Pi_K m(W,h,h_t)     \rangle_{\Sigma_0(W)}     \big] = \sum_{i=1}^{d_{\rho}}  V_i \; \; ,\\ &
V_l =  \E_n \big(   [ \Sigma_0(W) \Pi_K m(W,h,h_t)  ]_{l} \big[ G_{b,K}^{-1/2}b^K  (W) \big]'      \big) [ \widehat{G}_{b,K}^{o}  ]^{-1}    \frac{1}{n}  \sum_{i=1}^n   G_{b,K}^{-1/2}   b^K(W_i)                 \rho_{l}(Y_i,h_0(X_i)).
\end{align*} 
Similarly, we can express $\E_n \big[   \langle  \rho(Y,h_0(X))  , \Pi_K \big[ \Sigma_0(W) \Pi_K m(W,h,h_t)  \big]   \rangle_{}     \big] $ as $\sum_{i=1}^{d_{\rho}} \widetilde{V}_i $ where
\begin{align*}
  \widetilde{V}_l =   \E \big(   [ \Sigma_0(W) \Pi_K m(W,h,h_t)  ]_{l} \big[ G_{b,K}^{-1/2}b^K  (W) \big]'      \big)    \frac{1}{n}  \sum_{i=1}^n   G_{b,K}^{-1/2}   b^K(W_i)                 \rho_{l}(Y_i,h_0(X_i)).
\end{align*}
The $\|. \|_{\ell^2}$ norm of the sample average on the right is of order $ \sqrt{K} /  \sqrt{n}  $ (by Lemma \ref{emp-b}). As the eigenvalues of $\Sigma_0(.)$ are uniformly bounded above, a straightforward application of Lemma \ref{emp-bloc} and Condition \ref{misc1}$(ii)$  implies that \begin{align*}
\sup_{l =1,\dots,d_{\rho}}  \E  \bigg [  \sup_{h \in \Theta_n} \|  (\E_n - \E) \big(   [ \Sigma_0(W) \Pi_K m(W,h,h_t)  ]_{l} \big[ G_{b,K}^{-1/2}b^K  (W) \big]'      \big)      \|_{\ell^2} \bigg ]  \leq \frac{s_n}{\sqrt{n}}
\end{align*}
for some sequence $s_n $ satisfying $s_n \sqrt{K} \sqrt{\log K} \downarrow 0$. Furthermore, by Lemma \ref{aux2}, we have $ \| [ \widehat{G}_{b,K}^{o}  ]^{-1}  -I_K  \|_{op} \leq D  \sqrt{K \log(K)} /\sqrt{n}       $ for some universal constant $D$, with $\mathbb{P}$ probability approaching $1$. From combining the preceding bounds and an application of Cauchy-Schwarz, we obtain
\begin{align*}
& \sup_{h \in \Theta_n}  \left| \E_n \big[   \langle  \widehat{m}(W,h_0)   ,  \Pi_K m(W,h,h_t)     \rangle_{\Sigma_0(W)}     \big]  - \E_n \big[   \langle  \rho(Y,h_0(X))  , \Pi_K \big[ \Sigma_0(W) \Pi_K m(W,h,h_t)  \big]   \rangle_{}     \big]  \right| \\ & = o_{\mathbb{P}}(n^{-1}).
\end{align*}
Next, we write $m(W,h) = R_{h_0}(h,W) +D_{h_0}[h-h_0]$ and obtain the expansion: \begin{align*}
& \E_n \big[   \langle  \rho(Y,h_0(X))  , \Pi_K \big[ \Sigma_0(W) \Pi_K m(W,h,h_t)  \big]   \rangle_{}     \big] \\ & = \E_n \big[   \langle  \rho(Y,h_0(X))  , \Pi_K \big[ \Sigma_0(W) \Pi_K R_{h_0}(h,W)  \big]   \rangle_{}     \big] - \E_n \big[   \langle  \rho(Y,h_0(X))  , \Pi_K \big[ \Sigma_0(W) \Pi_K R_{h_0}(h_t,W)  \big]   \rangle_{}     \big] \\ & + \E_n \big[   \langle  \rho(Y,h_0(X))  , \Pi_K \big[ \Sigma_0(W) \Pi_K D_{h_0}[h- h_t]   \big]   \rangle_{}     \big].
\end{align*}
Similar to our bounds above, by interchanging $\E_n$ and the inner product, the first two terms on the right side of the equality can be analyzed through the terms: \begin{align*}
&  Q_{l,1} =  \E \big(   [ \Sigma_0(W) \Pi_K R_{h_0} (h,W)  ]_{l} \big[ G_{b,K}^{-1/2}b^K  (W) \big]'      \big)    \frac{1}{n}  \sum_{i=1}^n   G_{b,K}^{-1/2}   b^K(W_i)                 \rho_{l}(Y_i,h_0(X_i)) \; , \\ & Q_{l,2} =      -  \E \big(   [ \Sigma_0(W) \Pi_K R_{h_0} (h_t,W)  ]_{l} \big[ G_{b,K}^{-1/2}b^K  (W) \big]'      \big)    \frac{1}{n}  \sum_{i=1}^n   G_{b,K}^{-1/2}   b^K(W_i)                 \rho_{l}(Y_i,h_0(X_i)).
\end{align*}
The $\| . \|_{\ell^2}$ norm of the sample average on the right of both the preceding terms is of order $ \sqrt{K} /  \sqrt{n}  $ (by Lemma \ref{emp-b}). Furthermore, by the Bessel inequality, we obtain
\begin{align*}
  & \|  \E \big(   [ \Sigma_0(W) \Pi_K R_{h_0} (h,W)  ]_{l} \big[ G_{b,K}^{-1/2}b^K  (W) \big] \|_{\ell^2}^2 \leq  \|  [ \Sigma_0(W) \Pi_K R_{h_0} (h,W)  ]_{l}   \|_{L^2(\mathbb{P})}^2 \;, \\ & \|  \E \big(   [ \Sigma_0(W) \Pi_K R_{h_0} (h_t,W)  ]_{l} \big[ G_{b,K}^{-1/2}b^K  (W) \big] \|_{\ell^2}^2 \leq  \|  [ \Sigma_0(W) \Pi_K R_{h_0} (h_t,W)  ]_{l}   \|_{L^2(\mathbb{P})}^2.
\end{align*} 
As the eigenvalues of $\Sigma_0(.)$ are uniformly bounded above, the preceding bounds provide us with the expansion \begin{align*}
   & \E_n \big[   \langle  \rho(Y,h_0(X))  , \Pi_K \big[ \Sigma_0(W) \Pi_K m(W,h,h_t)  \big]   \rangle_{}     \big] \\ & =   \E_n \big[   \langle  \rho(Y,h_0(X))  , \Pi_K \big[ \Sigma_0(W) \Pi_K D_{h_0}[h- h_t]   \big]   \rangle_{}     \big]    \\ &   + \frac{\sqrt{K}}{\sqrt{n}} O_{\mathbb{P}} \bigg(  \sup_{h \in \Theta_n}   \| \Pi_K R_{h_0}(h,W)  \|_{L^2(\mathbb{P})}  + \sup_{h \in \Theta_n}  \| \Pi_K R_{h_0}(h_t,W)  \|_{L^2(\mathbb{P})}     \bigg)
\end{align*}
uniformly over $h \in \Theta_n$. Hence, by Condition \ref{misc1}$(iii)$, it follows that \begin{align*}
    & \E_n \big[   \langle  \rho(Y,h_0(X))  , \Pi_K \big[ \Sigma_0(W) \Pi_K m(W,h,h_t)  \big]   \rangle_{}     \big] \\ &  =  \E_n \big[   \langle  \rho(Y,h_0(X))  , \Pi_K \big[ \Sigma_0(W) \Pi_K D_{h_0}[h- h_t]   \big]   \rangle_{}     \big] + o_{\mathbb{P}}(n^{-1})
\end{align*}
uniformly over $h \in \Theta_n$. 

Note that, by construction $h - h_t =  t \tilde{\Phi} / \sqrt{n}  $. Since $D_{h_0}(.)$ is a linear operator, it follows that the preceding term can be expressed as \begin{align*}
\E_n \big[   \langle  \rho(Y,h_0(X))  , \Pi_K \big[ \Sigma_0(W) \Pi_K D_{h_0}[h- h_t]   \big]   \rangle_{}     \big] = \ \frac{t}{\sqrt{n}} \E_n \big[   \langle  \rho(Y,h_0(X))  , \Pi_K \big[ \Sigma_0(W) \Pi_K D_{h_0}[ \tilde{\Phi}  ]   \big]   \rangle_{}     \big] .
\end{align*}
Hence, to show (\ref{sn-verify}), it suffices to verify that \begin{align*} 
     \E_n \big[   \langle  \rho(Y,h_0(X))  , \Pi_K \big[ \Sigma_0(W) \Pi_K D_{h_0}[ \tilde{\Phi}  ]   \big]   \rangle_{}     \big] =    \E_n \big[   \langle  \rho(Y,h_0(X))  ,   \Sigma_0(W)  D_{h_0}[ \tilde{\Phi}  ]     \rangle_{}     \big] + o_{\mathbb{P}}(n^{-1/2}).
\end{align*}
To show this, we write the expression as \begin{align*}
    & \E_n \big[   \langle  \rho(Y,h_0(X))  , \Pi_K \big[ \Sigma_0(W) \Pi_K D_{h_0}[ \tilde{\Phi}  ]   \big]   \rangle_{}     \big] \\ & =  \E_n \big[ \langle  \rho(Y,h_0(X))  , ( \Pi_K - I)\big[ \Sigma_0(W) \Pi_K D_{h_0}[\tilde{\Phi}]  \big]   \rangle  \big] \\ & +   \E_n \big[  \langle  \rho(Y,h_0(X))  , \big[ \Sigma_0(W) (\Pi_K-I) D_{h_0}[\tilde{\Phi}]   \big]   \rangle  \big]  + \E_n \big[   \langle  \rho(Y,h_0(X))  ,   \Sigma_0(W)  D_{h_0}[ \tilde{\Phi}  ]     \rangle    \big].
\end{align*}
Since  $\E[\rho(Y,h_0(X))|W] = m(W,h_0) =  0$, the sample means appearing above are over mean zero random variables. Furthermore, since $\E\big( \| \rho(Y,h_0(X)) \|_{\ell^2}^2|W )$ is bounded above (with $\mathbb{P}$ probability $1$), we obtain \begin{align*}
  & n \E  \left|   \E_n \big[ \langle  \rho(Y,h_0(X))  , (\Pi_K -I)\big[ \Sigma_0(W) \Pi_K D_{h_0}[\tilde{\Phi}]   \big]   \rangle  ]  \right|^2 \\ & = \E \bigg( \left|  \langle  \rho(Y,h_0(X))  , (\Pi_K -I)\big[ \Sigma_0(W) \Pi_K D_{h_0}[\tilde{\Phi}]   \big]   \rangle       \right|^2    \bigg) \\ & \rightarrow 0
\end{align*}
because  $ \| (\Pi_K - I)  \Sigma_0(W) \Pi_K D_{h_0}[\tilde{\Phi}]    \|_{L^2(\mathbb{P})} \rightarrow 0 $ as $K \rightarrow \infty$. This is because $\Pi_K$ is a projection operator that approximates the identity (as $K \rightarrow \infty$) when acting on functions already in $L^2(W)$. Similarly, we obtain \begin{align*}
  & n \E  \left|   \E_n \big[  \langle  \rho(Y,h_0(X))  , \big[ \Sigma_0(W) (\Pi_K-I) D_{h_0}[\tilde{\Phi}]   \big]   \rangle  ]  \right|^2 \\ & = \E \bigg( \left|   \langle  \rho(Y,h_0(X))  , \big[ \Sigma_0(W) (\Pi_K-I) D_{h_0}[\tilde{\Phi}]   \big]   \rangle    \right|^2    \bigg) \\  & \rightarrow 0.
\end{align*}
The claim in (\ref{sn-verify}) follows from the preceding bounds.
\item[\textbf{$(vi)$}]
The preceding steps $(i-v)$ show that  \begin{align*}
& \E_n  \big(   \|   \widehat{m}(W,h)   \|_{\widehat{\Sigma}(W)}^2    \big) - \E_n  \big(    \|   \widehat{m}(W,h_t)   \|_{\widehat{\Sigma}(W)}^2    \big) \\ & =   \E_n \big(  \| \widehat{m}(W,h,h_0)     \|_{\widehat{\Sigma}(W)}^2         \big) - \E_n \big(  \| \widehat{m}(W,h_t,h_0)     \|_{\widehat{\Sigma}(W)}^2         \big) + 2 \E_n \big[   \langle  \widehat{m}(W,h_0)   ,  \widehat{m}(W,h,h_t)     \rangle_{\widehat{\Sigma}(W)}      \big] \\ & = \E \big( \| \Pi_K D_{h_0}[h-h_0]   \|_{\Sigma_0(W)}^2      \big) - \E \big( \| \Pi_K D_{h_0}[h_t-h_0]   \|_{\Sigma_0(W)}^2      \big) +2 \frac{t}{\sqrt{n}} S_n + o_{\mathbb{P}}(n^{-1})
\end{align*}
uniformly over $h \in \Theta_n$, where $S_n$ is as in (\ref{smean}). Furthermore, since $D_{h_0}(.)$ is a linear operator, we obtain \begin{align*}
& \frac{n}{2} \bigg[ \E \big( \| \Pi_K D_{h_0}[h-h_0]   \|_{\Sigma_0(W)}^2      \big) - \E \big( \| \Pi_K D_{h_0}[h_t-h_0]   \|_{\Sigma_0(W)}^2      \big) \bigg] \\ & = - \frac{t^2}{2 } \E \big( \|   \Pi_K  D_{h_0}[\tilde{\Phi}]     \|_{\Sigma_0(W)}^2   \big)  + t \sqrt{n} \E \big[  \langle  \Pi_K  D_{h_0}[  h - h_0   ],  \Pi_K D_{h_0}[\tilde{\Phi}]           \rangle_{\Sigma_0(W)}        \big] .
\end{align*}
For the first term, since $K \uparrow \infty$, continuity yields $$  - \frac{t^2}{2 } \E \big( \|   \Pi_K  D_{h_0}[\tilde{\Phi}]     \|_{\Sigma_0(W)}^2   \big) = - \frac{t^2}{2 } \E \big( \|    D_{h_0}[\tilde{\Phi}]     \|_{\Sigma(W)}^2   \big) + o(1). $$
For the second term, we expand it as \begin{align*}
& \E \big[  \langle  \Pi_K  D_{h_0}[  h - h_0   ],  \Pi_K D_{h_0}[\tilde{\Phi}]           \rangle_{\Sigma_0(W)}        \big]  \\ & = \E \big[  \langle  \Pi_K  D_{h_0}[  h - h_0   ],   \Pi_K \big\{ \Sigma_0(W)  \Pi_K D_{h_0}[\tilde{\Phi}] \big \}          \rangle        \big]  \\ & = \E \big[  \langle  \Pi_K  D_{h_0}[  h - h_0   ],   \Pi_K \big\{ \Sigma_0(W)  (\Pi_K-I) D_{h_0}[\tilde{\Phi}] \big \}          \rangle        \big]  +  \E \big[  \langle  \Pi_K  D_{h_0}[  h - h_0   ],   \Pi_K \big\{ \Sigma_0(W)   D_{h_0}[\tilde{\Phi}] \big \}          \rangle        \big] .
\end{align*}
Since the eigenvalues of $\Sigma_0(.)$ are  bounded above, Condition \ref{phi-reg} and Cauchy-Schwarz yields \begin{align*}
   &  \sup_{h \in \Theta_n}  \sqrt{n} \left| \E \big[  \langle  \Pi_K  D_{h_0}[  h - h_0   ],   \Pi_K \big\{ \Sigma_0(W)  (\Pi_K - I) D_{h_0}[\tilde{\Phi}] \big \}          \rangle        \big]  \right| \\ & \lessapprox \sqrt{n} \epsilon_n \sqrt{\log n} \| (\Pi_K - I) D_{h_0}[\tilde{\Phi}]   \|_{L^2(\mathbb{P})} \\ & = \sqrt{K} \sqrt{\log n}  \| (\Pi_K - I) D_{h_0}[\tilde{\Phi}]   \|_{L^2(\mathbb{P})} \\ & = o(1).
\end{align*}
Next, by orthogonality we have that \begin{align*}
    & \E \big[  \langle  \Pi_K  D_{h_0}[  h - h_0   ],   \Pi_K \big\{ \Sigma_0(W)   D_{h_0}[\tilde{\Phi}] \big \}          \rangle        \big] \\ & = \E \big[  \langle    D_{h_0}[  h - h_0   ],     \Sigma_0(W)   D_{h_0}[\tilde{\Phi}]           \rangle        \big] + \E \big[  \langle  (\Pi_K-I)  D_{h_0}[  h - h_0   ],   (\Pi_K-I) \big\{ \Sigma_0(W)   D_{h_0}[\tilde{\Phi}] \big \}          \rangle        \big].
\end{align*}
Similar to above, by Cauchy-Schwarz, we obtain \begin{align*}
  & \sup_{h \in \Theta_n}  \sqrt{n} \left|   \E \big[  \langle  (\Pi_K-I)  D_{h_0}[  h - h_0   ],   (\Pi_K-I) \big\{ \Sigma_0(W)   D_{h_0}[\tilde{\Phi}] \big \}          \rangle        \big]    \right| \\ & \lessapprox  \sqrt{n} \epsilon_n \sqrt{\log n} \| (\Pi_K - I) \Sigma_0(W) D_{h_0}[\tilde{\Phi}]    \|_{L^2(\mathbb{P})}  \\ & = \sqrt{K} \sqrt{\log n}  \| (\Pi_K - I) \Sigma_0(W) D_{h_0}[\tilde{\Phi}]    \|_{L^2(\mathbb{P})}  \\ & = o(1).
\end{align*}
From combining the preceding bounds, we obtain the expansion
\begin{align*}
&  \frac{-n}{2} \bigg[ \E_n  \big(   \|   \widehat{m}(W,h)   \|_{\widehat{\Sigma}(W)}^2    \big) - \E_n  \big(    \|   \widehat{m}(W,h_t)   \|_{\widehat{\Sigma}(W)}^2    \big) \bigg] \\ & =    \frac{t^2}{2 } \E \big( \|    D_{h_0}[\tilde{\Phi}]     \|_{\Sigma_0(W)}^2   \big) - t \sqrt{n} \E \big[  \langle    D_{h_0}[  h - h_0   ],        D_{h_0}[\tilde{\Phi}]           \rangle_{\Sigma_0(W)}        \big] - t \sqrt{n} S_n +  o_{\mathbb{P}}(1)
\end{align*}
uniformly over $h \in \Theta_n$. By definition of the adjoint $D_{h_0}^*$ and Condition \ref{phi-reg}, we can write \begin{align*}
    t \sqrt{n} \E  \big[ \langle    D_{h_0}[  h - h_0   ],        D_{h_0}[\tilde{\Phi}]           \rangle_{\Sigma_0(W)}        \big] &= t \sqrt{n}  \langle      h - h_0   ,       D_{h_0}^*  D_{h_0}[\tilde{\Phi}]           \rangle_{L^2(\mathbb{P})}     \\ & = t \sqrt{n}   \langle      h - h_0   ,       \Phi          \rangle_{L^2(\mathbb{P})} .
\end{align*}

\item[\textbf{$(vii)$}]
We compute the Laplace transform of the random variable $\sqrt{n} \big[ \langle h  - h_0 , \Phi \rangle_{L^2(\mathbb{P})} + S_n \big]$ where $h \sim \mu^\star(.\:|\:\mathcal{D}_n)$ and $S_n$ is as in (\ref{smean}). Fix any $t \in \R$.  From the conclusion of part $(vi)$, we can deduce that the Laplace transform admits the expansion: \begin{align*}
&  \E^\star  \bigg[ \exp \bigg \{t  \sqrt{n} \big[ \langle  h - h_0 , \Phi       \rangle_{L^2(\mathbb{P})} + S_n  \big] \bigg \}  \bigg| \:\mathcal{D}_n   \bigg] \\ &  = \frac{ \int_{\Theta_n} \exp \bigg \{t  \sqrt{n} \big[ \langle  h - h_0 , \Phi       \rangle_{L^2(\mathbb{P})} + S_n   \big] \bigg \}     \exp \bigg \{   -\frac{n}{2} \bigg[ \E_n  \big(   \|   \widehat{m}(W,h)   \|_{\widehat{\Sigma}(W)}^2    \big)  - \E_n  \big(   \|   \widehat{m}(W,h_t)   \|_{\widehat{\Sigma}(W)}^2    \big)      \bigg]    \bigg \} }{\int_{\Theta_n}  \exp \big(  -\frac{n}{2} \E_n  \big(   \|   \widehat{m}(W,h)   \|_{\widehat{\Sigma}(W)}^2    \big)    \big) d \mu (h) } \\ & \times \exp \bigg\{   -\frac{n}{2}    \E_n  \big(   \|   \widehat{m}(W,h_t)   \|_{\widehat{\Sigma}(W)}^2    \big)            \bigg\} d \mu(h) \\ & = \exp \bigg[ \frac{t^2}{2}         \E \big[ (D_{h_0} \tilde{\Phi})' \Sigma_0(W) (D_{h_0} \tilde{\Phi})      \big]     + o_{\mathbb{P}}(1)     \bigg] \times \frac{\int_{\Theta_n}  \exp \big(  -\frac{n}{2} \E_n  \big(   \|   \widehat{m}(W,h_t)   \|_{\widehat{\Sigma}(W)}^2    \big)    \big) d \mu (h)}{\int_{\Theta_n}  \exp \big(  -\frac{n}{2} \E_n  \big(   \|   \widehat{m}(W,h)   \|_{\widehat{\Sigma}(W)}^2    \big)    \big) d \mu (h)}.
\end{align*}
Next, we verify that  \begin{align*}
    \frac{\int_{\Theta_n}  \exp \big(  -\frac{n}{2} \E_n  \big(   \|   \widehat{m}(W,h_t)   \|_{\widehat{\Sigma}(W)}^2    \big)    \big) d \mu (h)}{\int_{\Theta_n}  \exp \big(  -\frac{n}{2} \E_n  \big(   \|   \widehat{m}(W,h)   \|_{\widehat{\Sigma}(W)}^2    \big)    \big) d \mu (h)} \xrightarrow{\mathbb{P}} 1.
\end{align*}
Let $ \mu_{t,\tilde{\Phi}}(h) $ denote the measure obtained from  translating $  \mu(\cdot) $ around $t \tilde{\Phi} / \sqrt{n}$. To be specific, $$ d \mu_{t,\tilde{\Phi}} \sim  \frac{G_{\alpha}}{\sqrt{K} } - \frac{t}{\sqrt{n}} \tilde{\Phi}.   $$
Since $\tilde{\Phi}$ is an element of the RKHS $\mathbb{H}$, it follows from \citep[Proposition I.20]{ghosal2017fundamentals} that $\mu_{t,\tilde{\Phi}}(\cdot)$ is absolutely continuous with respect to $\mu(\cdot)$ and admits a Radon–Nikodym density \begin{align} \label{gauss-cov}  \frac{d \mu_{t,\tilde{\Phi}}(h)}{d \mu(h)}  = \exp \bigg \{ \frac{t}{\sqrt{n}}  \langle  h , \tilde{\Phi}   \rangle_{\mathbb{H}_n}   - \frac{t^2}{2n} \|  \tilde{\Phi}   \|_{\mathbb{H}_n}^2    \bigg \} . \end{align}
From the definition of $\Theta_n$, we have   \begin{align*}
    \sup_{h \in \Theta_n} \left| \frac{t}{\sqrt{n}}  \langle  h , \tilde{\Phi}   \rangle_{\mathbb{H}_n} \right| \lessapprox  \epsilon_n \|  \tilde{\Phi}  \|_{\mathbb{H}_n}  & = \epsilon_n  \sqrt{K} \|  \tilde{\Phi}  \|_{\mathbb{H}} \; , 
\end{align*}
where we used the fact that $  \|  \tilde{\Phi}  \|_{\mathbb{H}_n} = \sqrt{K }  \|  \tilde{\Phi}  \|_{\mathbb{H}}$. It follows that  \begin{align*}
   &  \sup_{h \in \Theta_n} \left| \frac{t}{\sqrt{n}}  \langle  h , \tilde{\Phi}   \rangle_{\mathbb{H}_n} \right| \lessapprox \frac{K }{\sqrt{n}} = o(1) \; \; ,\: \:  \frac{t^2}{2n} \|  \tilde{\Phi}  \|_{\mathbb{H}_n}^2 \lessapprox  \frac{K }{\sqrt{n}} = o(1).
\end{align*}
Define the translated set: $$ \Theta_{n,\tilde{\Phi}} = \Theta_n - \frac{t}{\sqrt{n}} \tilde{\Phi} = \bigg \{ g : g= h - \frac{t}{\sqrt{n}} \tilde{\Phi} \; , h \in \Theta_n      \bigg \} .  $$
By the Gaussian change of variables in (\ref{gauss-cov}) and the preceding bounds, we obtain
 \begin{align*}
    &   \frac{\int_{\Theta_n}  \exp \big(  -\frac{n}{2} \E_n  \big(   \|   \widehat{m}(W,h_t)   \|_{\widehat{\Sigma}(W)}^2    \big)    \big) d \mu (h)}{\int_{\Theta_n}  \exp \big(  -\frac{n}{2} \E_n  \big(   \|   \widehat{m}(W,h)   \|_{\widehat{\Sigma}(W)}^2    \big)    \big) d \mu (h)}   = e^{o(1)} \frac{\mu(\Theta_{n,\tilde{\Phi}}\:|\:\mathcal{D}_n)}{\mu(\Theta_n\:|\:\mathcal{D}_n)}.
\end{align*}
Since $\mu(\Theta_n^c\:|\:\mathcal{D}_n) \xrightarrow{\mathbb{P}} 0$, the preceding expression reduces to \begin{align*}
    &   \frac{\int_{\Theta_n}  \exp \big(  -\frac{n}{2} \E_n  \big(   \|   \widehat{m}(W,h_t)   \|_{\widehat{\Sigma}(W)}^2    \big)    \big) d \mu (h)}{\int_{\Theta_n}  \exp \big(  -\frac{n}{2} \E_n  \big(   \|   \widehat{m}(W,h)   \|_{\widehat{\Sigma}(W)}^2    \big)    \big) d \mu (h)}    = e^{o(1)} \frac{\mu(\Theta_{n,\tilde{\Phi}}\:|\:\mathcal{D}_n)}{1+ o_{\mathbb{P}}(1)}.
\end{align*}
By replacing $D,M$ in the definition of $\Theta_n$ with a larger $D',M'$ if necessary, it is straightforward to verify that $\mu(\Theta_{n,\tilde{\Phi}}\:|\:\mathcal{D}_n) \xrightarrow{\mathbb{P}} 1$. From combining the preceding bounds, we obtain \begin{align} \label{laplace-bound}
   \nonumber  & \E^\star  \bigg[ \exp \bigg \{t  \sqrt{n} \big[ \langle  h - h_0 , \Phi       \rangle_{L^2(\mathbb{P})} + S_n  \big] \bigg \}   \:\bigg| \:\mathcal{D}_n   \bigg] \\ &  =[1+o_{\mathbb{P}}(1)] \exp \bigg[ \frac{t^2}{2}         \E \big[ (D_{h_0} \tilde{\Phi})' \Sigma_0(W) (D_{h_0} \tilde{\Phi})      \big]          \bigg].
\end{align}
In particular, we have that $$ \E^\star  \bigg[ \exp \bigg \{t  \sqrt{n} \big[ \langle  h - h_0 , \Phi       \rangle_{L^2(\mathbb{P})} + S_n  \big] \bigg \}   \:\bigg| \:\mathcal{D}_n   \bigg] \xrightarrow{\mathbb{P}} \exp \bigg[ \frac{t^2}{2}         \E \big[ (D_{h_0} \tilde{\Phi})' \Sigma_0(W) (D_{h_0} \tilde{\Phi})      \big]          \bigg]. $$

Since this is  true for every fixed $t \in \R$, it follows from \citep[Lemma 1]{castillo2015bernstein} that \begin{align} \label{weak-conv-1}
     \sqrt{n} \big( S_n + \langle  h - h_0 , \Phi       \rangle_{L^2(\mathbb{P})}   \big) \:\big| \:\mathcal{D}_n \: \:\overset{\mathbb{P}}{\rightsquigarrow} \: \: N \big( 0 ,  \E \big[ (D_{h_0} \tilde{\Phi})' \Sigma_0 (D_{h_0} \tilde{\Phi})      \big]             \big) . 
\end{align}
\item[\textbf{$(viii)$}]
Recall that \begin{align*}   S_n =    \E_n \big[  \langle  \rho(Y,h_0(X)), D_{h_0}[\tilde{\Phi} ](W)        \rangle_{\Sigma_0(W)}            \big].
\end{align*}
Since $S_n$ is the sample mean of a mean zero  random variable with finite variance, we have $ n \E[ S_n^2] = O(1)$. From (\ref{weak-conv-1}) and Lemma \ref{posmean-conv}, we can deduce (using a uniform integrability in probability argument) that: \begin{align*}
     \langle \E \big[  h\:|\:\mathcal{D}_n  \big] , \Phi \rangle_{L^2(\mathbb{P})} =  \langle h_0 , \Phi \rangle_{L^2(\mathbb{P})} - S_n + o_{\mathbb{P}}(n^{-1/2}).
\end{align*}
The first implication of this is that by substituting this identity back into (\ref{weak-conv-1}), we obtain \begin{align*}
\sqrt{n} \langle  h - \E \big[  h\:|\:\mathcal{D}_n  \big] , \Phi       \rangle\: \big| \: \mathcal{D}_n  \;  \overset{\mathbb{P}}{\rightsquigarrow} \;  N \big( 0 ,  \E \big[ (D_{h_0} \tilde{\Phi})' \Sigma_0 (D_{h_0} \tilde{\Phi})      \big]             \big) .  
\end{align*}
The second implication is that  $  \sqrt{n}\langle \E \big[  h\: | \:\mathcal{D}_n  \big] - h_0 , \Phi \rangle_{L^2(\mathbb{P})} $ is asymptotically equivalent to $ - \sqrt{n} S_n $. Hence, by the central limit theorem \begin{align*}
    \sqrt{n} \langle  h_0 - \E \big[  h\:|\:\mathcal{D}_n  \big] , \Phi       \rangle  =   \sqrt{n} S_n + o_{\mathbb{P}}(1)  \rightsquigarrow \; &   N(0, \E \big[ (D_{h_0} \tilde{\Phi})' \Sigma_0  \rho_{\star} \rho_{\star}'  \Sigma_0  (D_{h_0} \tilde{\Phi}) \big]   ) \;,
\end{align*}
where $\rho_{\star} = \rho(Y,h_0(X))$. The claim follows.
\end{enumerate}
\end{proof}
\begin{lemma}
\label{posmean-conv} Suppose the hypothesis of Theorem \ref{bvm} holds. Then \begin{align*} 
    n \E  \big[ \left|  \langle h - h_0 , \Phi    \rangle_{L^2(\mathbb{P})} \right|^2 \:\big| \: \mathcal{D}_n  \big]  = O_{\mathbb{P}}(1).
\end{align*}
\end{lemma}
\begin{proof}[Proof of Lemma \ref{posmean-conv}] 
Let $C$ denote a generic universal constant that may change from line to line. Define the sequences \begin{align}
\label{seqs}  & \epsilon_{n}^{} =    \frac{  \sqrt{ K}}{\sqrt{n}}     \; \; , \; \; \delta_{n} =  \begin{cases}     n^{- \frac{\alpha}{2[\alpha + \zeta] + d}} \sqrt{\log n}     & \text{mildly ill-posed}  \\  ( \log n)^{- \alpha/ \zeta} \sqrt{\log \log n}   & \text{severely ill-posed}.    \end{cases}  
\end{align}
First, we state a few preliminary observations from the proof of Theorem \ref{rate}. There exists a universal constant $c > 0$ such that \begin{align} \label{lb-newproof}   \int  \exp\bigg(    - \frac{n}{2}  \E_n \big[    \widehat{m}(W,h) ' \widehat{\Sigma}(W)  \widehat{m}(W,h)            \big]       \bigg) d  \mu (h)  \geq   \exp \big(   - c n    \epsilon_{n}^2   \big)   \end{align}
holds with $\mathbb{P}$ probability approaching $1$. Furthermore, for every $E' > 0$, there exists a sufficiently large $E$ (which depends on $E'$) such that \begin{align}
    \label{elarge} \mu \big( \| h- h_0 \|_{L^2(\mathbb{P})} \leq E \delta_n  \: \big| \: \mathcal{D}_n  \big) \geq 1- \exp(-E' n  \epsilon_n^2)
\end{align}
holds with $\mathbb{P}$ probability approaching $1$. Fix any $E ' > c$ and let $E$ be as specified above. Write  \begin{align*}
  & \E  \bigg[ \left|  \langle h - h_0 , \Phi    \rangle_{L^2(\mathbb{P})} \right|^2 \bigg| \: \mathcal{D}_n  \bigg]  \\ & = \E  \bigg[ \left|  \langle h - h_0 , \Phi    \rangle_{L^2(\mathbb{P})} \right|^2 \mathbbm{1} \{  \| h - h_0  \|_{L^2(\mathbb{P})} \leq E \delta_n  \} \bigg| \;\mathcal{D}_n  \bigg]      \\ & + \E  \bigg[ \left|  \langle h - h_0 , \Phi    \rangle_{L^2(\mathbb{P})} \right|^2 \mathbbm{1} \{  \| h - h_0  \|_{L^2(\mathbb{P})} > E \delta_n  \} \bigg| \;\mathcal{D}_n  \bigg]   \\ & = A_1  + A_2.
\end{align*}
For $A_2$, Cauchy-Schwarz yields \begin{align*}
    A_2^2 \leq \bigg ( \E  \bigg[ \left|  \langle h - h_0 , \Phi    \rangle_{L^2(\mathbb{P})} \right|^4 \;  \bigg| \;\mathcal{D}_n  \bigg]   \bigg )  \times  \mu \big( \| h- h_0 \|_{L^2(\mathbb{P})} > E \delta_n   \;\big| \; \mathcal{D}_n  \big)   .
\end{align*}
From $(\ref{lb-newproof})$, we obtain \begin{align*}
    & \E  \bigg[ \left|  \langle h - h_0 , \Phi    \rangle_{L^2(\mathbb{P})} \right|^4  \; \bigg| \;\mathcal{D}_n  \bigg]  \\ &  = \frac{\int  \left|  \langle h - h_0 , \Phi    \rangle_{L^2(\mathbb{P})} \right|^4    \exp\bigg(    - \frac{n}{2}  \E_n \big[    \widehat{m}(W,h) ' \widehat{\Sigma}(W)  \widehat{m}(W,h)            \big]       \bigg) d  \mu (h) }{\int \exp\bigg(    - \frac{n}{2}  \E_n \big[    \widehat{m}(W,h) ' \widehat{\Sigma}(W)  \widehat{m}(W,h)            \big]       \bigg) d  \mu (h)  } \\ & \leq \exp(c n  \epsilon_n^2) \int  \left|  \langle h - h_0 , \Phi    \rangle_{L^2(\mathbb{P})} \right|^4    \exp\bigg(    - \frac{n}{2}  \E_n \big[    \widehat{m}(W,h) ' \widehat{\Sigma}(W)  \widehat{m}(W,h)            \big]       \bigg) d  \mu (h) \\ & \leq \exp(c n  \epsilon_n^2) \|  \Phi \|_{L^2(\mathbb{P})}^4  \int  \| h -h_0 \|_{L^2(\mathbb{P})}^4  d  \mu (h) \\ & \leq C \exp(cn  \epsilon_n^2).
\end{align*}
Hence, by (\ref{elarge}) it follows that $
    A_2^2 \leq C \exp(  (c- E')     n  \epsilon_n^2  ) .
$
Since $E' > c$, we obtain $n A_2 = o_{\mathbb{P}}(1) $.

Let $\Theta_n$ be defined as in the proof of Theorem \ref{bvm}. From part $(i)$ of the proof of Theorem \ref{bvm}, we have $\mu(\Theta_n^c\:|\:\mathcal{D}_n)  \leq R'e^{-R n \epsilon_n^2 } $ with $\mathbb{P}$ probability approaching $1$, for some universal constant $R ,R'> 0$. We denote by $\E^\star(.\:|\:\mathcal{D}_n)$, the expectation with respect to the localized (to $\Theta_n$) posterior measure \begin{align*}  \mu^{\star}(A\:|\:\mathcal{D}_n)  = \frac{  \int_{A \cap \Theta_n} \exp\big(    - \frac{n}{2}  \E_n \big[ \|  \widehat{m}(W,h)  \|_{\widehat{\Sigma}(W)}^2    \big]      \big) d \mu(h)  }{\int_{\Theta_n} \exp\big(    - \frac{n}{2} \E_n \big[ \|  \widehat{m}(W,h)  \|_{\widehat{\Sigma}(W)}^2    \big]      \big) d \mu(h) } \; \; \; \; \; \; \; \forall \;  \text{Borel A}.
\end{align*}
 Under this setting, it follows that $A_1$ can be expressed as \begin{align*}
    A_1 & =  \E^\star  \bigg[ \left|  \langle h - h_0 , \Phi    \rangle_{L^2(\mathbb{P})} \right|^2 \mathbbm{1} \{  \| h - h_0  \|_{L^2(\mathbb{P})} \leq E \delta_n  \} \:\bigg|  \:\mathcal{D}_n  \bigg]  \\ & +  \int  \left|  \langle h - h_0 , \Phi    \rangle_{L^2(\mathbb{P})} \right|^2 \mathbbm{1} \{  \| h - h_0  
 \|_{L^2(\mathbb{P})} \leq E \delta_n  \} d \big[ \mu(h\;|\;\mathcal{D}_n) - \mu^{\star}(h\;|\;\mathcal{D}_n)   ] \\ & = A_{1,1} + A_{1,2}.
\end{align*}
From the general bound $x^2 \leq 2 \big( e^x + e^{-x} ) $ for every $x \in \R$, it follows from (\ref{laplace-bound}) with $t= \pm 1$ that \begin{align*}
    n A_{1,1} \leq C \big( e^{\sqrt{n} S_n} + e^{- \sqrt{n} S_n}      \big) \;,
\end{align*}
with $\mathbb{P}$ probability approaching $1$, where $S_n$ is defined as in (\ref{smean}). Since $S_n$ is a sample mean of a mean zero random variable with finite variance, the central limit theorem implies $nA_{1,1} = O_{\mathbb{P}}(1)$.

For $A_{1,2}$, if $\| . \|_{TV}$ denotes the total variation metric, we have that \begin{align*}
    A_{1,2} \leq E^2 \delta_n^2 \| \Phi  \|_{L^2(\mathbb{P})}^2 \| \mu - \mu^\star   \|_{TV}  \leq E^2 \delta_n^2 2 \mu(\Theta_n^c\:|\:\mathcal{D}_n) \leq  C \delta_n^2 e^{-R n \epsilon_n^2 }.
\end{align*}
It follows that $n A_{1,2} = o_{\mathbb{P}}(1)$.

\end{proof}

\begin{proof}[Proof of Corollary \ref{posgp}]
Let $\delta_n$ denote the stated contraction rate and $\epsilon_n = \sqrt{K_n} / \sqrt{n}$. From the proof of Theorem \ref{rate}, there exists a universal constant $D > 0$ such that for all sufficiently large $L > 0$, we have
\begin{align*}
    \mu \big(  \| h -h_0 \|_{L^2} > L \delta_n  \mid \mathcal{D}_n     \big) \leq \exp(-D L  n \epsilon_n^2).
\end{align*}
with $\mathbb{P}$ probability approaching $1$. Fix any $\overline{L}$ such that the preceding bound holds for all $ L \geq \overline{L} > 0$. Then, we have that \begin{align*}
     & \|  h_0 -  \E \big[ h\mid\mathcal{D}_n \big]  \|_{L^2}^2 \\ &   \leq \E \big(  \| h - h_0  \|_{L^2}^2 \mid \mathcal{D}_n        ) \\ & = \int_{\|  h - h_0  \|_{L^2} < \overline{L} \delta_n }  \| h - h_0  \|_{L^2}^2 d \mu(h \mid \mathcal{D}_n)   + \sum_{j=1}^{\infty}  \int \limits_{ j \overline{L} \delta_n \leq  \|  h - h_0  \|_{L^2}  < (j+1) \overline{L} \delta_n }  \| h - h_0  \|_{L^2}^2 d \mu(h \mid \mathcal{D}_n)  \\ & \leq \overline{L}^2 \delta_n^2 +   \overline{L}^2 \delta_n^2 \sum_{j=1}^{\infty} (j+1)^2 \exp(-D j \overline{L}  n \epsilon_n^2   ) .
\end{align*}
Since the preceding sum is finite, the claim follows.

\end{proof}

\begin{proof}[Proof of Corollary \ref{bvm-col}]
The set $C_n(\gamma)$ can equivalently be expressed as \begin{align*}
   & C_n(\gamma) = \{  t \in \R : \sqrt{n} \left|  t - \mathbf{L}\big( \E [h \mid  \mathcal{D}_n  ] \big)    \right|  \leq c_{1- \gamma}   \}\:, \\ & c_{1- \gamma} = (1- \gamma) \; \:  \text{quantile of} \:   \;  \sqrt{n} \left| \mathbf{L}(h) - \mathbf{L}\big( \E [h \mid \mathcal{D}_n ]   \big)  \right| \;  ,  \; h \sim \mu(\cdot \mid \mathcal{D}_n).
\end{align*}
Define \begin{align*}
    \sigma_{\Phi}^2 =  \E \big[ (D_{h_0} \tilde{\Phi} )' \{ \E[ \rho(Y,h_0(X)) \rho(Y,h_0(X))'|W  ] \}^{-1} (D_{h_0} \tilde{\Phi} )     \big]      .
\end{align*}

By Theorem \ref{bvm}$(i)$, we have \begin{align}
 \label{limvarc}   c_{1-\gamma} \xrightarrow{\mathbb{P}} (1- \gamma) \; \:  \text{quantile of} \:   \;  \left|Z \right|  \; \;, \; \; Z \sim N(0, \sigma_{\Phi}^2 ).
\end{align}
By Theorem \ref{bvm}$(ii)$, the distribution of $ \sqrt{n} \big( \mathbf{L}(h_0) - \mathbf{L}\big( \E [h \mid \mathcal{D}_n ] \big) \big)$ is asymptotically Gaussian with variance $\sigma_{\Phi}^2$. From this observation and (\ref{limvarc}), it follows that the frequentist coverage of $C_n(\gamma)$ is given by
 \begin{align*}
    \mathbb{P} \big(  \sqrt{n} \left|  \mathbf{L}(h_0) - \mathbf{L}\big( \E [h \mid \mathcal{D}_n ]  \big)   \right|  \leq c_{1- \gamma}      \big) = 1- \gamma + o_{\mathbb{P}}(1).
\end{align*}

\end{proof}

\end{document}